\documentclass[onecolumn,letter]{IEEEtran}
\usepackage{latexsym}
\usepackage{amssymb}
\usepackage{amsmath}

\usepackage{multirow}
\usepackage{bm}
\usepackage[dvips]{graphicx}

\newtheorem{thm}    {Theorem}
\newtheorem{lem}     {Lemma}
\newtheorem{cor}  {Corollary}
\newtheorem{proposition}        {Proposition}
\newtheorem{rem}     {Remark}
\newtheorem{condition}  {Condition}

\newcommand{\sX}{\mathsf{X}}
\newcommand{\sZ}{\mathsf{Z}}
\newcommand{\sW}{\mathsf{W}}

\def\mix{\mathop{\rm mix}}

\def\Add{\mathop{\rm Add}\nolimits}

\newcommand{\bR}{\mathbb{R}}

\newcommand{\bF}{\mathbb{F}}
\newcommand{\bbZ}{\mathbb{Z}}
\newcommand{\cH}{{\cal H}}

\def\sM{\mathsf{M}}
\def\sL{\mathsf{L}}

\def\cA{{\cal A}}

\def\cY{{\cal Y}}

\def\cM{{\cal M}}
\def\cD{{\cal D}}

\def\rE{{\rm E}}
\def\rP{{\rm P}}

\newcommand{\cX}{{\cal X}}

\newcommand{\Tr}{{\rm Tr}\,}

\newcommand{\bZ}{{\bf Z}}

\newcommand{\bY}{{\bf Y}}

\newcommand{\bg}{{\bf g}}

\makeatletter
\newcommand{\lleq}{\mathrel{\mathpalette\gl@align<}}
\newcommand{\ggeq}{\mathrel{\mathpalette\gl@align>}}
\newcommand{\gl@align}[2]{
\vbox{\baselineskip\z@skip\lineskip\z@
\ialign{$\m@th#1\hfil##\hfil$\crcr#2\crcr{}_{{}_{(=)}}\crcr}}}
\makeatother

\def\Label#1{\label{#1}\ [\ \text{#1}\ ]\ }
\def\Label{\label}

 \newenvironment{proof}{\par \noindent
            {\bf Proof. \hspace{2mm}}}{\hfill$\Box$ \vspace*{3mm}}

 \newenvironment{proofof}[1]{\vspace*{5mm} \par \noindent
         \quad{\it Proof of #1:\hspace{2mm}}}{\endproof
}

\begin{document}

\title{Quantum wiretap channel with non-uniform random number 
and 
its exponent and equivocation rate of leaked information}
\author{
Masahito Hayashi
\thanks{
M. Hayashi is with Graduate School of Mathematics, Nagoya University, 
Furocho, Chikusaku, Nagoya, 464-860, Japan, and
Centre for Quantum Technologies, National University of Singapore, 3 Science Drive 2, Singapore 117542.
(e-mail: masahito@math.nagoya-u.ac.jp)
This paper was
presented in part at 2012 IEEE International Symposium on Information Theory, 
Cambridge, MA, USA, July 2012.}%
}
\date{}
\maketitle

\begin{abstract}
A usual code for quantum wiretap channel requires 
an auxiliary random variable subject to the perfect uniform distribution.
However, it is difficult to prepare such an auxiliary random variable.
We propose a code that requires only an auxiliary random variable subject to a non-uniform distribution instead of the perfect uniform distribution.
Further, we evaluate the exponential decreasing rate of 
leaked information and derive its equivocation rate.
For practical constructions,
we also discuss the security when our code consists of a linear error correcting code.
\end{abstract}

\maketitle

\section{Introduction}
Secure communication is one of important topics in quantum information.
Quantum wiretap channel model is one of most adequate formulations of
this problem.
In this model,
the authorized sender Alice sends her message to the authorized receiver Bob and keeps its secrecy for the eavesdropper Eve.
The sender is often assumed to be able to use 
an auxiliary random variable subject to the perfect uniform distribution.
This assumption is usual in the quantum setting \cite{Deve,WNI} as well as in 
the classical setting\cite{Wyner,CK79,Csiszar,H-leaked}.
However, in some case, it is not necessarily easy to prepare a perfect uniform random number
with low cost.
In order to resolve it, it is needed to construct a code that properly works with 
an auxiliary non-uniform random number.

Further, 
the quantum wiretap channel capacity does not give an estimate 
for the leaked information with a finite-length code
because the capacity is the asymptotic limit of securely transmitted rate. 
Hence, 
from a practical viewpoint, 
we need to treat the decreasing speed of leaked information 
rather than the quantum wiretap channel capacity.
For this purpose, 
it is usual to focus on 
the exponential decreasing rate (exponent) of leaked information
with a given sacrifice information rate in information theory.
In this setting, we fix the sacrifice information rate,
and treat the speed of convergence rate of leaked information
by deriving its proper upper bounds. 
In fact,
as another setting,
similar to the second order analysis for classical channel coding \cite{Hay1,Pol},
given a fixed leaked information quantity,
we can treat the asymptotic rate of coding length up to the second order.  
This setting essentially employs the central limit theorem.
However, when the fixed leaked information quantity is quite small, 
the convergence of the central limit theorem is slow.
Hence, the latter setting cannot yield a good bound for a finite-length code in such a case.
Since the former setting gives upper bounds for leaked information,
the former method can provide useful upper bounds for leaked information
for a finite-length code.

In fact, Devetak \cite{Deve} applied the random coding to the quantum wiretap channel model.
For the evaluation for leaked information,
he essentially used the quantum version of the channel resolvability.
The original classical version is invented by Han and Verd\'{u}\cite{han-verdu}.
In the quantum version, we approximate the given output state by 
the output state of the input mixture 
of as a small number of input states as possible.
Devetak \cite{Deve}
treated the output approximation 
when the input distribution is the perfect uniform distribution.
Hence, his method does not work when 
the perfect uniform distribution is not available.
After his achievement, the book \cite{Hayashi-book} derived 
the exponent of leaked information, however, 
it also assumes that the perfect uniform distribution 
is available because its method is based on the quantum version of the channel resolvability.
Further, there is a possibility to improve the exponent in \cite{Hayashi-book}
because the commutative case of the exponent in \cite{Hayashi-book}
is smaller than that by the previous result \cite{H-leaked}.
On the other hand,
the previous paper 
\cite{SMC-non}
treated secure multiplex coding in the classical setting,
which requires the generalization of the channel resolvability 
to the case when the uniform random number is not available.

In this paper, 
we treat the quantum wiretap channel model 
when 
the sender Alice cannot use the uniform random number.
Alternatively, we assume that she knows 
the concrete form of the distribution of the auxiliary random number
and 
its R\'{e}nyi entropy of order 2,
which can be regarded as the sacrifice information rate.
In the quantum wiretap channel model, 
two classical-quantum channels are given. 
One is the channel from Alice's classical information to Bob's quantum state.
The other is the channel from Alice's classical information to Eve's quantum state.
In this model, we focus on the code given by random coding method.
Under this protocol, we give an upper bound for leaked information with the quantum mutual information criterion.
We also derive an upper bound for 
leaked information in terms of $L_1$ distinguishability.
The both bounds go to zero exponentially when the generated key rate is less than the capacity.
The exponent for the former
can be regarded as the quantum version of \cite{H-leaked}
because its commutative case coincides with that by \cite{H-leaked}.
The exponent for the latter 
is smaller than that by \cite{H-tight}.

When the generated key rate is larger than the capacity,
the leaked information does not go to zero.
In this paper, we derive the minimum leaked information rate.
That is, we calculate the maximum conditional entropy\cite{Wyner}.
In the degraded case, we obtain its single-letterized formula.
Further, 
in order to treat a more practical setting,
we derive similar bounds when our error correction codes are restricted to linear codes.
Finally, as a typical example,
we treat the Pauli channel.
In the classical case, more deeper analyses are given in \cite{SMC-non}.

This paper is organized as follows.
In Section \ref{s22}, we prepare quantum versions of information quantities and several fundamental inequalities
for latter discussion.
In Section \ref{sV}, 
we treat a non-uniform extension of quantum channel resolvability,
which is a strong tool for quantum wiretap channel with 
an auxiliary non-uniform random number.
In Section \ref{s2}, 
we proceed to the quantum wiretap channel model
and derive a lower bound of the exponent of leaked information,
whose commutative version coincides with that by the previous paper \cite{H-leaked}.
In Section \ref{s2-1}, 
we treat the case when the sacrifice information rate is less than
Eve's mutual information.
Then, we derive the equivocation rate.
In Section \ref{s3}, 
we treat the case when only liner codes are available.
Finally, in Section \ref{s10}, 
we discuss the case of Pauli channels.
In Appendix \ref{s4-1}, 
we review a result concerning privacy amplification by universal$_2$
hash functions,
which is shown in \cite{H-cq}.

\section{Information quantities}\Label{s22}
\subsection{Notations for distributions and states}\Label{s22-1}
In this paper, we denote the classical probability space by the 
calligraphic capital letter (e.g., ${\cal X}$) and the corresponding random variable
by the capital letter (e.g., $X$).
We also denote the probability distribution on ${\cal X}$ by $P_X$.
In order to describe the transition matrix from ${\cal V}$ to ${\cal X}$, we often 
use the letter $\Gamma$, in which,
$\Gamma_v$ denotes the probability distribution on ${\cal X}$ for any $v \in {\cal V}$.
Using these notations, we define 
a probability distribution on the composite system ${\cal X} \times {\cal V}$
by $(\Gamma \times P_V) (x,v):= \Gamma_v(x) P_V(x)$,
and 
a probability distribution on the system ${\cal X}$
by $(\Gamma \circ P_V) (x):= \sum_{v \in {\cal V}} \Gamma_v(x) P_V(x)$.
Since any function $f$ from ${\cal V}$ to ${\cal X}$ can be regarded as 
a transition matrix from ${\cal V}$ to ${\cal X}$,
we can define the probability distribution 
$f \times P_V$ on the composite system ${\cal X} \times {\cal V}$ and 
the probability distribution $f \circ P_V$ on the system ${\cal X}$ in the above way.
When two distributions $P_X$ and $P_V$ are given,
$P_X\times P_V$ expresses their independent product distribution 
on ${\cal X} \times {\cal V}$.
For a given subset $\Omega$ of ${\cal X}$, 
we denote the uniform distribution on $\Omega$ by $P_{\mix,\Omega}$.
For positive integers $\sM$ and ${\sL}$,
we denote the sets $\{1, \ldots, \sM\}$ and $\{1, \ldots, {\sL}\}$
by ${\cal M}$ and ${\cal L}$, respectively.

Next, we introduce notations by using 
a classical-quantum channel $W_E:x \mapsto W_{E|x}$ from the classical system ${\cal X}$ to the quantum system ${\cal H}_{E}$. 
When we consider only one quantum system, we simplify it as $W:x \mapsto W_{x}$. 
For a simple treatment, we identify the state $\sum_{x \in {\cal X}}P_X(x)|x\rangle \langle x|$
with the distribution $P_X$.
We define the state 
$W \times P_X:=\sum_{x} W_x \otimes P_X(x) |x \rangle \langle x| $ on 
the composite system ${\cal H}_{E} \otimes {\cal X}$,
and the state
$W \circ P_X:=\sum_{x}P_X(x) W_x$ on the system ${\cal H}_{E}$,
in which,  ${\cal H}_X$ is a classical system spanned by the basis $\{|x\rangle\}$.
For a transition matrix $\Gamma$ from ${\cal V}$ to ${\cal X}$, we 
define a classical-quantum channel
$W \circ \Gamma:v \mapsto \sum_{x} W_x \Gamma_v(x)$
from the classical system ${\cal V}$ to the quantum system ${\cal H}_{E}$. 
Since the state $(W \circ \Gamma) \circ P_V$ coincides with $W \circ (\Gamma \circ P_V)$
as a state on ${\cal H}_{E}$,
we simply denote it by $W \circ \Gamma \circ P_V$.

\subsection{Single system}\Label{s22-2}
Given a normalized state $\rho$ and a non-negative operator $\sigma$
on a single quantum system $\cH$, we prepare 
relative entropy type of information quantities 
\begin{align}
D(\rho\|\sigma) &:= \Tr \rho (\log \rho-\log \sigma) \\
\underline{D}(\rho\|\sigma) &:= \Tr \rho \log (\sigma^{-1/2} \rho \sigma^{-1/2}) .
\end{align}
Similarly, for $s\in (-1,\infty)$, we can define the functions
\begin{align}
D_s(\rho\|\sigma) & := \frac{1}{s}\log \Tr \rho^{1+s} \sigma^{-s} \\
\underline{D}^*_s(\rho\|\sigma) & := \frac{1}{s}\log \Tr \rho (\sigma^{-1/2}\rho \sigma^{-1/2})^{s} .
\end{align}
Indeed, 
the quantity ${D}_s(\rho\|\sigma)$ equals
the quantity $\frac{1}{s}{\psi}(s|\rho\|\sigma)$ given in 
\cite{Ogawa-Nagaoka,Hayashi-book,H-cq}.
While 
the quantity $\underline{D}^*_1(\rho\|\sigma) $
is the same as the quantity $\underline{\psi}(1|\rho\|\sigma)$ given in \cite{H-cq},
the quantity $\underline{D}^*_s(\rho\|\sigma)$
is different from 
the quantity $\frac{1}{s}\underline{\psi}(s|\rho\|\sigma)$ given in \cite{H-cq}.
Hence, we use the notation $\underline{D}^*_s(\rho\|\sigma)$ instead of $\underline{D}_s(\rho\|\sigma)$.

Since
$s\mapsto s D_s(\rho\|\sigma)$ and
$s\mapsto s \underline{D}^*_s(\rho\|\sigma)$ 
are convex
and
$\lim_{s \to 0}s D_s(\rho\|\sigma)=\lim_{s \to 0}s \underline{D}^*_s(\rho\|\sigma)=0$,
$D_s(\rho\|\sigma)$ and $\underline{D}^*_s(\rho\|\sigma)$
are monotone increasing for $s\in \bR$.
(The convexity of $s\mapsto s D_s(\rho\|\sigma)$ is shown in \cite{Hayashi-book}.
We can show the convexity of $s\mapsto s \underline{D}^*_s(\rho\|\sigma)$ in the same way
by calculating the second derivative.) 
Further,
since $\lim_{s \to 0}D_s(\rho\|\sigma)=D(\rho\|\sigma)$
and
$\lim_{s \to 0}\underline{D}^*_s(\rho\|\sigma)=\underline{D}(\rho\|\sigma)$,
\begin{align}
D(\rho\|\sigma) & \le D_s(\rho\|\sigma), \\
\underline{D}(\rho\|\sigma) & \le \underline{D}^*_s(\rho\|\sigma) 
\Label{8-15-14-b} 
\end{align}
for $s \in (0,1]$.
The information processing inequalities
\begin{align}
D({\cal E}(\rho)\|{\cal E}(\sigma)) & \le D(\rho\|\sigma) , \quad
D_s({\cal E}(\rho)\|{\cal E}(\sigma)) \le D_s(\rho\|\sigma)
\Label{8-21-7} 
\end{align}
hold for $s\in (-1,1]$ and a TP-CP map ${\cal E}$ \cite[(5.30), (5.41)]{Hayashi-book}.
However, this kind of inequality does not fold for $\underline{D}(\rho\|\sigma)$ or 
$\underline{D}^*_s(\rho\|\sigma)$ in general.
Since the inequalities (\ref{8-21-7}) are natural property
for information quantities,
we consider that
the quantities $D(\rho\|\sigma)$ and $D_s(\rho\|\sigma)$
describe
the essential information quantities,
and 
the quantities $\underline{D}(\rho\|\sigma)$ 
and $\underline{D}^*_s(\rho\|\sigma)$ 
are technical tools for our derivation. 
As is shown in the end of this subsection, the relations
\begin{align}
\underline{D}(\rho\|\sigma) & \le D(\rho\|\sigma) \Label{8-26-1}\\
\underline{D}^*_s(\rho\|\sigma) & \le D_s(\rho\|\sigma) \Label{8-26-2} 
\end{align}
hold for $s\in [-1,1]$.
When $\sigma $ is $I$, we obtain von Neumann entropy and R\'{e}nyi entropy of $\rho$ as
\begin{align}
H(\rho)&:=-D(\rho\|I)= -\underline{D}(\rho\|I)\\
H_{1+s}(\rho)&:=-D_s(\rho\|I)= -\underline{D}^*_s(\rho\|I).
\end{align}

When the state $\sigma$
has the spectral decomposition $\sigma= \sum_i s_i E_i$,
the pinching map ${\cal E}_\sigma$ is defined as
\begin{align}
{\cal E}_\sigma(\rho):=\sum_{i} E_i \rho E_i.
\end{align}
When $v$ is the number of the eigenvalues of $\sigma$,
the inequality 
\begin{align}
\rho \le v {\cal E}_{\sigma}(\rho).
\Label{8-15-23}
\end{align}
holds\cite[Lemma 3.8]{Hayashi-book},\cite{H2001}.
Hence, we obtain
\begin{align}
\sigma^{-1/2} \rho \sigma^{-1/2}
\le
v \sigma^{-1/2} {\cal E}_{\sigma}(\rho)  \sigma^{-1/2}. \Label{8-26-3}
\end{align}
Since $x \mapsto \log x$ is matrix monotone,
\begin{align}
\log \sigma^{-1/2} \rho \sigma^{-1/2}
\le
\log v+ 
\log \sigma^{-1/2} {\cal E}_{\sigma}(\rho)  \sigma^{-1/2} .
\end{align}
Since
\begin{align}
\Tr \rho \log \sigma^{-1/2} {\cal E}_{\sigma}(\rho)  \sigma^{-1/2}
=\Tr {\cal E}_{\sigma}(\rho) \log \sigma^{-1/2} {\cal E}_{\sigma}(\rho)  \sigma^{-1/2} ,
\end{align}
we obtain 
\begin{align}
D(\rho\|\sigma)
\le 
D({\cal E}_{\sigma}(\rho)\|\sigma) +\log v=
\underline{D}({\cal E}_{\sigma}(\rho)\|\sigma) +\log v .
\Label{8-15-8-a}
\end{align}

\quad {\it Proofs of (\ref{8-26-1}) and (\ref{8-26-2}):}\quad
Here, we show Inequalities (\ref{8-26-1}) and (\ref{8-26-2})
by using Inequality (\ref{8-26-3}).
Since the method based on Inequalities (\ref{8-15-23}) and (\ref{8-26-3})
is very important in this paper,
we put these proofs here not in Appendix.

For $s \in [0,1]$, 
since $x^s$ is operator monotone,
the inequality (\ref{8-26-3}) implies that
\begin{align*}
& \Tr \rho (\sigma^{-1/2}\rho\sigma^{-1/2})^s
\le
v \Tr \rho (\sigma^{-1/2}{\cal E}_{\sigma}(\rho)\sigma^{-1/2})^s \\
= &
v \Tr {\cal E}_{\sigma}(\rho) (\sigma^{-1/2}{\cal E}_{\sigma}(\rho)\sigma^{-1/2})^s \\
= &
v \Tr {\cal E}_{\sigma}(\rho)^{1+s} \sigma^{-s}
\le
v \Tr \rho^{1+s} \sigma^{-s},
\end{align*}
where $v$ is the number of eigenvalues of $\sigma$.
In the $n$-fold setting,
we obtain
\begin{align*}
e^{s\underline{D}^*_s(\rho^{\otimes n}\|\sigma^{\otimes n})}
\le
v_n 
e^{s D_s(\rho^{\otimes n}\|\sigma^{\otimes n} )},
\end{align*}
where $v_n$ is the number of eigenvalues of $\sigma^{\otimes n}$.
That is, the relation
\begin{align}
s\underline{D}^*_s(\rho \|\sigma )
\le
\frac{\log v_n}{n}+ s D_s(\rho\|\sigma)
\end{align}
holds.
Taking the limit $n \to \infty$, we obtain (\ref{8-26-1}) with $s \in [0,1]$.
Finally, taking the limit $s \to 0$, we obtain (\ref{8-26-2}).

For $s \in [-1,0]$, 
since $-x^s$ is operator monotone,
the inequality (\ref{8-26-3}) implies that
\begin{align*}
& \Tr \rho (\sigma^{-1/2}\rho\sigma^{-1/2})^s
\ge
v \Tr \rho (\sigma^{-1/2}{\cal E}_{\sigma}(\rho)\sigma^{-1/2})^s \\
= &
v \Tr {\cal E}_{\sigma}(\rho) (\sigma^{-1/2}{\cal E}_{\sigma}(\rho)\sigma^{-1/2})^s \\
= &
v \Tr {\cal E}_{\sigma}(\rho)^{1+s} \sigma^{-s}
\ge
v \Tr \rho^{1+s} \sigma^{-s},
\end{align*}
which implies that
\begin{align}
s\underline{D}^*_s(\rho \|\sigma )
\ge
\frac{\log v_n}{n}+ s D_s(\rho\|\sigma).
\end{align}
Taking the limit $n \to \infty$, we obtain (\ref{8-26-1}) with $s \in [-1,0]$.
\endproof

\subsection{Composite system}\Label{s22-3}
This paper heavily employs mutual information type quantities rather than conditional entropy type quantities
because the channel coding is closely linked to the mutual information rather than the conditional entropy even with the existence of the eavesdropper.
Now, we introduce mutual information type quantities 
in a composite system ${\cal H}_X \otimes {\cal H}_{E}$
by using the relative entropy type quantities
$D(\rho\|\sigma)$, $D_s(\rho\|\sigma)$,
$\underline{D}(\rho\|\sigma)$, and
$\underline{D}^*_s(\rho\|\sigma)$
when the composite state $\rho_{E,X}$ is 
given as $W \times P_X$ by using 
a distribution $P_X$ on ${\cal X}$ and
a classical-quantum channel $W:x \mapsto W_x$ from the classical system ${\cal X}$ to the quantum system ${\cal H}_{E}$.
Using two kinds of relative entropies
$D(\rho\|\sigma)$ and $\underline{D}(\rho\|\sigma)$,
we define two kinds of quantum versions of
the mutual information between ${\cal H}_X$ and ${\cal H}_{E}$ as
\begin{align}
I(X;E|\rho_{E,X}) &:= D( \rho_{E,X} \| \rho_{E} \otimes \rho_X) \\
\underline{I}(X;E|\rho_{E,X}) &:= \underline{D}( \rho_{E,X} \| \rho_{E} \otimes \rho_X) ,
\end{align}
where $\rho_E$, $\rho_X$ are reduced density operators
while the conditional entropy is given as
\begin{align}
H(X|E|\rho_{E,X}) &:= - D( \rho_{E,X} \| \rho_{E} \otimes I) .
\end{align}
In this notation,
we use the random variable for identifying the classical system while we use the subscript of the quantum system for identifying the quantum system.
Further, we define modified version mutual informations as 
\begin{align}
I_s(X;E|\rho_{E,X}) &:= D_s( \rho_{E,X} \| \rho_{E} \otimes \rho_X) \\
\underline{I}^*_s(X;E|\rho_{E,X}) &:= \underline{D}^*_s( \rho_{E,X} \| \rho_{E} \otimes \rho_X)
\end{align}
while 
the conditional R\'{e}nyi entropy is defined \cite{H-cq,H-precise} as
\begin{align}
H_{1+s}(X|E|\rho_{E,X}) &:= - D_s( \rho_{E,X} \| \rho_{E} \otimes I) .
\end{align}
Since $D_s(\rho\|\sigma)$ is monotone increasing for $s\in \bR$,
the function
$ {H}_{1+s}(X|E|\rho_{E,X})$ is monotone decreasing for $s\in \bR$.
In particular,
\begin{align}
H(X|E|\rho_{E,X}) & \ge H_{1+s}(X|E|\rho_{E,X})
\Label{8-15-14} 
\end{align}
for $s\in (0,1]$ because $\lim_{s \to 0}H_{1+s}(X|E|\rho_{E,X})=H(X|E|\rho_{E,X})$.
When a TP-CP map ${\cal E}$ on the system $\cH_E$ is applied, 
Relations (\ref{8-21-7}) imply that
\begin{align}
I(X;E|{\cal E}(\rho_{E,X})) \le I(X;E|\rho_{E,X}),
\quad
I_s(X;E|{\cal E}(\rho_{E,X})) \le I_s(X;E|\rho_{E,X})\Label{9-4-1}
\end{align}
for $s \in (-1,1]$.
However, $\underline{I}(X;E|\rho_{E,X})$
and $\underline{I}^*_s(X;E|\rho_{E,X})$
do not satisfy the same property.
In this paper, due to these properties,
we consider that
the quantities $I(X;E|{\cal E}(\rho_{E,X}))$
and $I_s(X;E|{\cal E}(\rho_{E,X}))$
describe essential information quantities.
The other quantities
$\underline{I}(X;E|\rho_{E,X})$
and $\underline{I}^*_s(X;E|\rho_{E,X})$
are technical tools for overcoming the difficulty 
caused by the non-commutativity.
That is, our final result will be described by using 
the quantities $I(X;E|{\cal E}(\rho_{E,X}))$
and $I_s(X;E|{\cal E}(\rho_{E,X}))$.

Next, using a state $\sigma_E$ on ${\cal H}_E$, 
we extend $I(X;E|\rho_{E,X})$ and $I_s(X;E|\rho_{E,X})$ as
\begin{align}
I(X;E|\rho_{E,X}\|\sigma_E) &:= D( \rho_{E,X} \| \sigma_{E} \otimes \rho_X) \\
I_s(X;E|\rho_{E,X}\|\sigma_E) &:= D_s( \rho_{E,X} \| \sigma_{E} \otimes \rho_X) .
\end{align}
The conditional R\'{e}nyi entropy $H(X|E|\rho_{E,X})$ is also generalized \cite{H-cq} as
\begin{align*}
H_{1+s}(X|E|\rho_{E,X}\|\sigma_E) &:=
- D_s( \rho_{E,X} \| \sigma_{E} \otimes I).
\end{align*}
Note that our definition of the quantity $H_{2}(X|E|\rho_{E,X}\|\sigma_E)$
is different from 
Renner's definition $H_{2}(X|E|\rho_{E,X}\|\sigma_E)= - \underline{D}_1^*( \rho_{E,X} \| \sigma_{E} \otimes I)$ \cite{Ren05} in the non-commutative case.
Generally, due to (\ref{8-26-2}), our definition gives a smaller value than Renner's definition.

The first quantity $I(X;E|\rho_{E,X}\|\sigma_E)$ satisfies
\begin{align}
I(X;E|\rho_{E,X}\|\sigma_E) - I(X;E|\rho_{E,X}) 
= D( \rho_E \| \sigma_{E}) \ge 0\Label{8-04-a}
\end{align}
for any state $\sigma_E$ on $\cH_E$.
That is, 
\begin{align}
\min_{\sigma_E} I(X;E|\rho_{E,X}\|\sigma_E) = I(X;E|\rho_{E,X}) .
\end{align}

For the minimization of the second quantity $I_s(X;E|\rho_{E,X}\|\sigma_E)$, 
we introduce another type of mutual information as
\begin{align}
I_s^{{\rm G}}(X;E|\rho_{E,X}) := 
\frac{1+s}{s}\log \Tr  \left( \sum_x P_X(x)
(W_x)^{1+s}\right)^{\frac{1}{1+s}} .
\end{align}
This quantity can be written by using 
a quantum extension of Gallager function\cite{Gal}
\begin{align*}
\phi(s|W,P_X) := 
\log \Tr  \left( \sum_x P_X(x)
(W_x)^{1/(1-s)}\right)^{1-s} 
\end{align*}
as
\begin{align*}
I_s^{{\rm G}}(X;E|\rho_{E,X}) = \frac{1+s}{s}\phi(\frac{s}{1+s}|W,P_X).
\end{align*}
These quantities satisfy
\begin{align}
\lim_{s\to 0} I_s^{{\rm G}}(X;E|\rho_{E,X})
=\lim_{s\to 0} I_s(X;E|\rho_{E,X})
=I (X;E|\rho_{E,X}),
\Label{9-5-c1}
\end{align}
and can be characterized by the following lemmas
\begin{lem}\Label{l3}
The equation
\begin{align}
\min_{\sigma_E} 
I_s(X;E|\rho_{E,X} \| \sigma_{E} )
=
I_s^{{\rm G}}(X;E|\rho_{E,X})
\Label{8-12-a}
\end{align}
holds for $s \in (-1,\infty)$,
where $\sigma$ is restricted to normalized states on ${\cal H}_{E}$.
The minimum is realized when\par
\noindent $\sigma_{E}= (\sum_x P_X(x) W_x^{1+s})^{\frac{1}{1+s}}/\Tr (\sum_x P_X(x) W_x^{1+s})^{\frac{1}{1+s}}$.
Thus, the inequality
\begin{align}
I_s(X;E|\rho_{E,X}) \ge I_s^{{\rm G}}(X;E|\rho_{E,X})
\Label{10-25-101}
\end{align}
holds for $s \in (-1,\infty)$.
\end{lem}
The proof of Lemma \ref{l3} is given in Appendix \ref{sl3}. 
Since $D_s(\rho\|\sigma)$ is monotone increasing for $s$,
$I_s(X;E|\rho_{E,X} \| \sigma_{E} )$ is also monotone increasing for $s$.
Due to Lemma \ref{l3}, 
$I_s^{{\rm G}}(X;E|\rho_{E,X})$ is also monotone increasing for $s$.

Conversely, as shown in \cite[(16)]{H-leaked}, the following lemma holds.
\begin{lem}\Label{l2}
When all of the states $W_x$ are commutative each other,
the inequality
\begin{align}
I_s(X;E|\rho_{E,X}) \le I_{\frac{s}{1-s}}^{{\rm G}}(X;E|\rho_{E,X})
\Label{8-23-4b}
\end{align}
holds for $s \in [0,1)$.
\end{lem}

\subsection{Classical-quantum channel}\Label{s22-4}
Next, we characterize the above defined quantities 
by using a classical-quantum channel $W$ from ${\cal X}$ to ${\cal H}_E$.
\begin{align}
I(X;E|W \times P_X) &= \sum_x P_X(x) \Tr W_x (\log W_x -\log (W \circ P_X) ) \\
\underline{I}(X;E|W \times P_X) 
&= \sum_x P_X(x) \Tr W_x \log ( (W \circ P_X)^{-1/2}W_x (W \circ P_X)^{-1/2}) \\
I_s(X;E|W \times P_X) 
&= \frac{1}{s}\log \Tr \sum_x P_X(x) \Tr W_x^{1+s} (W \circ P_X)^{-s} \\
\underline{I}^*_s(X;E|W \times P_X)
& = \frac{1}{s}\log \Tr \sum_x P_X(x) \Tr W_x ( (W \circ P_X)^{-1/2}W_x (W \circ P_X)^{-1/2})^{s}.
\end{align}
When we apply a quantum operation ${\cal E}$ on ${\cal H}_{E}$, 
we can define another channel ${\cal E}[W]:x \mapsto {\cal E}(W_x)$ from the classical system ${\cal X}$ to the quantum system ${\cal H}_{E}$.
Then, Inequalities (\ref{9-4-1}) are rewritten as
\begin{align}
I(X;E|{\cal E}[W] \times P_X)
&\le I(X;E|W \times P_X)
\Label{8-15-7}
\\
I_s(X;E|{\cal E}[W] \times P_X)
&\le I_s(X;E|W \times P_X)
\Label{8-15-7-b}
\end{align}
hold for $s\in (-1,1]$.
When $v$ is the number of the eigenvalues of $W \circ P_X$,
Relation (\ref{8-15-8-a}) yields an inequality opposite to (\ref{8-15-7}):
\begin{align}
I(X;E|W \times P_X) 
& \le 
I(X;E|{\cal E}_{W \circ P_X}[W] \times P_X ) +\log v \nonumber \\
&= \underline{I}(X;E|{\cal E}_{W \circ P_X}[W] \times P_X ) +\log v .
\Label{8-15-8} 
\end{align}

Further, we obtain the following lemma.
\begin{lem}\Label{l1}
The function $P_X \mapsto 
e^{\frac{s}{1+s}I_s^{{\rm G}}(X;E|W \times P_X )}
=e^{\phi(\frac{s}{1+s}|W,P_X)}$ 
is convex with $s\in [-\frac{1}{2},0]$, and is concave with $s\in [0,\infty]$.
\end{lem}
However, the quantity $P_X \mapsto e^{s I_s(X;E|{\cal E}[W] \times P_X)} $ does not satisfy the convexity in general.

\begin{proof}
The convexity and concavity of $P_X \mapsto 
e^{\phi(\frac{s}{1+s}|W,P_X)} = \Tr  \left( \sum_x P_X(x) (W_x)^{1+s}\right)^{\frac{1}{1+s}}$ 
follow from the operator convexity and operator concavity of
$x^{\frac{1}{1+s}}$ for the respective parameter $s$.
\end{proof}

\subsection{Security criteria}\Label{s22-5}
Next, in order to give the security criteria,
we characterize the leaked information when an information $X$ is transmitted via the classical-quantum channel $W$.
When the information $X$ is subjected to the uniform distribution $P_{\mix,\cX}$ on $\cX$,
the leaked information is given as
\begin{align*}
& I_{\mix}(X;E|W):= I(X;E|W\times P_{\mix,\cX})
= D(W\times P_{\mix,\cX} \| (W \circ P_{\mix,\cX}) \otimes P_{\mix,\cX}) \\
=& \sum_{x\in \cX} P_{\mix,\cX}(x) D(W_x \| W \circ P_{\mix,\cX})
=\sum_{x\in \cX} P_{\mix,\cX}(x) ( H(W \circ P_{\mix,\cX}) -  H(W_x) ).
\end{align*}
When we do not know the distribution of the information $X$,
we adopt the following value as the criterion of the leaked information:
\begin{align*}
& I_{\max}(X;E|W):= \max_{P_X} I(X;E|W\times P_{X})
=\max_{P_X} D(W\times P_{X} \| (W \circ P_{X}) \otimes P_{X}) \\
=& \max_{P_X} \sum_{x\in \cX} P_X(x) D(W_x \| W \circ P_X)
=\max_{P_X} \sum_{x\in \cX} P_{X}(x) ( H(W \circ P_{X}) -  H(W_x) ).
\end{align*}

Next, we consider the leaked information by using the trace norm instead of the mutual information.
When the state is given as 
$\rho_{E,X}=\sum_{x} P_X(x) W_x\otimes |x\rangle \langle x| $,
the leaked information is characterized as
\begin{align}
d_1(X;E|\rho_{E,X}):= \|\rho_{E,X}- \rho_E \otimes P_{X} \|_1.
\end{align}
When we take into account the uniformity as well as the independence,
we employ the following quantity \cite{Ren05}:
\begin{align}
d_1'(X;E|\rho_{E,X}):= \|\rho_{E,X}- \rho_E \otimes P_{\mix,{\cal X}} \|_1.
\end{align}
In this notation, 
when a function $f:{\cal X}\mapsto{\cal Y}$ is given,
$d_1'(f(X);E|\rho_{E,X})$ expresses the following quantity:
\begin{align}
d_1'(f(X);E|\rho_{E,X})
=\left\| \sum_{y} (\sum_{x \in f^{-1}(y)} P_X(x) W_x ) \otimes |y\rangle \langle y|
- \rho_E \otimes P_{\mix,{\cal Y}} \right\|_1.
\end{align}
Now, we consider the case when 
the eavesdropper's system consists of the quantum system $\cH_E$ and the classical system $\cY$.
As shown in Appendix \ref{a8-12-1}, when $X$ and $Y$ obey the uniform distribution independently,
this quantity satisfies
\begin{align}
d_1'(X;Y,E|\rho_{X,Y,E})
\le 2 d_1'(X,Y;E|\rho_{X,Y,E}).
\Label{8-09-o}
\end{align}

Using these quantities, we define the leaked information for the channel $W$.
When the input information $X$ is subjected to the uniform distribution $P_{\mix,\cX}$ on $\cX$,
the leaked information is given as
\begin{align*}
d_{1,\mix}(X;E|W):= 
d_1(X;E|W \times P_{\mix,\cX})
=
d_1'(X;E|W \times P_{\mix,\cX})
=
\sum_{x \in \cX} P_{\mix,\cX}(x) \| W_{x} - W \circ P_{\mix,\cX} \|_1 .
\end{align*}
When we do not know the distribution of the information $X$,
we adopt the following value as the criterion of the leaked information:
\begin{align*}
d_{1,\max}(X;E|W):= 
\max_{P_X} d_1(X;E|W \times P_X)
=
\max_{P_X} \sum_{x \in \cX} P_{X}(x) \| W_{x} - W \circ P_{X} \|_1 .
\end{align*}

In this paper,
we employ the quantities 
$I_{\mix}(X;E|W)$, $I_{\max}(X;E|W)$, $d_{1,\mix}(X;E|W)$, and $d_{1,\max}(X;E|W)$ 
as security criteria.
Since $I_{\max}(X;E|W)$ and $d_{1,\max}(X;E|W)$ 
do not depend on the distribution on the input messages,
the results based on these are called 
source universal \cite{korner80}.

Using the quantum version of Pinsker inequality,
we obtain
\begin{align}
d_{1,\mix}(X;E|W)^2 &\le 2 I_{\mix}(X;E|W) \Label{8-19-14-a} \\
d_{1,\max}(X;E|W)^2 &\le 2 I_{\max}(X;E|W) \Label{8-19-14-e}. 
\end{align}
Conversely, we can bound $I_{\mix}(W)$ and $I_{\max}(W)$ by using $d_{1,\mix}(W)$ and $d_{1,\max}(W)$
in the following way.
Applying Fannes's inequality, we obtain
\begin{align}
0 \le & I_{\mix}(X;E|W)  
=\sum_{x\in \cX} P_{\mix,\cX}(x)  (H(W \circ P_{\mix,\cX}) - H(W_x)) \nonumber \\
\le &  \sum_{x\in \cX} P_{\mix,\cX}(x) \eta (\|W \circ P_{\mix,\cX} - W_x \|_1, \log d_{E})  \nonumber \\
\le &  \eta (\sum_{x\in \cX} P_{\mix,\cX}(x) \|W \circ P_{\mix,\cX} - W_x \|_1, \log d_{E})  \nonumber \\
= &   \eta (d_{1,\mix}(X;E|W), \log d_{E})  ,
\Label{8-26-9-a} \\
0 \le & I_{\max}(X;E|W)  
=\max_{P_X}\sum_{x\in \cX} P_{X}(x)  (H(W \circ P_{X}) - H(W_x)) \nonumber \\
\le &  \max_{P_X}\sum_{x\in \cX} P_{X}(x) \eta (\|W \circ P_{X} - W_x \|_1, \log d_{E})  \nonumber \\
\le &  \eta (\max_{P_X} \sum_{x\in \cX} P_{X}(x) \|W \circ P_{X} - W_x \|_1, \log d_{E})  \nonumber \\
= &   \eta (d_{1,\max}(X;E|W), \log d_{E})  ,
\Label{8-26-9-e} 
\end{align}
where $d_{E}$ is the dimension of ${\cal H}_{E}$
and $\eta(x,y):= -x \log x +x y $.


Therefore,
if the quantity $I_{\mix}(X;E|W)$ or $I_{\max}(X;E|W)$ goes to zero,
the quantity $d_{1,\mix}(X;E|W)$ or $d_{1,\max}(X;E|W)$ goes to zero.
The converse is also true
when 
the quantity $d_{1,\mix}(X;E|W)$ or $d_{1,\max}(X;E|W)$ exponentially goes to zero and the dimension $d_{E}$ grows linearly.
However, their speeds of both convergence do not coincide with each other.
Hence, we consider both quantities.

Next, we consider the relation between 
$I_{\mix}(X;E|W)$ and $I_{\max}(X;E|W)$ 
in a special class of channels.
A channel $W$ from the set $\cX$ to the quantum system $\cH_E$ is called {\it additive}
when the set ${\cal X}$ has a structure of module and 
there exist a state $\rho$ on the system $\cH_E$ and a projective representation $U$ of ${\cal X}$ such that $W_{x}= U_x \rho U_x^\dagger$.
In this case, as shown in Appendix \ref{a9-4-c1}, the relation
\begin{align}
I_{\mix}(X;E|W)=I_{\max}(X;E|W)
\Label{9-4-c1}
\end{align}
holds.
This equation is useful in the latter discussions.

\section{Generalization of channel resolvability}\Label{sV}
In order to treat the quantum wiretap channel model, 
we treat the quantum channel resolvability problem for a given
classical-quantum channel $W$ from the classical system ${\cal X}$
to the quantum system ${\cal H}$,
in which,
the output quantum state $W_x$ on ${\cal H}$ is given for an element $x\in {\cal X}$.
In the channel resolvability,
we treat the approximation of the given output state 
with the output average state of the auxiliary input random variable
when 
the auxiliary input random variable is subject to 
the uniform distribution of the subset ${\cal X}_0$ of 
the input system ${\cal X}$ of the given channel. 
That is, the purpose is 
minimizing the cardinality of the input subset ${\cal X}_0$ 
when 
$W \circ P_{\mix,\cX_0}=\sum_{x\in {\cal X}_0}\frac{1}{|{\cal X}_0|}W_{x}$
approximates a given output state $\rho$.

Now, we generalize this problem to the case when
the uniform distribution on ${\cal X}_0$ is not available.
We assume that the auxiliary random variable is subject to 
a given distribution $P_A$ on the set ${\cal A}$.
Choosing a map $f$ from ${\cal A}$ to ${\cal X}$,
we approximate a given state $\rho$ by
$W \circ f \circ P_A =\sum_{a\in {\cal A}} P_A(a) W_{f(a)}$. 

Now, we apply the random coding to the alphabet $X$ with the probability distribution $P_X$.
The map $\Phi$ from ${\cal A}$ to ${\cal X}$ is
randomly chosen in the following way.
For each $a \in {\cal A}$, $\Phi(a)$ is the random variable subject to the distribution $P_X$ on ${\cal X}$.
For $a \neq a' \in {\cal A}$, 
$\Phi(a)$ is independent of $\Phi(a')$.
Then, for a distribution $P_A$ on ${\cal A}$,
we can define the distribution $\Phi \circ P_A(x)$ on ${\cal X}$.


Then, we have the following lemma:
\begin{lem}
\Label{Lee1}
For $s \in (0,1]$, we obtain
\begin{align}
\rE_{\Phi} D( W \circ \Phi \circ P_A \| W \circ P_X)
\le
\frac{v^s e^{-s H_{1+s}(P_A)+s I_s(X;E|W \times P_X)}}{s},
\Label{9-19-1}
\end{align}
where $v$ is the number of eigenvalues of $W \circ P_X$.
Similarly,
\begin{align}
& 
\rE_{\Phi} e^{s \underline{D}( W \circ \Phi \circ P_A \| W \circ P_X)}
\le
\rE_{\Phi} e^{s \underline{D}^*_s(W \circ \Phi \circ P_A \| W\circ P_X)} \nonumber \\
\le &
1+ e^{-s H_{1+s}(P_A)+s \underline{I}^*_s(X;E|W \times P_X)}.
\Label{9-19-2}
\end{align}
\end{lem}

Using (\ref{9-19-2}) in Lemma \ref{Lee1}, we obtain
\begin{align}
& 
\rE_{\Phi} e^{s D( W \circ \Phi \circ P_A \| W \circ P_X)}
\le
v^s \rE_{\Phi} e^{s D( {\cal E}_{W \circ {P_X}}(W \circ \Phi \circ P_A) \| W \circ {P_X})} \nonumber \\
\le &
v^s (1+ e^{-s H_{1+s}(P_A)}e^{s \underline{I}^*_s(X;E|{\cal E}_{W \circ P_X}[W] \times P_X)}) \nonumber \\
= &
v^s (1+ e^{-s H_{1+s}(P_A)}e^{s I_s (X;E|{\cal E}_{W \circ P_X}[W] \times P_X)}) \nonumber \\
\le &
v^s (1+ e^{-s H_{1+s}(P_A)}e^{s I_s(X;E| W \times P_X)}).
\Label{9-20-1}
\end{align}

Next, we consider 
this kind of approximation when the space ${\cal X}$ has the structure of a module.
That is, 
when a submodule $C\subset {\cal X}$ and an element $y \in {\cal X}$ are given,
using the uniform distribution $P_{\mix,C+y}$ on the subset $C+y:=\{x+y|x\in C\}$,
we approximate the output state $W \circ {P_{\mix,{\cal X}}}$ with the uniform input distribution.
In this case, we evaluate 
$D(W \circ {P_{\mix,C+y}}\| W \circ {P_{\mix,{\cal X}}}) $.
Now, we consider the condition for ensemble of submodules $\{C[\bZ]\}$ of ${\cal X}$ when the submodule $C[\bZ]$ is decided by a random variable $\bZ$.
\begin{condition}\Label{C2}
The cardinality of $C[\bZ]$ is a constant ${\sL}$.
Any element $x \neq 0 \in {\cal X}$
is included in $C[\bZ]$ with probability at most 
$\frac{{\sL}}{|{\cal X}|}$.
\end{condition}

Then, we have the following lemma:
\begin{lem}
\Label{Lee2}
When the random variable $\tilde{X}\in {\cal X}$ obeys the uniform distribution on ${\cal X}$
and is independent of the choice of $C[\bZ]$,
we obtain
\begin{align}
\rE_{\bZ,\tilde{X}} D( W \circ {P_{\mix,C[\bZ]+\tilde{X}}} \| W \circ {P_{\mix,{\cal X}}})
\le
\frac{v^s e^{s I_s (X;E|W \times P_{\mix,\cX})}}{{\sL}^s s}
\Label{9-19-1-a}
\end{align}
for $s \in (0,1]$, where $v$ is the number of eigenvalues of $W \circ {P_{\mix,{\cal X}}}$.
Similarly,
\begin{align}
& 
\rE_{\bZ,\tilde{X}} e^{s \underline{D}( W \circ {P_{\mix,C[\bZ]+\tilde{X}}} \| W \circ {P_{\mix,{\cal X}}})} \nonumber \\
\le &
\rE_{\bZ,\tilde{X}} e^{s \underline{D}^*_s( W \circ {P_{\mix,C[\bZ]+\tilde{X}}} \| W \circ {P_{\mix,{\cal X}}})} \nonumber \\
\le &
1+ \frac{1}{{\sL}^s}e^{s \underline{I}^*_s(X;E|W \times P_{\mix,{\cal X}})}.
\Label{9-19-2-a}
\end{align}
\end{lem}

Using (\ref{9-19-2-a}) in Lemma \ref{Lee2}, we obtain
\begin{align}
& 
\rE_{\bZ,\tilde{X}} e^{s D( W \circ {P_{\mix,C[\bZ]+\tilde{X}}} \| W \circ {P_{\mix,{\cal X}}})} \nonumber \\
\le &
v^s \rE_{\bZ,\tilde{X}} e^{s D( {\cal E}_{W \circ {P_{\mix,{\cal X}}}}(W \circ {P_{\mix,C[\bZ]+\tilde{X}}}) \| W \circ {P_{\mix,{\cal X}}})} \nonumber \\
\le &
v^s (1+ \frac{1}{{\sL}^s}e^{s\underline{I}^*_s(X;E|
{\cal E}_{W \circ {P_{\mix,{\cal X}}}}[W] \times P_{\mix,{\cal X}})}) \nonumber \\
= &
v^s (1+ \frac{1}{{\sL}^s} 
e^{s I_s(X;E|{\cal E}_{W \circ {P_{\mix,{\cal X}}}}[W] \times P_{\mix,{\cal X}})}) \nonumber \\
\le &
v^s (1+ \frac{1}{{\sL}^s} e^{s I_s (X;E|W \times P_{\mix,{\cal X}})}).
\Label{9-20-1-a}
\end{align}

Next, we consider the case when the uniform distribution on $C[\bZ]$ is not available
and only a distribution $P_A$ on another module ${\cal A}$ is available. 
Now, we assume the following condition for the ensemble for injective homomorphisms 
$f_{\bZ}$
from ${\cal A}$ to ${\cal X}$.
\begin{condition}\Label{C2-b}
Any elements 
$x \neq 0 \in {\cal X}$ and $a \neq 0 \in {\cal C}$,
the relation $f_{\bZ}(a)=x$ holds with probability at most 
$\frac{1}{|{\cal X}|-1}$.
\end{condition}
When ${\cal X}$ and ${\cal A}$ are vector spaces of a finite field $\bF_q$,
the set of all injective morphisms from ${\cal A}$ to ${\cal X}$ satisfies Condition \ref{C2-b}.

We choose the random variable $\tilde{X}\in {\cal X}$ that obeys the uniform distribution on ${\cal X}$ 
that is independent of the choice of $f_{\bZ}$.
Then, we define a map $f_{\bZ|\tilde{X}}(a):=f_{\bZ}(a)+\tilde{X}$
and have the following lemma:

\begin{lem}
\Label{Lee3}
Under the above choice,
we obtain
\begin{align}
\rE_{\bZ,\tilde{X}} D( W \circ f_{\bZ|\tilde{X}} \circ P_A \| W \circ P_{\mix,{\cal X}})
\le &
\frac{v^s e^{-s H_{1+s}(P_A)} e^{s I_s(X;E|W \times P_{\mix,{\cal X}})}}{ s}
\Label{9-19-1-b}
\end{align}
for $s \in (0,1]$, 
where $v$ is the number of eigenvalues of $W \circ {P_{\mix,{\cal X}}}$.
Similarly,
\begin{align}
& 
\rE_{\bZ,\tilde{X}} e^{s \underline{D}( W \circ f_{\bZ|\tilde{X}}\circ P_A \| W \circ {P_{\mix,{\cal X}}})} \nonumber \\
\le &
\rE_{\bZ,\tilde{X}} e^{s\underline{D}^*_s(W \circ f_{\bZ|\tilde{X}} \circ P_A \| W \circ {P_{\mix,{\cal X}}})} \nonumber \\
\le &
1+ e^{-s H_{1+s}(P_A)}e^{s \underline{I}^*_s(X;E|W \times P_{\mix,{\cal X}})}.
\Label{9-19-2-b}
\end{align}
\end{lem}

Using (\ref{9-19-2-b}) in Lemma \ref{Lee3}, we obtain
\begin{align}
& 
\rE_{\bZ,\tilde{X}} e^{s D( W \circ f_{\bZ|\tilde{X}} \circ P_A \| W \circ {P_{\mix,{\cal X}}})} \nonumber \\
\le &
v^s \rE_{\bZ,\tilde{X}} e^{s D( {\cal E}_{W \circ {P_{\mix,{\cal X}}}}
[W] \circ f_{\bZ|\tilde{X}}\circ P_A  \| W \circ {P_{\mix,{\cal X}}})} \nonumber\\
\le &
v^s (1+ e^{-s H_{1+s}(P_A)}e^{s\underline{I}^*_s
(X;E|{\cal E}_{W \circ {P_{\mix,{\cal X}}}}[W] \times P_{\mix,{\cal X}})}) \nonumber\\
= &
v^s (1+ e^{-s H_{1+s}(P_A)} e^{s I_s(X;E|{\cal E}_{W \circ {P_{\mix,{\cal X}}}}[W] \times P_{\mix,{\cal X}})}) \nonumber\\
\le &
v^s (1+ e^{-s H_{1+s}(P_A)} e^{s I_s(X;E|W \times P_{\mix,{\cal X}})}).
\Label{9-20-1-b}
\end{align}

Here, we construct 
an ensemble of submodules $\{C[\bZ]\}$ satisfying Condition \ref{C2}
and 
an ensemble of injective isomorphisms $\{f_{\bZ}\}$ satisfying Condition \ref{C2-b}
when ${\cal X}$ and ${\cal A}$ are given as vector spaces $\bF_q^k$ and $\bF_q^l$ of a finite field $\bF_q$
($k \ge l$).
Let $\bZ$ be the Toeplitz matrix of the size $(k-l) \times l$, which contains
$k -1$ random variables taking values in the finite field $\bF_q$, and
$\bZ'$ be the random matrix taking values in the set of invertible matrixes of the size $l \times l$
with the uniform distribution.
When $f_{\bZ}$ is given by the multiplication of the random matrix $(I,\bZ)^T$,
and $C[\bZ]$ is given by $f_{\bZ}({\cal A})$,
$\{C[\bZ]\}$ satisfies Condition \ref{C2}.
When $f_{\bZ',\bZ}$ is given by the multiplication of the random matrix $(\bZ',\bZ)^T$
with two independent random variables $\bZ'$ and $\bZ$,
the ensemble $\{f_{\bZ}\}$ satisfies Condition \ref{C2-b}.

\section{Quantum wiretap channel in a general framework}\Label{s2}
\subsection{Single-shot bounds}\Label{s2a}
Next, we consider the quantum wiretap channel model, in which
the eavesdropper (wire-tapper) Eve and the authorized receiver Bob
receive information from the authorized sender Alice.
In this case, in order for
Eve to have less information,
Alice chooses a suitable encoding.
This problem is formulated as follows.
Let $\cH_{B}$ and $\cH_{E}$ be the quantum systems of 
Bob and Eve,
and $\cX$ be the alphabet sent by Alice.
Then, the main quantum channel from Alice to Bob
is described by $W_{B}:x \mapsto W_{B|x}$,
and the wire-tapper quantum channel from Alice to 
Eve is described by $W_{E}:x \mapsto W_{E|x}$.
That is, $W_{B|x}$ ($W_{E|x}$) is Bob's (Eve's) density matrix on the system $\cH_{B}$ ($\cH_{E}$).
In this setting,
Alice chooses $\sM$ 
distributions $\Gamma_1, \ldots, \Gamma_{\sM}$ on $\cX$,
and she generates $x\in \cX$ subject to $\Gamma_m$
when she wants to send the message $m \in {\cal M}= \{1, \ldots, \sM\}$.
In this case, Bob (Eve) receives the density matrix $W_{B} \circ \Gamma_m = \sum_{x} \Gamma_m (x) W_{B|x}$ 
($W_{E} \circ \Gamma_m:= \sum_{x} \Gamma_m (x) W_{E|x}$).
Bob prepares $\sM$-outcome POVM 
$\{\cD_1,\ldots, \cD_{\sM}\}$ of $\cH_{B}$.
Here, we regard $\{\Gamma_1, \ldots, \Gamma_{\sM}\}$ as a transition matrix $\Gamma$ from
${\cal M}$ to $\cX$.
Therefore, the triplet $(\sM,\Gamma,
\{\cD_1,\ldots, \cD_{\sM}\})$ is called a
code, and is described by $\Phi$.
Its performance is given by the following three quantities.
The first is the size $\sM$, which is denoted by $|\Phi|$.
The second is the average 
error probability $\epsilon_{B}(\Phi)$:
\begin{align*}
\epsilon(W_{B}| \Phi) &:= \frac{1}{\sM} \sum_{m=1}^{\sM}  \Tr W_{B}\circ {\Gamma_m} (I-\cD_i ) \\
\epsilon_{\max}(W_{B}| \Phi) &:= \max_{m=1, \ldots, \sM} \Tr W_{B} \circ {\Gamma_m} (I-\cD_i ) ,
\end{align*}
and the third is Eve's information regarding the transmitted message.
Using the channel $W_{E}\circ \Gamma :m \mapsto W_{E} \circ {\Gamma_m}$,
we can describe this quantity by the following ways:
\begin{align*} 
I_{\mix}(M;E|W_{E}\circ \Gamma) &= I(M;E|W_{E} \circ \Gamma \times P_{\mix,{\cal M}}) \\
I_{\max}(M;E|W_{E}\circ \Gamma) &= \max_{P_M} I(M;E|W_{E} \circ \Gamma\times P_{M}) \\
d_{1,\mix}(M;E|W_{E}\circ \Gamma) &= 
\sum_{m=1}^{\sM} \frac{1}{\sM} \| W_{E} \circ {\Gamma_m} - W_{E} \circ \Gamma \circ P_{\mix,{\cal M}} \|_1 \nonumber \\
d_{1,\max}(M;E|W_{E}\circ \Gamma) &=
\max_{P_M} \sum_{m=1}^{\sM} P_M(m) \| W_{E} \circ {\Gamma_m} - W_{E} \circ \Gamma \circ P_{M} \|_1. \nonumber
\end{align*}
In the usual setting, the distribution $\Gamma_i$ can be chosen to 
a uniform distribution.
However, sometimes, 
it is difficult to prepare a perfect uniform distribution
in a realistic setting.
So, in the following, 
we make our code with a non-uniform distribution $P_L$ on ${\cal L}=\{1, \ldots, \sL\}$.

Now, we make a code $\Phi[\bZ]$ for the quantum wiretap channel based on the random coding method
for given integers ${\sL}$ and $\sM$.
In the following, Alice is allowed to generate a random number on $\{1,\ldots, {\sL}\}$ with the distribution $P_L$.
First, we generate the random code $\Phi[\bZ]'$ with size ${\sL}\sM$,
which is described by the
${\sL}\sM$ independent and identical random
variables $\bZ$ subject to the distribution $P_X$ on $\cX$.
That is, all of $\{\Phi[\bZ]_{l,m}'\}$ are independent and obeys the distribution $P_X$ on $\cX$.
For integers $l=1, \ldots, {\sL}$ and $m=1, \ldots, \sM$, 
As is guaranteed in the previous paper \cite{expo-chan}, 
we choose the decoder (POVM) $\cD_{l,m}'[\bZ]$ of the code $\Phi[\bZ]'$
such that
the ensemble expectation of the average error probability concerning decoding the input message $A$ is less than
$4(\sM {\sL})^{s}e^{-s I_{-s}(X;B|W_{B} \times P_X)}$ for $0 \le s \le 1$.
The proof given in \cite{expo-chan} is valid even if 
the prior distribution for sent messages is not uniform.
That is, 
when the message $(l,m)$ is sent with the probability $\frac{P_L(l)}{\sM}$,
the ensemble expectation of the average error probability concerning decoding the input message $A$ is 
bounded by $4(\sM{\sL})^{s}e^{-s I_{-s}(X;B|W_{B} \times P_X)}$ for $0 \le s \le 1$.
Using the code $\Phi[\bZ]'$,
Alice encodes her message $M=1,\ldots, \sM$ in the following way.
In order to send her message $m$,
she generates the random number $L$ subject to $P_L$,
and inputs the element $\Phi[\bZ]_{L,m}'$ in the channel $W_{B}$.
That is,
Alice generates the input signal $X$ subject to the distribution 
$\Gamma[\bZ]_m$, which is defined as $\Gamma[\bZ]_m(x):= 
\sum_l P_L(l) \delta_{\Phi[\bZ]_{l,m}',x}$,
where 
\begin{align}
\delta_{x,x'}
= \left\{
\begin{array}{ll}
1 & \hbox{if~} x=x' \\
0 & \hbox{if~} x\neq x' .
\end{array} 
\right.
\end{align}
Bob recovers the values $l$ and $m$ by using the decoder $\{\cD_{l,m}'[\bZ]\}$,
and discards the value $l$.
That is, 
Bob recovers the value $m$ by using the decoder $\cD_{m}[\bZ]:=\sum_{l=1}^{\sL} \cD_{l,m}'[\bZ]$.
We denote the code $(\sM, \Gamma[\bZ], \{\cD_{m}[\bZ]\})$ by $\Phi[\bZ]$.
Therefore, 
the above discussion in \cite{expo-chan} yields that
\begin{align}
\rE_{\bZ} \epsilon(W_{B}|\Phi[\bZ])  \le 4 \min_{0\le s\le 1}(\sM{\sL})^{s}e^{-s I_{-s}(X;B|W_{B} \times P_X)}.
\end{align}

Using (\ref{9-19-1}),
for $0<s\le 1$,
we obtain
\begin{align}
\rE_{\bZ} D(W_{E} \circ \Gamma[\bZ]_m \| W_{E} \circ P_X ) 
\le 
\frac{v^s}{s} 
e^{-s H_{1+s}(P_L) + s I_{s}(X;E|W_{E} \times P_X) },\Label{10-11-1}
\end{align}
where $v$ is the number of eigenvalues of $W_{E}\circ P_X $.
Thus, using (\ref{8-04-a}) and (\ref{10-11-1}), we obtain
\begin{align}
& \rE_{\bZ} I_{\mix}(M;E|W_{E} \circ \Gamma[\bZ]) \nonumber \\
= & 
\rE_{\bZ} 
I(M;E| (W_{E} \circ \Gamma[\bZ]) \times P_{\mix,{\cal M}}) \nonumber \\
\le & 
\rE_{\bZ} 
I(M;E| (W_{E} \circ \Gamma[\bZ]) \times P_{\mix,{\cal M}} \| W_{E} \circ P_X )
\nonumber \\
= &
\rE_{\bZ} 
\sum_{m=1}^{\sM} \frac{1}{{\sM}}
D(W_{E} \circ \Gamma[\bZ]_m \| W_{E}\circ P_X ) \nonumber \\
\le &
\frac{v^s}{s} 
e^{-s H_{1+s}(P_L) }
e^{s I_{s}(X;E|W_{E} \times P_X)}.
\Label{2-25-1} 
\end{align}
Similarly, as is shown latter,
we obtain 
\begin{align}
& \rE_{\bZ} d_{1,\mix}(M;E|W_{E} \circ \Gamma[\bZ]) \nonumber \\
\le &
(4+ \sqrt{v_s})
e^{-\frac{s}{2} H_{1+s}(P_L) + \frac{s}{2}I^{{\rm G}}_s(X;E|W_{E} \times P_X)}
\Label{8-26-10-a} \\
& \rE_{\bZ} d_{1,\mix}(M;E|W_{E} \circ \Gamma[\bZ]) \nonumber \\
\le &
(4+ \sqrt{\lceil \lambda_s \rceil})
e^{-\frac{s}{2} H_{1+s}(P_L) + \frac{s}{2}I^{{\rm G}}_s(X;E|W_{E} \times P_X)+ \frac{s}{2}}
\Label{8-26-10-a02} 
\end{align}
for $0 < s \le 1$,
where 
$v_s$ is the number of eigenvalues of $\sum_{x}P_X(x) W_{E|x}^{1+s}$
and 
$\lambda_s$ is defined as
the real number
$\log a_1-\log a_0$
by using 
the maximum eigenvalue $a_1$
and  
the minimum eigenvalue $a_0$ of $\sum_{x}P_X(x) W_{E|x}^{1+s}/\Tr \sum_{x}P_X(x) W_{E|x}^{1+s}$.

Finally, we consider what code is derived from the above random coding discussion.
Using the Markov inequality, we obtain
\begin{align}
\rP_{\bZ} 
\{ \epsilon(W_{B}|\Phi[\bZ]) \ge 3 \rE_{\bZ} \epsilon(W_{B}|\Phi[\bZ]) \}^c
&\,< \frac{1}{3} \nonumber \\
\rP_{\bZ} 
\{ I_{\mix}(M;E|W_{E} \circ \Gamma[\bZ]) \ge 3 \rE_{\bZ} I_{\mix}(M;E|W_{E} \circ \Gamma[\bZ]) \}^c
&\,< \frac{1}{3} \nonumber \\
\rP_{\bZ} 
\{ d_{1,\mix}(M;E|W_{E} \circ \Gamma[\bZ]) \ge 3 \rE_{\bZ} d_{1,\mix}(M;E|W_{E} \circ \Gamma[\bZ])\}^c
&\,< \frac{1}{3} .
\Label{10-20-101}
\end{align}
Therefore, the existence of a good code is guaranteed in the following way.
That is, we give the concrete performance of a code 
whose existence is shown in the above random coding method.

\begin{thm}\Label{3-6}
Assume that a random variable subject to the distribution $P_L$ on $\{1, \ldots, {\sL}\}$
is available.
For any integer ${\sM}$ and any probability distribution $P_X$ on $\cX$,
there exists a code $\Phi$ with an encoder $\Gamma$ 
such that
the code $\Phi$ only uses the distribution $P_L$ for mixing the input alphabet
and
\begin{align}
|\Phi| &={\sM} \nonumber \\
\epsilon(W_{B}|\Phi) & \le 12 \min_{0\le s\le 1}({\sM}{\sL})^{s}e^{-s I_{-s}(X;B|W_{B} \times P_X)}
\Label{3-8-1}\\
I_{\mix}(M;E|W_{E} \circ \Gamma) & \le  3
\min_{0 \le s \le 1}
v^s \frac{e^{-s H_{1+s}(P_L)+s I_s(X;E|W_{E} \times P_X)}}{s}
\Label{7-1-2} \\
d_{1,\mix}(M;E|W_{E} \circ \Gamma)
& \le 
3
\min_{0 \le s \le 1} \mu_s
e^{-\frac{s}{2}H_{1+s}(P_L)+\frac{s}{2}I^{{\rm G}}_s(X;E|W_{E} \times P_X)} 
\end{align}
where 
$\mu_s := 
\min\{4 +\sqrt{v_s}, (4+ \sqrt{\lceil \lambda_s \rceil})e^{\frac{s}{2}} \}$.
\end{thm}

\begin{cor}\Label{c3-6}
Assume that a random variable $L$ subject to the distribution $P_L$ on $\{1, \ldots, {\sL}\}$
is available.
Then, 
for any integer ${\sM}$
and any probability distribution $P_X$ on $\cX$,
there exists a code $\Phi$ with an encoder $\Gamma$ 
such that
the code $\Phi$ only uses the distribution $P_L$ for mixing the input alphabet
and
\begin{align}
|\Phi| &={\sM}/4 \nonumber \\
\epsilon_{\max}(W_{B}| \Phi) & \le 48 \min_{0\le s\le 1}({\sM}{\sL})^{s}e^{-s I_{-s}(X;B|W_{B} \times P_X)}
\Label{3-8-1-a}\\
I_{\max}(M;E|W_{E}\circ \Gamma) & \le  12
\min_{0 \le s \le 1} v^s
\frac{e^{-s H_{1+s}(P_L)+s I_{s}(X;E|W_{E} \times P_X)}}{ s}
\Label{7-1-2-x} \\
d_{1,\max}(M;E|W_{E}\circ \Gamma) 
& \le 
24
\min_{0 \le s \le 1}
\mu_s e^{-\frac{s}{2}H_{1+s}(P_L)+\frac{s}{2}I^{{\rm G}}_s(X;E|W_{E} \times P_X)}.
\end{align}
\end{cor}

\begin{proof}
Now, we prove Corollary \ref{c3-6} using code $\Phi$ given in Theorem \ref{3-6}.
When $M$ is regarded as a random variable obeying the uniform distribution on ${\cal M}$,
Markov inequality guarantees that
$\Tr W_{B} \circ {\Gamma_M} (I-\cD_M ) \ge  4 \epsilon_{B}(\Phi)$ 
holds at most probability $1/4$.
Similarly,
$D(W_{E} \circ {\Gamma_M}\|W_{E} \circ \Gamma \circ P_{\mix,{\cal M}}) \ge 4 I_{\mix}(M;E|W_{E}\circ \Gamma)$
and
$\| W_{E} \circ {\Gamma_M} - W_{E} \circ \Gamma \circ P_{\mix,{\cal M}}\|_1 \ge 4 d_{1,\mix}(M;E|W_{E}\circ \Gamma) $
hold at most probability $1/4$, respectively.
So, the random variable $M$ satisfies the three relations 
$\Tr W_{B} \circ {\Gamma_M} (I-\cD_M ) \le  4 \epsilon_{B}(\Phi)$, 
$D(W_{E} \circ {\Gamma_M}\|W_{E} \circ \Gamma \circ P_{\mix,{\cal M}}) \le 4 I_{\mix}(M;E|W_{E}\circ \Gamma)$,
and
$\| W_{E} \circ {\Gamma_M} - W_{E} \circ \Gamma \circ P_{\mix,{\cal M}}\|_1 \le 4 d_{1,\mix}(M;E|W_{E}\circ \Gamma) $
at least probability $1/4=1-3/4$.
In other words, there exist at least ${\sM}/4$ elements $m_1,\ldots, m_{{\sM}/4}$ such that
$\Tr W_{B} \circ {\Gamma_{m_l}} (I-\cD_{i_l} ) \le  4 \epsilon_{B}(\Phi)$,
$D(W_{E} \circ {\Gamma_{m_l}}\|W_{E}\circ \Gamma \circ P_{\mix,{\cal M}}) \le 4 I_{\mix}(M;E|W_{E}\circ \Gamma)$,
$\| W_{E} \circ {\Gamma_{m_l}} - W_{E}\circ \Gamma \circ P_{\mix,{\cal M}}\|_1 \le 4 d_{1,\mix}(M;E|W_{E}\circ \Gamma) $
for $l=1,\ldots, {\sM}/4$.

So, we define the code 
$\tilde{\Phi}:=({\sM}/4,\{\tilde{\Gamma}_1,\ldots,\tilde{\Gamma}_{{\sM}/4}\}, \{{\cal D}_1,\ldots,{\cal D}_{{\sM}/4}\})$
with $\tilde{\Gamma}_l:=\Gamma_{m_l}$ and $\tilde{{\cal D}}_l:={\cal D}_{m_l}$.
Then, 
$\epsilon_{\max}(W_{B}| \tilde{\Phi})\le 4 \epsilon(W_{B}|\Phi)$.
Consider a distribution $P_{M'}$ on ${\cal M}':=\{1, \ldots, {\sM}/4 \}$.
Then, using (\ref{8-04-a}), we obtain
\begin{align*}
& I_{\max}(M';E|W_{E} \circ {\tilde{\Gamma}}) 
=\max_{P_{M'}} I(M';E| W_{E} \circ {\tilde{\Gamma}}\times P_{M'}) \\
\le &
\max_{P_{M'}} I(M';E| W_{E} \circ {\tilde{\Gamma}}\times P_{M'}\| W_E \circ \Gamma \circ P_{\mix,{\cal M}} ) \\
=&
\max_{P_{M'}} 
\sum_{l=1}^{{\sM}/4}
P_{M'}(l)
D(W_{E} \circ {{\Gamma}_{m_l}}\| W_E \circ \Gamma \circ P_{\mix,{\cal M}}) \\
= &
\max_{l=1,\ldots, {\sM}/4}
D(W_{E} \circ {{\Gamma}_{m_l}}\| W_E \circ \Gamma \circ P_{\mix,{\cal M}}) 
\le 4 I_{\mix}(M;E|W_{E}\circ \Gamma).
\end{align*}
Since 
any distribution $P_{M'}$ on ${\cal M}'$ satisfies 
$\| W_E \circ \tilde{\Gamma} \circ P_{M'} - W_E \circ \Gamma \circ P_{\mix,{\cal M}}\|_1
\le \max_{l=1,\ldots, {\sM}/4} \| W_{E} \circ {\Gamma_{m_l}} - W_E \circ \Gamma \circ P_{\mix,{\cal M}}\|_1$,
\begin{align*}
& d_{1,\max}(M';E|W_{E}\circ {\tilde{\Gamma}})
=
\max_{P_{M'}} \sum_{l=1}^{{\sM}/4}
P_{M'}(l)
\| W_{E}\circ {\Gamma_{m_l}} -  W_E \circ \tilde{\Gamma} \circ P_{M'} \|_1 \\
\le &
\max_{P_{M'}} \sum_{l=1}^{{\sM}/4}
P_{M'}(l)
(\| W_{E}\circ {\Gamma_{m_l}} - W_E \circ \Gamma \circ P_{\mix,{\cal M}}\|_1 + \| W_E \circ \tilde{\Gamma} \circ P_{M'} - W_E \circ \Gamma \circ P_{\mix,{\cal M}}\|_1) \\
= &
\max_{l=1,\ldots, {\sM}/4} 
(\| W_{E}\circ {\Gamma_{m_l}} - W_E \circ \Gamma \circ P_{\mix,{\cal M}}\|_1 + \| W_E \circ \tilde{\Gamma} \circ P_{M'}- W_E \circ \Gamma \circ P_{\mix,{\cal M}}\|_1) \\
\le &
2 \max_{l=1,\ldots, {\sM}/4} \| W_{E}\circ {\Gamma_{m_l}} - W_E \circ \Gamma \circ P_{\mix,{\cal M}}\|_1
\le 8 d_{1,\mix}(M;E|W_{E}\circ \Gamma) .
\end{align*}
So, we obtain the desired argument.
\end{proof}


\quad {\it Proofs of (\ref{8-26-10-a}) and (\ref{8-26-10-a02}):}\quad
In order to show (\ref{8-26-10-a}), we consider another protocol generating a secret random number.
Alice and Bob prepare $l$ random permutations ${\bg}=(g_1,\ldots, g_l)$ among $\{1,\ldots, {\sM}\}$.
First, Alice sends Bob the message $l$ and $m$ based on the code $ \Phi[\bZ]'$.
Second, 
Alice and Bob apply the $l$ random permutations to their message in the way $(m,l)\mapsto (g_l(m),l)$.
Finally, Alice and Bob discard $l$ and obtain $g_l(m)$.
That is, Alice and Bob apply a hash function $f_{{\bg}}: (m,l)\to g_l(m)$.
Alice apply the encoder $\Gamma[\bZ,\bg]_m(x):=\sum_{l=1}^{\sL}P_L(l) \delta_{\Phi[\bZ]'_{l,g_l(l)},x}$.
Bob recovers the value $m$ by using the decoder $\cD_{m}[\bZ,\bg]:=\sum_{l=1}^{\sL} \cD_{l,g_l(m)}'[\bZ]$.
We denote the code $(\sM, \Gamma[\bZ,\bg], \{\cD_{m}[\bZ,\bg]\})$ by $\Phi[\bZ,\bg]$.
Due to the construction, the hash function $f_{{\bg}}$ satisfies 
Condition \ref{C1} given in Appendix \ref{s4-1}.
Applying (\ref{8-26-13-a}), 
for any density $\sigma$ on the system ${\cal H}_{E}$,
we obtain
\begin{align}
& \rE_{{\bg}|\bZ} d_{1,\mix}(M;E|W_{E}\circ \Gamma[\bZ,\bg]) \nonumber \\
=& \rE_{{\bg}|\bZ} 
d_{1}'(f_{\bg}(M,L);E|
W_{E}\circ \Phi[\bZ]' \times P_{\mix,{\cal M}\times {\cal L}}) \nonumber \\
\le & (4+ \sqrt{v}) {\sM}^{\frac{s}{2}} (
\Tr 
\sum_{m=1}^{{\sM}}  
\sum_{l=1}^{{\sL}}
(\frac{P_L(l)}{{\sM}})^{1+s}
W_{E|\Phi[\bZ]_{l,m}'}^{1+s}
\sigma^{-s}
)^{\frac{1}{2}} \nonumber \\
= & (4+ \sqrt{v})
(
\sum_{m=1}^{{\sM}}  
\frac{1}{{\sM}}
\sum_{l=1}^{{\sL}}
P_L(l)^{1+s}
\Tr 
W_{E|\Phi[\bZ]_{l,m}'}^{1+s}
\sigma^{-s}
)^{\frac{1}{2}} ,\Label{8-27-1}
\end{align}
where $v$ is the number of eigenvalues of $\sigma$.
Taking the average concerning the variable $\bZ$,
we obtain
\begin{align}
& \rE_{\bZ} \rE_{\bg|\bZ} 
d_{1,\mix}(M;E|W_{E}\circ \Gamma[\bZ,\bg])
\nonumber \\
\le & 
(4+ \sqrt{v})
\rE_{\bZ}
(
\sum_{m=1}^{{\sM}}  
\frac{1}{{\sM}}
\sum_{l=1}^{{\sL}}
P_L(l)^{1+s}
\Tr 
W_{E|\Phi[\bZ]_{l,m}'}^{1+s}
\sigma^{-s}
)^{\frac{1}{2}} \nonumber \\
\le & (4+ \sqrt{v})
(
\sum_{m=1}^{{\sM}}  
\frac{1}{{\sM}}
\sum_{l=1}^{{\sL}}
P_L(l)^{1+s}
\rE_{\bZ}
\Tr 
W_{E|\Phi[\bZ]_{l,m}'}^{1+s}
\sigma^{-s}
)^{\frac{1}{2}} \nonumber \\
= & (4+ \sqrt{v}) 
(
\sum_{l=1}^{{\sL}}
P_L(l)^{1+s}
\Tr 
\sum_{x}P_X(x)  (W_{E|x})^{1+s} \sigma^{-s}
)^{\frac{1}{2}}  .
\Label{8-25-1}
\end{align}
According to Lemma \ref{l3},
we choose $\sigma$ to be 
$c(\sum_{x}P_X(x)(W_{E|x})^{1+s})^{1/(1+s)}$ with the normalizing constant $c$
and $v$ to be the number of eigenvalues of  
$\sum_{x}P_X(x)(W_{E|x})^{1+s}$.
Then, we obtain
\begin{align}
& \rE_{\bg} \rE_{\bZ|\bg} 
d_{1,\mix}(M;E|W_{E}\circ \Gamma[\bZ,\bg])
 \nonumber \\
\le & (4+ \sqrt{v_s} )
(\sum_{l=1}^{{\sL}} P_L(l)^{1+s})^{-\frac{1}{2}}
(\Tr (\sum_{x}P_X(x)  W_{E|x}^{1+s})^{\frac{1}{1+s}})^{\frac{1+s}{2}} \nonumber \\
= & (4+ \sqrt{v_s})
e^{-\frac{s}{2}H_{1+s}(P_L)+\frac{s}{2}I^{{\rm G}}_s(X;E|W_{E} \times P_X)}.
\Label{8-26-10}
\end{align}
Finally, 
when $\bg$ is fixed, 
The statistical behavior of $\Gamma[\bZ,\bg]$ is the same as that of $\Gamma[\bZ]$.
That is,
\begin{align}
 \rE_{\bZ} d_{1,\mix}(M;E|W_{E}\circ \Gamma[\bZ,\bg])
=
 \rE_{\bZ} d_{1,\mix}(M;E|W_{E}\circ \Gamma[\bZ])
\end{align}
for $\bg$, which implies (\ref{8-26-10-a}).
Similarly, we can show (\ref{8-26-10-a02}) by using (\ref{8-26-13-f2})
instead of (\ref{8-26-13-a}).
\endproof

\begin{rem}
One might consider that
Inequality (\ref{10-11-1}) could be derived from a kind of privacy amplification lemma
\cite[Theorem 1]{H-precise} similar to 
(\ref{8-26-10-a}) from \cite[(12)]{H-leaked}.
However, the strategy cannot yield Inequality (\ref{10-11-1}) due to the following reason.
In order to show an inequality corresponding to (\ref{8-25-1}),
we need to show 
\begin{align}
& 
\rE_{\bZ}
\sum_{m,l}
\frac{P_L(l)^{1+s}}{{\sM}}
\Tr 
 (W_{E} \circ {\Gamma[\bZ]_{l,m}})^{1+s}
(
\sum_{m,l}
\frac{P_L(l)}{{\sM}}
W_{E} \circ {\Gamma[\bZ]_{l,m}}
)^{-s} \nonumber \\
\le & 
(
\sum_{l}
P_L(l)^{1+s})
\Tr 
(\sum_{x}P_X(x)  W_{E|x}^{1+s} )
(\sum_{x}P_X(x) W_{E|x})^{-s}.
\Label{8-25-2}
\end{align}
However, only the opposite inequality holds, in general.
Hence, this method cannot be applied to the proof of (\ref{10-11-1}).
\end{rem}

Further,
using (\ref{8-19-14-a}) and (\ref{8-26-9-a}),
we can obtain other type bounds.
Combining (\ref{8-19-14-a}) and (\ref{10-11-1}),
we obtain an alternative bound of $\rE_{\bZ} d_{1,\mix}(M;E|W_{E}\circ \Gamma[\bZ])$ as
\begin{align}
& \rE_{\bZ} d_{1,\mix}(M;E|W_{E}\circ \Gamma[\bZ])
\le \sqrt{\rE_{\bZ} d_{1,\mix}(M;E|W_{E}\circ \Gamma[\bZ])^2} \nonumber \\
\le &
\frac{\sqrt{2}v^{s/2}}{\sqrt{s}} 
e^{-\frac{s}{2} H_{1+s}(P_L) +\frac{s}{2} I_{s}(X;E|W_{E} \times P_X)}.
\Label{8-24-11}
\end{align}
Combining (\ref{8-26-9-a}), 
(\ref{8-26-10-a}), and (\ref{8-26-10-a02}),
we obtain
\begin{align}
& \rE_{\bZ} D(W_{E}\circ {\Gamma[\bZ]_m} \| W_{E} \circ P_X ) \nonumber \\
\le &
(\log d_{E} - \log \mu_s
-\frac{s}{2} H_{1+s}(P_L) - \frac{s}{2}I^{{\rm G}}_s(X;E|W_{E} \times P_X))\nonumber \\
& \cdot \mu_s
e^{-\frac{s}{2} H_{1+s}(P_L) + \frac{s}{2}I^{{\rm G}}_s(X;E|W_{E} \times P_X)}.
\nonumber 
\end{align}
Hence, instead of (\ref{10-11-1}),
we obtain an alternative bound of $\rE_{\bZ} I_{\mix}(M;E|W_{E} \circ \Gamma[\bZ])$ as
\begin{align}
& 
\rE_{\bZ} I_{\mix}(M;E|W_{E} \circ \Gamma[\bZ]) \nonumber \\
\le & 
(\log d_{E} - \log \mu_s
-\frac{s}{2} H_{1+s}(P_L) - \frac{s}{2}I^{{\rm G}}_s(X;E|W_{E} \times P_X))\nonumber \\
& \cdot \mu_s
e^{-\frac{s}{2} H_{1+s}(P_L) + \frac{s}{2}I^{{\rm G}}_s(X;E|W_{E} \times P_X)} .
\Label{8-24-12}
\end{align}
As is shown in Subsection \ref{s2c},
these alternative bounds are weaker than the bounds (\ref{2-25-1}), 
(\ref{8-26-10-a}), and (\ref{8-26-10-a02}) in the asymptotic setting.

\subsection{Asymptotic analysis}\Label{s2b}
In the following, we focus on the $n$-fold discrete memoryless channels 
of the channels $W_{B}$ and $W_{E}$,
which are written as $W_{B}^{(n)}$ and $W_{E}^{(n)}$.
The $n$-independent and identical distribution $P_X^n$ of $P_X$
satisfies the additive equation
$I_s(X;B|W_{B}^{(n)} \times P_X^n)= n I_s(X;B|W_{B} \times P_X)$.
In this case, we assume that a random variable subject to the distribution $P_{L_n}$ on ${\cal L}_n:= \{1, \ldots, {\sL}_n\}$
is available.
Thus, there exists a code $\Phi_n[P_X]$ with an encoder $\Gamma_n[P_X]$ for any integer ${\sM}_n$,
and any probability distribution $P_X$ on $\cX$
such that
the code $\Phi_n$ only uses the distribution $P_{L_n}$ for mixing the input alphabet
and
\begin{align}
|\Phi_n| &={\sM}_n/4 \nonumber \\
\epsilon_{\max}(W_{B}^{(n)}| \Phi_n[P_X]) & \le 48 \min_{0\le s\le 1}({\sM}_n {\sL}_n)^{s}e^{n -s I_{-s}(X;B|W_{B} \times P_X)}
\Label{7-1-2-b-} 
\\
I_{\max}(W_{E}^{(n)}\circ \Gamma_n[P_X]) & \le  12
v_n \min_{0 \le s \le 1}
\frac{e^{-s H_{1+s}(P_{L_n}) +n s I_{s}(X;E|W_{E} \times P_X)}}{s}
\Label{7-1-2-a} 
\end{align}
and
\begin{align}
d_{1,\max}(W_{E}^{(n)}\circ \Gamma_n[P_X]) 
\le &
24 \min_{0 \le s \le 1}
\mu_{s,n}
e^{-\frac{s}{2}H_{1+s}(P_{L_n}) +n \frac{s}{2}I^{{\rm G}}_s(X;E|W_{E} \times P_X)},
\Label{7-1-2-b} \\
\mu_{s,n}:=&
\min\{
4+\sqrt{v_{s,n}},
(4+ \sqrt{\lceil n \lambda_{s} \rceil})e^{\frac{s}{2}} \},
\end{align}
where $v_n$ and $v_{s,n}$ are the numbers of eigenvalues of 
$(W_{E} \circ P_X)^{\otimes n}$
and
$(\sum_{x}P(x)(W_{E|x})^{1+s})^{\otimes n}$.

The numbers $v_n$ and $v_{s,n}$ are bounded by $(\dim {\cal H}_{E}-1)^{n+1}$.
Hence, due to (\ref{9-5-c1}),
when 
${\sM}_n {\sL}_n\cong e^{n I(X;B|W_{B} \times P_X)}$ and $e^{-s H_{1+s}(P_{L_n})} \cong {\sL}_n^s
\cong e^{n sI(X;E|W_{E} \times P_X)}$,
the values (\ref{7-1-2-b-}), (\ref{7-1-2-a}), and (\ref{7-1-2-a}) go to zero.
That is, the rate $\max_{P_X} I(X;B|W_{B}\times P_X)-I(X;E|W_{E} \times P_X)$ can be asymptotically attained, as shown by Devetak \cite{Deve}.
Now, we focus on the exponential decreasing rates of our upper bounds.
We assume that $H_{2}(P_{L_n}) \ge nR$, which implies that
$H_{1+s}(P_{L_n}) \ge H_{2}(P_{L_n}) \ge nR$ for $s\in [0,1]$.
In fact, 
$H_{2}(P_{L_n})$ can be regarded as the sacrifice information.

Now, we denote the ensemble of codes and encoders given in Subsection \ref{s2a} with the $n$-i.i.d. distribution $P_X^n$
by $\Phi[\bZ_n,P_X^n]$ and $\Gamma[\bZ_n,P_X^n]$, respectively.
Using the ensemble,
we define two kinds of the decreasing rates under the above code:
\begin{align*}
e_I(R|W_{E},P_X) &:=
\lim_{n \to \infty} \frac{-1}{n}\log 
\rE_{\bZ_n} I_{\mix}(M_n;E|W_{E}^{(n)}\circ \Gamma[\bZ_n,P_X^n]) \\
e_d(R|W_{E},P_X) &:=
\lim_{n \to \infty} \frac{-1}{n}\log 
\rE_{\bZ_n} d_{1,\mix}(M_n;E|W_{E}^{(n)} \circ \Gamma[\bZ_n,P_X^n]) .
\end{align*}
Inequality (\ref{7-1-2-a}) yields 
\begin{align*}
e_I(R|W_{E},P_X) 
\ge 
e_{{\rm R}}(R|W_{E},P_X):=
\max_{0\le s\le1} s R-s I_s(X;E|W_{E} \times P_X),
\end{align*}
and Inequality (\ref{7-1-2-b}) yields 
\begin{align*}
& e_d(R|W_{E},P_X) \\
\ge &
e_{{\rm G}}(R|W_{E},P_X):=
\max_{0\le s\le1} \frac{s}{2} R -\frac{s}{2} I^{{\rm G}}_s(X;E|W_{E} \times P_X) .
\end{align*}
Note that
the exponent $ e_{{\rm R}}(R|W_{E},P_X)$ with the commutative case
is the same as 
that by the previous paper \cite{H-leaked}.
However,
the exponent $e_{{\rm G}}(R|W_{E},P_X)$ with the commutative case
is smaller than that by the previous paper \cite{H-tight}.

In the quantum wiretap channel model, 
it is known that
a pre noisy processing $\Gamma:{\cal V} \to {\cal X}$ may improve the capacity
in the classical case,
where $\Gamma$ is a stochastic matrix from ${\cal V}$ to ${\cal X}$.
When we apply the pre noisy processing $\Gamma$,
the rate $\max_{P_V} I(V;B|W_{B}\circ \Gamma \times P_V)-I(V;E|W_{E}\circ \Gamma \times P_V)$ can be attained asymptotically,
where $P_V$ is the distribution on ${\cal V}$.
Applying our method to the pair of channels $W_{B}\circ \Gamma$ and $W_{E}\circ \Gamma$,
we obtain an upper bound for error probability and leaked information, which goes to zero exponentially
and can attain the rate $\max_{P_V,\Gamma} I(V;B|W_{B}\circ \Gamma \times P_V)-I(V;E|W_{E}\circ \Gamma \times P_V)$.
Further, by choosing $\Gamma$ as a stochastic matrix from ${\cal V}$ to ${\cal X}^n$,
the rate 
$\lim_{n \to \infty} \frac{1}{n}\max_{P_V,\Gamma} I(V;B|W_{B}^{(n)}\circ \Gamma \times P_V)-I(V;E|W_{E}^{(n)}\circ \Gamma \times P_V)$ 
can be attained asymptotically.

Now, we define the capacity: $C_{W_{B},W_{E}}$
\begin{align*}
& C_{W_{B},W_{E}} \\
:=& 
\sup_{\{\Phi_n ,\Gamma_n\}}
\left \{
\lim_{n \to \infty}\frac{\log |\Phi_n|}{n}
\left|
\begin{array}{l}
\lim_{n \to \infty} I_{\mix}(M_n;E|W_{E} \circ \Gamma_n)=0,\\
\lim_{n \to \infty} \epsilon_{B}(\Phi_n)=0
\end{array}
\right. \right\}.
\end{align*}
Then, we obtain 
\begin{align}
C_{W_{B},W_{E}}\ge
\lim_{n \to \infty} \frac{1}{n}\max_{P_V,\Gamma} I(V;B|W_{B}^{(n)}\circ \Gamma \times P_V)-I(V;E|W_{E}^{(n)}\circ \Gamma \times P_V).
\Label{8-23-1a}
\end{align}

In fact, the following proposition was shown by Devetak \cite{Deve}.
\begin{proposition}
\begin{align}
C_{W_{B},W_{E}} =
\lim_{n \to \infty} \frac{1}{n}\max_{P_V,\Gamma} I(V;B|W_{B}^{(n)}\circ \Gamma \times P_V)-I(V;E|W_{E}^{(n)}\circ \Gamma \times P_V).
\end{align}
\end{proposition}

If there is a quantum channel $C$ such that
$C(W_{B|x})=W_{E|x} $, the quantum wiretap channel $W_{B},W_{E}$ is called degraded.
It is known that the degraded channel $W_{B},W_{E}$ satisfies that\cite{DevP}\cite[(9.62)]{Hayashi}
\begin{align*}
& \max_{P_X} I(X;B|W_{B} \times P_X)-I(X;E|W_{E} \times P_X)\\
=&
\frac{1}{n}\max_{P_V,\Gamma} I(V;B|W_{B}^{(n)}\circ \Gamma \times P_V)-I(V;E|W_{E}^{(n)}\circ \Gamma \times P_V).
\end{align*}
That is,
\begin{align*}
C_{W_{B},W_{E}} 
=
\max_{P_X} I(X;B| W_{B} \times P_X)-I(X;E|W_{E} \times P_X).
\end{align*}
The detail property of $C_{W_{B},W_{E}} $
has been studied by \cite{Wat} in this case.

In the classical case, 
Csisz\'{a}r et al \cite{CK79} showed 
\begin{align*}
& \max_{P_V,\Gamma} I(V;B|W_{B}\circ \Gamma \times P_V)-I(V;E|W_{E}\circ \Gamma \times P_V) \\
= &
\frac{1}{n}\max_{P_V,\Gamma} I(V;B|W_{B}^{(n)}\circ \Gamma \times P_V)-I(V;E|W_{E}^{(n)}\circ \Gamma \times P_V).
\end{align*}
That is,
\begin{align*}
C_{W_{B},W_{E}} 
=
\max_{P_V,\Gamma} I(V;B|W_{B}\circ \Gamma \times P_V)-I(V;E|W_{E}\circ \Gamma \times P_V).
\end{align*}

\subsection{Comparison}\Label{s2c}
When we replace the role of the pair of (\ref{8-26-10-a}) and (\ref{8-26-10-a02}) by that of (\ref{8-24-11}),
we obtain another bound
\begin{align*}
e_d(R|W_{E},P_X) 
\ge
\frac{1}{2}e_{{\rm R}}(R|W_{E},P_X).
\end{align*}
Similarly, replacing the role of (\ref{10-11-1}) by that of (\ref{8-24-12}), 
we obtain
\begin{align*}
e_I(R|W_{E},P_X) \ge e_{{\rm G}}(R|W_{E},P_X).
\end{align*}

For a comparison between these inequalities and 
Inequalities (\ref{7-1-2-b}) and (\ref{7-1-2-a}),
we have the following lemma.
\begin{lem}\Label{l-13}
In the general case, we have
\begin{align}
\frac{1}{2}e_{{\rm R}}(R|W_{E},P_X) \le e_{{\rm G}}(R|W_{E},P_X). 
\Label{9-2-1}
\end{align}
When $W_x$ are commutative each other,
we have
\begin{align}
e_{{\rm G}}(R|W_{E},P_X) 
\le e_{{\rm R}}(R|W_{E},P_X) .
\Label{9-2-2}
\end{align}
\end{lem}
Hence, 
Inequality (\ref{7-1-2-b}) provides a better bound for $e_d(R|W_{E},P_X)$,
That is,
Inequality (\ref{10-11-1}) is a better evaluation for a sufficiently large number $n$.
In the commutative case,
Inequality (\ref{7-1-2-a}) provides a better bound for $e_I(R|W_{E},P_X)$.
That is,
Inequalities 
(\ref{8-26-10-a}) and (\ref{8-26-10-a02})
 are better evaluations for a sufficiently large number $n$.
These numerical comparisons for a non-commutative example
will be given in Section \ref{s10}.

\begin{proof}
Inequality (\ref{10-25-101}) of Lemma \ref{l3} yields that
\begin{align*}
& \frac{1}{2}e_{{\rm R}}(R|W_{E},P_X)
=
\max_{0\le s\le1} \frac{s}{2} R-\frac{s}{2} I_{s}(X;E|W_{E} \times P_X) \\
\le &
\max_{0\le s\le1} \frac{s}{2} R-\frac{s}{2}I^{{\rm G}}_s(X;E|W_{E} \times P_X)
=
e_{{\rm G}}(R|W_{E},P_X),
\end{align*}
which implies (\ref{9-2-1}).

Inequality (\ref{9-2-2}) is shown in the following way.
Since the map $s\mapsto I_{s}(X;E|W_{E} \times P_X)$
is monotonically increasing,
Lemma \ref{l2} yields  
\begin{align*}
& \frac{s}{2} I^{{\rm G}}_s(X;E|W_{E} \times P_X)
\ge
\frac{s}{2} I_{\frac{s}{1+s}}(X;E|W_{E} \times P_X) \\
\ge &
\frac{s}{2} I_{\frac{s}{2}}(X;E|W_{E} \times P_X).
\end{align*}
Then, 
\begin{align*}
& e_{{\rm G}}(R|W_{E},P_X)
=\max_{0\le s\le1} \frac{s}{2} R-\frac{s}{2}I^{{\rm G}}_s(X;E|W_{E} \times P_X) \\
\le &
\max_{0\le s\le1} \frac{s}{2} R- \frac{s}{2} I_{s/2}(X;E|W_{E} \times P_X) 
\le 
e_{{\rm R}}(R|W_{E},P_X),
\end{align*}
which implies (\ref{9-2-2}).
\end{proof}

On the other hand,
combining (9.79) and (9.53) in the book \cite{Hayashi-book},
we obtain 
\begin{align*}
e_I(R|W_{E},P_X) 
&\ge 
\max_{0\le s\le1} \frac{s R-s I_{s}(X;E|W_{E} \times P_X)}{2(1+s)} \\
e_d(R|W_{E},P_X) 
&\ge 
\max_{0\le s\le1} \frac{s R-s I_{s}(X;E|W_{E} \times P_X)}{2(1+s)} .
\end{align*}
Our lower bounds of exponents
$e_{{\rm R}}(R|W_{E}, P_X)$
and
$e_{{\rm G}}(R|W_{E},P_X)$
improve them.

Next, we compare the evaluations (\ref{10-11-1}) and (\ref{8-24-12})
for $\rE_{\Phi} D( W \circ \Phi \circ P_A \| W \circ {P_X}) $
when the number $n$ is not sufficiently large.
In this case, 
the polynomial factors $v$ and $\mu_s$
play an important role.
When $d_{E}=\dim {\cal H}_{E}$ is two,
$v$ increases only linearly.
However, when $d_{E} \ge 3$,
they increase with the order $n^{d_{E}-1}$.
This factor might be not negligible 
when $n$ is not sufficiently large.
However, when $d_{E}$ is finite,
$\lambda$ increases linearly. 
In the evaluation (\ref{8-24-12}),
we have other factors, i.e.,  
$\log d_{E}$ and the logarithm of the upper bound given in (\ref{8-26-10-a}) or 
(\ref{8-26-10-a02}).
These factors also increase linearly.
Hence, in the evaluation (\ref{8-24-12}),
all of the polynomial factors 
increase in the order $n^{3/2}$
while
the polynomial factor in the evaluation (\ref{10-11-1})
increases in the order $n^{s(d_{E}-1)}$.
Now, we fix the optimal $s$ in the evaluation (\ref{10-11-1}).
Then, when $d_{E} > \frac{3}{2s}+1$,
the polynomial factor 
in the evaluation (\ref{8-24-12})
is smaller than that in the evaluation (\ref{10-11-1}).
Hence, when $n$ is not sufficiently large,
the evaluation (\ref{8-24-12}) might be better than 
the evaluation (\ref{10-11-1}).
 

Finally, we compare these evaluations 
when the channel $W^{(n)}$ is not stationary memoryless
and the distribution $P_{L_n}$ is not independent and identical. 
In this case,
the speeds of increase of 
$v$ and $v_s$
are not polynomial in general.
Hence, even though $n$ is sufficiently large,
the factor $v$ and $v_s$ are not negligible.
However, 
$\lambda_s$
increases linearly 
when the logarithm of the minimum eigenvalue of 
$(\sum_x P_X(x)^{(n)} (W_x^{(n)})^{1+s})^{\frac{1}{1+s}}/\Tr (\sum_x P_X(x)^{(n)} (W_x^{(n)})^{1+s})^{\frac{1}{1+s}}$
behaves linearly.
Hence, the evaluations (\ref{8-26-10-a02}) and 
(\ref{8-24-12}) work well under the above weak assumption.

\section{Equivocation rate}\Label{s2-1}
When the coding rate $R$ is larger than the capacity,
the Eve's information does not go to zero.
In this case, it is usual in the classical setting to evaluate the limit 
of $ \frac{I_{\mix}(M_n;E|W_E^{(n)}\circ \Gamma_n)}{n }$ with a sequence of encoders $\Gamma_n$.
In the same construction as Section \ref{s2},
by using (\ref{8-04-a}) and the convexity of $x \mapsto e^x$,
(\ref{9-20-1}) yields the inequalities
\begin{align}
& \rE_{\bZ} e^{s I_{\mix}(M;E|W_{E}\circ \Gamma[\bZ]) } \Label{8-09-a} \\
\le
& \rE_{\bZ} e^{s{I}(M;E|W_{E} \circ \Gamma[\bZ] \times P_{\mix, {\cal M}} \|W_{E} \circ P_X) } \nonumber \\
= &
\rE_{\bZ} e^{\sum_{m=1}^{\sM} \frac{s}{{\sM}}{D}(W_{E}\circ {\Gamma[\bZ]_m} \|W_{E} \circ P_X) } \nonumber \\
\le &
\rE_{\bZ} \sum_{m=1}^{\sM} \frac{1}{{\sM}}
e^{s{D}(W_{E}\circ {\Gamma[\bZ]_m} \|W_{E} \circ P_X) } \Label{8-09-b} \\
= &
\sum_{m=1}^{\sM} \frac{1}{{\sM}}
\rE_{\bZ} e^{s{D}(W_{E}\circ {\Gamma[\bZ]_m} \|W_{E} \circ P_X) } \nonumber \\
\le &
\sum_{m=1}^{\sM} \frac{1}{{\sM}}
v^s(1+ e^{-s H_{1+s}(P_{L})} e^{s{I}_s(X;E|W_{E} \times P_X)})\Label{8-09-c}\\
=&
v^s(1+ e^{-s H_{1+s}(P_{L})} e^{s{I}_s(X;E|W_{E} \times P_X)})
\Label{2-25-1-d}
\end{align}
for $0 < s \le 1$,
where $v$ is the number of eigenvalues of $W \circ P_X$.
Here, (\ref{8-09-a}), (\ref{8-09-b}), and (\ref{8-09-c}) follow from 
(\ref{8-04-a}), the convexity of $x \mapsto e^x$,
and (\ref{9-20-1}), respectively.

Since the Markov inequality guarantees the inequality
\begin{align*}
\rP_{\bZ} 
\{ \epsilon(W_{B}|\Phi[\bZ]) > 2 \rE_{\bZ} \epsilon(W_{B}|\Phi[\bZ]) \}^c
&< \frac{1}{2} \\
\rP_{\bZ} 
\{ e^{s I_{\mix}(M;E|W_{E}\circ \Gamma[\bZ])} > 2 \rE_{\bZ} e^{s I_{\mix}(M;E|W_{E}\circ \Gamma[\bZ])}\}^c
&< \frac{1}{2} ,
\end{align*}
there exists a code $\Phi$ with an encoder $\Gamma$ such that 
$\epsilon(W_{B}|\Phi) \le 2 \rE_{\bZ} \epsilon(W_{B}|\Phi[\bZ])$
and
$e^{s I_{\mix}(M;E|W_{E} \circ \Gamma)} \le 
2 \rE_{\bZ} e^{s I_{\mix}(M;E|W_{E}\circ \Gamma[\bZ])}$.
Since the relations
\begin{align*}
& \frac{1}{s}\log 2 v^s (1+ e^{-s H_{1+s}(P_{L})} e^{s I_{s}(X;E|W_{E} \times P_X)} ) \nonumber \\
\le &
\log v 
+ \frac{1}{s}
(\log 4+ [ \log e^{-s H_{1+s}(P_{L})} e^{s I_{s}(X;E|W_{E} \times P_X)}]_+)
\nonumber \\
= &
\log v 
+ \frac{1}{s}(\log 4+ [   s I_{s}(X;E|W_{E} \times P_X)-s H_{1+s}(P_{L})  ]_+)
\end{align*}
hold, 
the existence of a good code is guaranteed in the following way.

\begin{thm}\Label{3-6-t}
Assume that a random variable $L$ 
subject to the distribution $P_{L}$ on $\{1, \ldots, {\sL}\}$
is available for an auxiliary random number.
There exists a code $\Phi$ with an encoder $\Gamma$ for any integer ${\sM}$,
and any probability distribution $P_X$ on $\cX$
such that
the encoder $\Gamma$ uses only the distribution $P_{L}$ for mixing the input alphabet and
\begin{align}
|\Phi| &={\sM} \nonumber \\
\epsilon(W_{B}| \Phi) & \le 8 \min_{0\le s'\le 1}({\sM}{\sL})^{s'}e^{-s'I_{-s'}(X;B|W_{B} \times P_X)}
\Label{3-8-1-b}\\
I_{\mix}(M;E|W_{E}\circ \Gamma) & \le 
\log v 
+ \frac{1}{s}(\log 4+ [  s I_{s}(X;E|W_{E} \times P_X)-s H_{1+s}(P_{L})  ]_+),
\Label{7-1-2-b-c} 
\end{align}
for $0 < s \le 1$, 
where $v$ is the number of eigenvalues of $W_{E} \circ P_X$.
\end{thm}

In the $n$-fold discrete memoryless channels $W_{B}^{(n)}$ and $W_{E}^{(n)}$ 
of the channels $W_{B}$ and $W_{E}$,
the additive equation
$I_s(X;B|W_{B}^{(n)} \times P_X^n)= n I_s(X;B|W_{B} \times P_X)$ 
holds.
Assume that a random variable subject to the distribution $P_{L_n}$ on $\{1, \ldots, {\sL}_n\}$
is available.
Thus, for any integer ${\sM}_n$ and any probability distribution $P_X$ on $\cX$,
there exists a code $\Phi_n$ with an encoder $\Gamma_n$ such that
the encoder $\Gamma_n$ only uses the distribution $P_{L_n}$ for mixing the input alphabet
and
\begin{align}
|\Phi_n| =& {\sM}_n \nonumber \\
\epsilon_{B}(\Phi) \le  & 8
\min_{0\le s\le 1}
({\sM}_n {\sL}_n)^{s}e^{-n s I_{-s}(X;B|W_{B} \times P_X)}
 \nonumber \\
\frac{I_{\mix}(M_n;E|W_{E}^{(n)}\circ \Gamma_n)}{n} \le &  
\frac{\log v_n}{n} 
+ \frac{1}{ns}(\log 4+ [   ns I_{s}(X;E|W_{E} \times P_X)-s H_{1+s}(P_{L_n}) ]_+),
\Label{7-1-2-c} 
\end{align}
for $0 < s \le 1$, 
where $v_n$ is the number of eigenvalues of ${W_{E} \circ P_X}^{\otimes n}$.

When the sacrifice information rate is $R_0$, i.e., $s H_{1+s}(P_{L_n})\cong snR_0$,
the above code $\Phi_n$ satisfies 
\begin{align}
\lim_{n \to \infty}\frac{I_{\mix}(M_n;E|W_{E}^{(n)}\circ \Gamma_n)}{n} \le  
[{ I_{s}(X;E|W_{E} \times P_X)}-R_0]_+
\Label{7-1-2-d} 
\end{align}
for $0 < s \le 1$.
Since the function $s \mapsto I_{s}(X;E|W_{E} \times P_X)$ is monotone increasing, 
$\inf_{0<s \le 1} I_{s}(X;E|W_{E} \times P_X)
= \lim_{s \to 0} I_{s}(X;E|W_{E} \times P_X)
= I(X;E|W_{E} \times P_X)$.
Therefore, we obtain 
\begin{align}
\lim_{n \to \infty}\frac{I_{\mix}(M_n;E|W_{E}^{(n)}\circ \Gamma_n)}{n} \le  
I(X;E|W_{E} \times P_X)-R_0 .
\Label{7-1-2-e} 
\end{align}

Now, we define the leaked information rate:
\begin{align*}
& I_{W_{B},W_{E}}(R) \\
:=&
\inf_{\{\Phi_n\}}
\left\{\left.\lim_{n \to \infty}\frac{I_{\mix}(M_n;E|W_{E}^{(n)}\circ \Gamma_n)}{n} 
\right|
\epsilon_{B}(\Phi_n)\to 0,
\frac{\log {\sM}_n}{n}\to R
\right\}.
\end{align*}

We assume ${\sM}_n = e^{nR-\sqrt{n}}, {\sL}_n= e^{n R_0}$ and 
$I(X;B|W_{B} \times P_X)=R+R_0$,
the error probability $\epsilon_{B}(\Phi_n)$ goes to zero
in the above construction.
Then,
\begin{align*}
& \lim_{n \to \infty}\frac{I_{\mix}(M_n;E|W_{E}^{(n)}\circ \Gamma_n)}{n} \le  
I(X;E|W_{E} \times P_X)-( I(X;B|W_{B} \times P_X) - R) \\
= &
R- (I(X;B|W_{B} \times P_X)-I(X;E|W_{E} \times P_X)),
\end{align*}
which implies
\begin{align}
I_{W_{B},W_{E}}(R) 
\le R- (I(X;B|W_{B} \times P_X)-I(X;E|W_{E} \times P_X)).
\Label{10-10-2}
\end{align}

Define
\begin{align}
H(R):=
\lim_{n \to \infty}
\max_{ \{ (P_{V_n},\Gamma_n) \} }
\frac{I(V_n;B|W_{B}^{(n)}\circ {\Gamma}_n \times P_{V_n})-I(V;E|W_{E}^{(n)}\circ {\Gamma_n} \times P_{V_n})}{n}
\Label{10-10-1},
\end{align}
where we take the maximum under the condition 
$R \le
\frac{I(V_n;B|W_{B}^{(n)}\circ {\Gamma_n} \times P_{V_n})}{n}$.
Note that
the limit in RHS of (\ref{10-10-1})
equals 
the limit infimum in RHS of (\ref{10-10-1}).
Now, we choose a pair $(P_{V_n},\Gamma_n)$ such that
$R \le
\frac{I(V_n;B|W_{B}^{(n)}\circ {\Gamma_n} \times P_{V_n})}{n}$.
Applying (\ref{10-10-2}) to 
the channel pair 
$(W_{B}^{(n)}\circ {\Gamma_n}, W_{E}^{(n)}\circ {\Gamma_n})$
with the distribution $P_{V_n}$,
we obtain
\begin{align*}
I_{W_{B},W_{E}}(R) \le  R - \frac{I(V_n;B| W_{B}^{(n)}\circ {\Gamma_n} \times P_{V_n})-
I(V_n;E|W_{E}^{(n)}\circ {\Gamma_n} \times P_{V_n})}{n}, 
\end{align*}
which implies that
\begin{align}
I_{W_{B},W_{E}}(R) \le  R - H(R). 
\Label{8-23-5}
\end{align}

In fact, we obtain the following theorem.
\begin{thm}\Label{t8-23-1}
\begin{align}
I_{W_{B},W_{E}}(R) = R- H(R).\Label{8-23-6}
\end{align}
\end{thm}

The proof of this theorem will be given in the end of this subsection.
According to \cite{Wyner}, we define the {\it equivocation rate}:
\begin{align*}
& H_{W_{B},W_{E}}(R) \\
:=&
\inf_{\{\Phi_n\}}
\left\{\left.
\lim_{n \to \infty}\frac{H(M_n|E| W_{E}^{(n)}\circ {\Gamma_n} \times P_{\mix,{\cal M}_n})}{n} 
\right|
\begin{array}{l}
\epsilon_{B}(\Phi_n)\to 0 \\
\frac{\log {\sM}_n}{n}\to R
\end{array}
\right\},
\end{align*}
where $M_n$ is the random variable to be sent and $\Phi_n=({\sM}_n, \{\Gamma_{n}\},\{{\cal D}_{n}\} )$.
Since $H(M_n|E| W_{E}^{(n)}\circ {\Gamma_n} \times P_{\mix,{\cal M}_n})= 
\log {\sM}_n - I_{\mix}(M;E|W_{E}^{(n)}\circ \Gamma_n)$,
we obtain
\begin{align}
H_{W_{B},W_{E}}(R)=H(R).
\end{align}

We define the critical rate
\begin{align}
R^*:= \sup_{\{(P_{V_n},\Gamma_n)\}}
I(V_n;B|W_{B}^{(n)}\circ {\Gamma_n} \times P_{V_n}),
\end{align}
where we take the supremum under the condition
$\lim_{n \to \infty} 
\frac{I(V_n;B|W_{B}^{(n)}\circ {\Gamma_n} \times P_{V_n})-I(V_n;E|W_{E}^{(n)}\circ {\Gamma_n} \times P_{V_n})}{n}
=C_{W_{B},W_{E}}$.
Then, we have
$H(R)=C_{W_{B},W_{E}}$ for $R \le R^*$.
For $R > R^*$,
$H(R)$ is smaller than the capacity $C_{W_{B},W_{E}}$.
We also have the following lemma.
\begin{lem}\Label{l8-23-1}
In the case of degraded channel,
$H(R)$ is calculated as follows.
\begin{align}
H(R)=
\max_{ P_X: I(X;B|W_{B}\times P_X) \ge R }
I(X;B|W_{B}\times P_X)-I(X;E|W_{E}\times P_X).\Label{8-23-7b}
\end{align}
\end{lem}

In the general classical case,
$H(R)$ can be single-letterized
by using two auxiliary random variables\cite{Wyner,CK79}.
In the general quantum case,
the converse part for the single-letterized formula
has a crucial difficulty,
and the direct part requires 
a quantum analogue of the superposition coding. 
Thus, we do not treat the single-letterization of $H(R)$.

\begin{proofof}{Lemma \ref{l8-23-1}}
Since the channel is degraded,
the inequality
\begin{align*}
& I(V_n;B|W_{B}^{(n)}\circ {\Gamma_n} \times P_{V_n})-I(V_n;E|W_{E}^{(n)}\circ {\Gamma_n} \times P_{V_n}) \\
\le &
I(X^n;B|W_{B}^{(n)} \times \Gamma_n \circ P_{V_n})-I(X^n;E|W_{E}^{(n)} \times \Gamma_n \circ P_{V_n})
\end{align*}
holds \cite[Exercise 9.19]{Hayashi}. 
Further, 
for two classical-quantum channels
$W_{B}$ and $\tilde{W}_{B}$ with the input classical system
${\cal X}$ and $\tilde{\cal X}$,
we define the classical-quantum channel
$W_{B} \otimes \tilde{W}_{B}$
by
$(W_{B} \otimes \tilde{W}_{B})_{x,\tilde{x}}:=
W_{B|x} \otimes \tilde{W}_{B|\tilde{x}}$.
For a distribution $P_{X,\tilde{X}}$ on ${\cal X} \times \tilde{\cal X}$,
we choose $P_X$ and $P_{\tilde{X}}$ to be the marginal distributions of $P_{X,\tilde{X}}$.
Then,
we have \cite[Exercise 9.13]{Hayashi}
\begin{align*}
& I(X,\tilde{X};B| W_{B} \otimes \tilde{W}_{B} \times P_{X,\tilde{X}})-
I(X,\tilde{X};E|W_{E} \otimes \tilde{W}_{E} \times P_{X,\tilde{X}}) \\
\le &
I(X;B|W_{B} \times P_X)-
I(X;E|W_{E} \times P_X)
+
I(\tilde{X};B|\tilde{W}_{B} \times P_{\tilde{X}})-
I(\tilde{X};E|\tilde{W}_{E} \times P_{\tilde{X}})
\end{align*}
and
\begin{align*}
I(X,\tilde{X};B|W_{B} \otimes \tilde{W}_{B} \times P_{X,\tilde{X}})
\le
I(X;B|W_{B}\times P_X)
+
I(\tilde{X};B|\tilde{W}_{B}\times P_{\tilde{X}}).
\end{align*}
Hence, for $0\le k \le n$,
\begin{align*}
&\max_{ (P_V,\Gamma) : 
I(V;B|W_{B}^{(n)}\circ {\Gamma} \times P_V) \ge R}
I(V;B|W_{B}^{(n)}\circ {\Gamma} \times P_V)
-
I(V;E|W_{E}^{(n)}\circ {\Gamma} \times P_V)
\\
=&
\max_{ P_{X^n}: I(X^n;B|W_{B}^{(n)} \times P_X) \ge R}
I(X^n;B|W_{B}^{(n)} \times P_{X^n})-I(X;E|W_{E}^{(n)} \times P_{X^n}) \\
=&
\max_{ (P_{X^k},P_{X^{n-k}}): I(X^k;B|W_B^{(k)} \times P_{X^k})+I(X^{n-k};B|W_B^{(n-k)}\times P_{X^{n-k}}) \ge R}
(I(X^k;B|W_B^{(k)}\times P_{X^k})-I(X^k;E|W_E^{(k)}\times P_{X^{k}})) \\
&\quad +
(I(X^{n-k};B|W_B^{(n-k)}\times P_{X^{n-k}})-I(X^{n-k};E|W_E^{(n-k)}\times P_{X^{n-k}})) \\
=&
\max_{ (P_{X,1}, \ldots, P_{X,n}): \sum_i I(X;B|W_{B}\times P_{X,i}) \ge R}
\sum_{i=1}^n (I(X;B|W_{B}\times P_{X,i})-I(X;E|W_{E}\times P_{X,i})).
\end{align*}
Since, 
as is shown latter,
\begin{align}
&\sum_{i=1}^n I(X;B|W_{B}\times P_{X,i}) \le n I(X;B|W_{B}\times \sum_{i=1}^n \frac{1}{n} P_{X,i}) ,
\Label{10-11-1-b}
\\
&\sum_{i=1}^n (I(X;B|W_{B}\times P_{X,i})-I(X;E|W_{E}\times P_{X,i})) \nonumber\\
\le &
n (I(X;B|W_{B}\times \sum_{i=1}^n \frac{1}{n} P_{X,i})- I(X;E|W_{E}\times \sum_{i=1}^n \frac{1}{n} P_{X,i})),
\Label{10-11-2}
\end{align}
the maximum is realized when $P_{X,j}=\sum_{i=1}^n \frac{1}{n} P_{X,i}$.
Therefore, we obtain (\ref{8-23-7b}).

Now, we show (\ref{10-11-1-b}) and (\ref{10-11-2}).
(\ref{10-11-1-b}) is shown from the concavity of von Neumann entropy.
The proof of (\ref{10-11-2}) is more difficult.
It is enough to show 
\begin{align*}
&\lambda (I(X;B|W_{B}\times P_X)-I(X;E|W_{E}\times P_X))+(1-\lambda)(I(X;B|W_{B}\times P_X')-I(X;E|W_{E}\times P_X')) \\
\le &
I(X;B|W_{B}\times  (\lambda P_X+ (1-\lambda)P_X'))- I(X;E|W_{E}\times (\lambda P_X+ (1-\lambda)P_X')).
\end{align*}

We choose a TP-CP map $\Lambda$ such that 
$W_{E|x}=\Lambda(W_{B|x})$.
Then, we obtain
\begin{align}
& D( W_{E} \circ P_X \|  W_{E} \circ (\lambda P_X+ (1-\lambda)P_X')) 
\nonumber \\
=&
D(\Lambda (W_{B} \circ P_X) \|
\Lambda (W_{B} \circ (\lambda P_X+ (1-\lambda)P_X'))) \nonumber \\
\le &
D(W_{B} \circ P_X \| W_{B} \circ (\lambda P_X + (1-\lambda)P_X')).
\Label{10-11-3}
\end{align}
Using (\ref{10-11-3}), we obtain
\begin{align*}
&I(X;B|W_{B}\times (\lambda P_X+ (1-\lambda)P_X'))- I(X;E|W_{E}\times (\lambda P_X+ (1-\lambda)P_X')) \\
& -
\lambda (I(X;B|W_{B}\times P_X)-I(X;E|W_{E}\times P_X))+(1-\lambda)(I(X;B|W_{B}\times P_X')-I(X;E|W_{E}\times P_X')) \\
=&
H( W_{B} \circ (\lambda P_X+ (1-\lambda)P_X'))\\
&-H( W_{E} \circ (\lambda P_X+ (1-\lambda)P_X')) \\
&-\lambda H(W_{B} \circ P_X) 
-(1-\lambda) H(W_{B} \circ P_X') \\
&+\lambda H( W_{E} \circ P_X) 
+(1-\lambda) H( W_{E} \circ P_X') \\
=&
\lambda D( W_{B} \circ P_X \| W_{B} \circ (\lambda P_X + (1-\lambda)P_X'))\\
&+
(1-\lambda) D(W_{B} \circ P_X' \| W_{B} \circ (\lambda P_X+ (1-\lambda)P_X'))
\\
&+\lambda D( W_{E}\circ P_X \| W_{E} \circ (\lambda P_X+ (1-\lambda)P_X'))\\
&+
(1-\lambda) D( W_{E} \circ P_X' \| W_{E} \circ (\lambda P_X+ (1-\lambda)P_X'))
\\
\ge & 0
\end{align*}
\end{proofof}

\begin{proofof}{Theorem \ref{t8-23-1}}
It is sufficient to show the inequality opposite to (\ref{8-23-5}).
Let $\Phi_n=({\sM}_n, \{\Gamma_{n}\},\{{\cal D}_{n}\} )$ be the sequence attaining $I_{W_{B},W_{E}}(R)$.
Then, 
when $M_n$ is the random variable to be sent
and $\hat{M}_n$ is the random variable to be received,
Fano inequality implies that
\begin{align*}
h(\epsilon(W_B^{(n)}|\Phi_n)) + \epsilon(W_B^{(n)}|\Phi_n)\log {\sM}_n \ge H(M_n|\hat{M}_n).
\end{align*}
That is,
\begin{align*}
& -h(\epsilon(W_B^{(n)}|\Phi_n)) + (1-\epsilon(W_B^{(n)}|\Phi_n))\log {\sM}_n \\
\le & I(M_n;\hat{M}_n) \le I(M_n;B|W_{B}^{(n)}\circ {\Gamma_n} \times P_{\mix, {\cal M}_n}).
\end{align*}
Thus,
\begin{align*}
\frac{1}{n}\log {\sM}_n
\le
\frac{I(M_n;B|W_{B}^{(n)}\circ {\Gamma_n}\times P_{\mix, {\cal M}_n}) + h(\epsilon(W_B^{(n)}|\Phi_n))}{n (1-\epsilon(W_B^{(n)}|\Phi_n))}.
\end{align*}
Taking the limit, we obtain
\begin{align*}
R
\le &
\lim_{n \to \infty}
\frac{I(M_n;B|W_{B}^{(n)}\circ {\Gamma_n}\times P_{\mix, {\cal M}_n}) + h(\epsilon(W_B^{(n)}|\Phi_n))}{n (1-\epsilon(W_B^{(n)}|\Phi_n))}
\\
=&\lim_{n \to \infty} \frac{I(M_n;B|W_{B}^{(n)}\circ {\Gamma_n}\times P_{\mix, {\cal M}_n}) }{n}.
\end{align*}
Further,
\begin{align*}
\lim_{n \to \infty} \frac{I(M_n;B|W_{B}^{(n)}\circ {\Gamma_n}\times P_{\mix, {\cal M}_n}) -I(M_n;E|W_{E}^{(n)}\circ {\Gamma_n}\times P_{\mix, {\cal M}_n}) }{n} 
\le H(R).
\end{align*}
Since $I_{W_{B},W_{E}}(R)= \lim_{n \to \infty} \frac{I(M_n;E|W_{E}^{(n)}\circ {\Gamma_n}\times P_{\mix, {\cal M}_n}) }{n}$,
\begin{align*}
& R- I_{W_{B},W_{E}}(R) 
\le 
\lim_{n \to \infty} 
\frac{I(M_n;B|W_{B}^{(n)}\circ {\Gamma_n}\times P_{\mix, {\cal M}_n}) -I(M_n;E|W_{E}^{(n)}\circ {\Gamma_n}\times P_{\mix, {\cal M}_n})}{n} \\
\le & H(R).
\end{align*}
That is 
\begin{align*}
I_{W_{B},W_{E}}(R) \ge 
R- H(R),
\end{align*}
which implies (\ref{8-23-6}).
\end{proofof}

\section{Wire-tap channel with linear coding}\Label{s3}
\subsection{The case when uniform distribution is available}\Label{s3a}
\subsubsection{General case}
In a practical sense, 
we need to take into account the decoding time.
For this purpose, we often restrict our error correcting codes to linear codes.
While the constructions of codes in this section are different from 
those in Section \ref{s2},
the bounds obtained in this section 
are similar to those in Section \ref{s2}.
Hence, the evaluation in the $n$-fold discrete memoryless case
can be derived in the same way as that in Section \ref{s2}
from the single-shot case by substituting the channel $W_E~{(n)}$ into the channel $W_E$.
The source-universality also can be shown in the same way.
Therefore, this section discusses only the single-shot case.

In the following, we consider the case where the sender's space $\cX$ has the structure of a module $\bF_q^m$.
First, we regard a submodule $C_1\subset \cX$ 
as the set of transmitted message,
and focus on its decoding $\{\cD_x\}_{x\in C_1}$ by the authorized receiver.
In the following, for any element $x \in {\cal X}$, 
$[x]_{C_2}$ denotes the coset concerning the quotient by $C_2$,
and 
$[x]_{C_1}$ denotes the coset concerning the quotient by $C_1$.
When the code $C_2$ is fixed, $[x]_{C_2}$ is simplified to $[x]$.
Based on a submodule $C_2$ of $C_1$,
we construct a code for a quantum wiretap channel
$\Phi_{C_1,C_2}= (|C_1/C_2|,
\{\Gamma_{C_1,C_2|[x]}\}_{[x]\in C_1/C_2},
\{\cD_{[x]}\}_{[x]\in C_1/C_2})$
as follows.
The encoding distribution corresponding to the message $[x]\in C_1/C_2$
is given as the uniform distribution $P_{\mix,[x]}$ on the coset $[x]:=x+C_2$, 
and 
the decoding $\cD_{[x]}$ is given as the subset 
$\sum_{x'\in x+C_2} \cD_{x'}$.
As a generalization, for $[y]_{C_1}\in {\cal X}/C_1$,
we consider
a code $\Phi_{[y]_{C_1},C_2}= (|C_1/C_2|,
\{\Gamma_{[y]_{C_1},C_2|[x]}\}_{[x]\in [y]_{C_1}/C_2},
\{\cD_{[x]}\}_{[x]\in [y]_{C_1}/C_2})$
as follows.
The encoding distribution corresponding to the message $[x]\in [y]_{C_1}/C_2$
is given as the uniform distribution on the coset $[x]:=x+C_2$, 
and 
the decoding $\cD_{[x]}$ is given as the subset 
$\sum_{x'\in x+C_2} \cD_{x'}$.
Sometimes, this type code is useful.

Now, we treat the ensemble of the code pairs $\{C_2[\bZ] \subset C_1[\bZ]\}$ 
with the size $|C_1[\bZ]|={\sM}{\sL}$ and $|C_2[\bZ]|={\sL}$.
We identify the set $[y]_{C_1[\bZ]}/C_2[\bZ]$ with the set of messages ${\cal M}=\{1, \ldots, \sM\}$.
We also define $v$,$v_s$, $\lambda_s$ with $P_X=P_{\mix,{\cal X}}$ similar to Subsection \ref{s2a}.


\begin{lem}\Label{L06a}
We suppose 
that the coset $[\bY]_{C_1[\bZ]} \in {\cal X}/C_1[\bZ]$ 
is chosen with the uniform distribution for the given code pair $C_2[\bZ] \subset C_1[\bZ]$
and that
the ensemble $\{C_2[\bZ]\}$ satisfies Condition \ref{C2} as a subset of $\cX$.
The relations 
\begin{align}
& \rE_{\bZ,\bY}
I_{\mix}(M;E|W_{E} \circ \Gamma_{[\bY]_{C_1[\bZ]},C_2[\bZ]}) \nonumber \\
\le &
v^s \frac{e^{s I_{s}(X;E|W_{E} \times P_{\mix,{\cal X}})}}{{\sL}^s s},
\Label{4-27-2-c} \\
& \rE_{\bZ,\bY}
e^{s I_{\mix}(M;E|W_{E} \circ \Gamma_{[\bY]_{C_1[\bZ]},C_2[\bZ]})} \nonumber \\
\le &
v^s (1+ \frac{e^{s I_{s}(X;E|W_{E} \times P_{\mix,{\cal X}})}}{{\sL}^s} )
\Label{2-25-1-c-2} \\
& \rE_{\bZ,\bY} d_{1,\mix}(M;E|W_{E} \circ \Gamma_{[\bY]_{C_1[\bZ]},C_2[\bZ]} ) \nonumber \\
\le &
2 \frac{\mu_s}{{\sL}^{s/2}} e^{\frac{s}{2}I^{{\rm G}}_s(X;E|W_{E} \times P_{\mix,{\cal X}})} 
\Label{8-06-d}
\end{align}
hold
for $0 < s \le 1$ and
$\mu_s := 
\min\{4 +\sqrt{v_s}, (4+ \sqrt{\lceil \lambda_s \rceil})e^{\frac{s}{2}} \}$.
\end{lem}

\begin{lem}\Label{L06b}
We suppose 
that the coset $[\bY]_{C_1[\bZ]} \in {\cal X}/C_1[\bZ]$ 
is chosen with the uniform distribution for the given code pair $C_2[\bZ] \subset C_1[\bZ]$
and that
the ensemble $\{C_1[\bZ]\}$ satisfies Condition \ref{C2} as a subset of $\cX$.
The relation
\begin{align}
& 
\rE_{\bZ,\bY}
\epsilon(W_{B}| \Phi_{[\bY]_{C_1[\bZ]},C_2[\bZ]}) \nonumber \\
\le & 4({\sM}{\sL})^{s}e^{-s I_{-s}(X;B|W_{B} \times P_{\mix,{\cal X}})} \Label{8-06-e}
\end{align}
holds for $0 < s \le 1$.
\end{lem}

Here, 
when ${\cal X}$ is given as a vector space $\bF_q^k$ of a finite field $\bF_q$,
we construct an ensemble of the code pairs $\{C_2[\bZ] \subset C_1[\bZ]\}$ 
such that the ensemble of submodules $\{C_2[\bZ]\}$ and $\{C_1[\bZ]\}$ satisfies Condition \ref{C2} as subsets of $\cX$.
In the following construction, we choose $C_1[\bZ]$ and $C_2[\bZ]$ as $l_1$-dimensional and $l_2$-dimensional subspaces.
First, we fix $ (k-l_2) \times (l_1-l_2)$ matrix $D$ with the rank $l_2-l_1$.
Let $\bZ$ be the Toeplitz matrix of the size $(k-l_2) \times l_2$, which contains
$k -1$ random variables taking values in the finite field $\bF_q$.
When the codes $C_1[\bZ]$ and $C_2[\bZ]$
are given by the images of the random matrix 
$\left(
\begin{array}{cc}
I   &   0 \\
\bZ & D
\end{array}
\right)$
and
$\left(
\begin{array}{c}
I \\
\bZ 
\end{array}
\right)$,
$\{C_2[\bZ]\}$ satisfies Condition \ref{C2}.

Furthermore, 
let $\bZ'$ be the Toeplitz matrix of the size $(k-l_1) \times (l_1-l_2)$, which contains
$k -l_2 -1$ random variables taking values in the finite field $\bF_q$.
When $C_1[\bZ,\bZ']$ and $C_2[\bZ]$
are given by the images of the random matrix 
$\left(
\begin{array}{cc}
I   &   0 \\
\bZ & 
\begin{array}{c}
I \\
\bZ'
\end{array}\\
\end{array}
\right)$
and
$\left(
\begin{array}{c}
I \\
\bZ 
\end{array}
\right)$,
the ensemble $\{C_1[\bZ,\bZ']\}$ 
satisfies Condition \ref{C2} as subsets of $\cX$ as well as $\{C_2[\bZ]\}$. 

\begin{proofof}{Lemma \ref{L06a}}
Using (\ref{8-04-a}) and (\ref{9-19-1-a}),
we obtain 
\begin{align}
& \rE_{\bZ} \rE_{[\bY]_{C_1[\bZ]}} 
I_{\mix}(M;E|W_{E} \circ \Gamma_{[\bY]_{C_1[\bZ]},C_2[\bZ]}) \nonumber \\
\le &
\rE_{\bZ} \rE_{[\bY]_{C_1[\bZ]}} 
I(M;E|W_{E} \circ \Gamma_{[\bY]_{C_1[\bZ]},C_2[\bZ]} \times P_{\mix, [\bY]_{C_1[\bZ]}/C_2[\bZ]}\| W_{E} \circ {P_{\mix,{\cal X}}}) \nonumber \\
= &
 \rE_{\bZ} \rE_{[\bY]_{C_1[\bZ]}} 
\frac{1}{{\sM}}\sum_{ [x]\in [\bY]_{C_1[\bZ]}/C_2[\bZ]} D(W_{E} \circ {P_{\mix,[x]}} 
\|W_{E} \circ {P_{\mix,{\cal X}}}) \nonumber \\
= & \rE_{\bZ} 
\frac{{\sL}}{|{\cal X}|}
\sum_{ [x]\in {\cal X}/C_2[\bZ]} D(W_{E} \circ {P_{\mix,[x]}} \|W_{E} \circ {P_{\mix,\cX}}) \nonumber \\
= & \rE_{\bZ} \rE_{Y} 
D(W_{E} \circ {P_{\mix, C_2[\bZ]+Y}} \|W_{E} \circ {P_{\mix,\cX}}) \nonumber \\
\le &
v^s \frac{e^{s I_{s}(X;E|W_{E} \times P_X)}}{{\sL}^s s},
\end{align}
which implies (\ref{4-27-2-c}).
Similarly, using (\ref{9-20-1-a}) and the convexity of $x \mapsto e^x$,
we obtain
\begin{align}
&
\rE_{\bZ} \rE_{[\bY]_{C_1[\bZ]}} 
e^{s {I}(M;E|W_{E}\circ \Gamma_{[\bY]_{C_1[\bZ]},C_2[\bZ]} \times P_{\mix, [\bY]_{C_1[\bZ]}/C_2[\bZ]}\| W_{E} \circ {P_{\mix,{\cal X}}}) } \nonumber \\
= &
\rE_{\bZ} \rE_{[\bY]_{C_1[\bZ]}} 
e^{s \sum_{ [x]\in [\bY]_{C_1[\bZ]}} \frac{1}{{\sM}}D(W \circ {P_{\mix,[x]}} \|W \circ {P_{\mix,{\cal X}}}) }\nonumber  \\
\le &
\rE_{\bZ} \rE_{[\bY]_{C_1[\bZ]}} 
\sum_{ [x]\in [\bY]_{C_1[\bZ]}} \frac{1}{{\sM}} e^{s  D(W \circ {P_{\mix,[x]}} \|W \circ {P_{\mix,{\cal X}}}) }\Label{8-09-e} \\
=&
\rE_{\bZ} \rE_{Y} 
e^{s  D(W \circ {P_{\mix, C_2[\bZ] +Y}} \|W \circ {P_{\mix,{\cal X}}}) }\nonumber  \\
\le &
v^s (1+ \frac{e^{s {I}_{s}(X;E|W_{E} \times P_{\mix,{\cal X}})}}{{\sL}^s} )
\Label{2-25-1-c-1}
\end{align}
for $0 < s \le 1$,
where (\ref{8-09-e}) and (\ref{2-25-1-c-1}) follow from 
the convexity of $x \mapsto e^x$ and (\ref{9-20-1-a}), respectively.
Relation (\ref{8-04-a}) guarantees that
\begin{align}
& e^{s I_{\mix}(M;E|W_{E} \circ \Gamma_{[\bY]_{C_1[\bZ]},C_2[\bZ]})} 
\le
e^{s I(M;E|W_{E} \circ \Gamma_{[\bY]_{C_1[\bZ]},C_2[\bZ]} \times P_{\mix, [\bY]_{C_1[\bZ]}/C_2[\bZ]}\| W_{E} \circ {P_{\mix,{\cal X}}})} .
\Label{8-06-a}
\end{align}
Thus, combination of (\ref{2-25-1-c-1}) and (\ref{8-06-a}) implies (\ref{2-25-1-c-2}).

Next, we show (\ref{8-06-d}).
For a given value $\bZ$, (\ref{8-09-o}) implies that
\begin{align}
& \rE_{[\bY]_{C_1[\bZ]}|\bZ} 
d_{1,\mix}(M;E|W_{E} \circ \Gamma_{[\bY]_{C_1[\bZ]},C_2[\bZ]}) \nonumber \\
= & d_1'([X] ;[X]_{_{C_1[\bZ]}} ,E|W_{E} \times P_{\mix,\cX}) \nonumber\\
\le & 2 d_1'([X],[X]_{_{C_1[\bZ]}} ; E|W_{E} \times P_{\mix,\cX}) 
=  2 d_1'([X] ; E|W_{E} \times P_{\mix,\cX}) ,
\Label{8-06-c}
\end{align}
where the final equation follows from the fact that
the random variable $[X]_{_{C_1[\bZ]}} $ is given as a function of $[X]$.
Since the ensemble of submodules $\{C_2[\bZ]\}$ satisfies Condition \ref{C2} as a subset of $\cX$,
the ensemble of the hash functions $X \mapsto [X]_{C_2[\bZ]}$ satisfies Condition \ref{C1}.
Hence, Lemmas \ref{l3} and \ref{L-15} guarantee
\begin{align}
\rE_{\bZ} d_1'([X] ; E|W_{E} \times P_{\mix,\cX}) 
\le & \mu_s 
(|\cX|/\sL)^{\frac{1}{2}} 
e^{\frac{s}{2} H_{1+s}(X|E|W_{E} \times P_{\mix,\cX}\|\sigma_E)} \nonumber \\
= &
\frac{\mu_s e^{\frac{s}{2} I_{s}(X;E|W_{E} \times P_{\mix,\cX}\|\sigma_E)} 
}{\sL^{\frac{1}{2}}} \nonumber \\
= &
\frac{\mu_s e^{\frac{s}{2} I_{s}^{{\rm G}}(X;E|W_{E} \times P_{\mix,\cX})} 
}{\sL^{\frac{1}{2}}} ,
\Label{8-06-b}
\end{align}
where $\sigma_E$ is 
$(\sum_x P_{\mix,\cX}(x) W_x^{1+s})^{\frac{1}{1+s}}/\Tr (\sum_x P_{\mix,\cX}(x) W_x^{1+s})^{\frac{1}{1+s}}$.
Combination of (\ref{8-06-c}) and (\ref{8-06-b}) yields (\ref{8-06-d}).
\end{proofof}

\begin{proofof}{Lemma \ref{L06b}}

Since the ensemble of submodules $\{C_1[\bZ]\}$ satisfies Condition \ref{C2} as a subset of $\cX$,
the proof given in \cite{expo-chan} is valid with $P_X=P_{\mix,\cX}$.
Thus, the ensemble expectation of the average error probability concerning decoding the input message is 
bounded by $4({\sM}{\sL})^{s}e^{-s I_{-s}(X;B|W_{B} \times P_{\mix,{\cal X}})}$ for $0 \le s \le 1$.
That is, 
\begin{align}
\rE_{\bZ} \rE_{\bY\in {\cal X}/C_1[\bZ]} 
\epsilon(W_{B}| \Phi_{[\bY]_{C_1[\bZ]},C_2[\bZ]})
\le 4({\sM}{\sL})^{s}e^{-s I_{-s}(X;B|W_{B} \times P_{\mix,{\cal X}})}.
\end{align}
\end{proofof}

\subsubsection{Additive case}
Next, we consider the case of additive channels.
Assume that the channel $W_{E}$ is called additive, i.e., 
the set ${\cal X}$ has a structure of module and 
there exist a state $\rho$ and a projective representation $U$ of ${\cal X}$ such that $W_{E|x}= U_x \rho U_x^\dagger$.
In this case, the relations
\begin{align*}
I_{\mix}(M;E|W_{E} \circ \Gamma_{C_1,C_2}) 
&=I_{\mix}(M;E|W_{E} \circ \Gamma_{[y]_{C_1},C_2}) \\
d_{1,\mix}(M;E|W_{E}\circ \Gamma_{C_1,C_2}) 
&=d_{1,\mix}(M;E|W_{E} \circ \Gamma_{[y]_{C_1},C_2}) 
\end{align*}
hold for $y \in {\cal X}$.
Since $W_{E} \circ \Gamma_{C_1,C_2}$ is an additive channel, 
(\ref{9-4-c1}) guarantees that
\begin{align}
I_{\max}(M;E|W_{E} \circ \Gamma_{C_1,C_2}) 
=
I_{\mix}(M;E|W_{E} \circ \Gamma_{C_1,C_2}) .
\Label{9-4-a8}
\end{align}
Further, any additive channel $W_{B}$ satisfies that
\begin{align}
\epsilon(W_{B}| \Phi_{[\bY]_{C_1[\bZ]},C_2[\bZ]})
=
\epsilon(W_{B}| \Phi_{C_1[\bZ],C_2[\bZ]}).
\end{align}

Hence, (\ref{4-27-2-c}), (\ref{2-25-1-c-2}), (\ref{8-06-d}), and (\ref{8-06-e}) are simplified to
\begin{align}
\rE_{\bZ} 
I_{\max}(M;E|W_{E} \circ \Gamma_{C_1[\bZ],C_2[\bZ]}) 
\le &
v^s \frac{e^{s I_s(X;E|W_{E} \times P_{\mix,\cX})}}{{\sL}^s s},
\Label{4-27-2-c-d} \\
\rE_{\bZ} 
e^{s I_{\max}(M;E|W_{E} \circ \Gamma_{C_1[\bZ],C_2[\bZ]}) }
\le &
v^s(1+ \frac{e^{{I}_s(X;E|W_{E} \times P_{\mix,\cX})}}{{\sL}^s} )
\Label{2-25-1-c-1-d} \\
\rE_{\bZ} d_{1,\mix}(M;E|W_{E} \circ \Gamma_{C_1[\bZ],C_2[\bZ]}) 
\le &
\frac{2 \mu_s}{{\sL}^{s/2}} e^{\frac{s}{2}I^{{\rm G}}_s(X;E|W_{E} \times P_{\mix,{\cal X}})} 
\Label{8-06-d2} \\
\rE_{\bZ}
\epsilon(W_{B}| \Phi_{C_1[\bZ],C_2[\bZ]}) 
\le & 4({\sM}{\sL})^{s}e^{-s I_{-s}(X;B|W_{B} \times P_{\mix,{\cal X}})} \Label{8-06-e2}
\end{align}
for $0 < s \le 1$.

Further, $I_s(X;E|W_{E} \times P_{\mix,\cX})$ 
and $I^{{\rm G}}_s(X;E|W_{E} \times P_{\mix,\cX})$ can be calculated as
\begin{align*}
s I_s(X;E|W_{E} \times P_{\mix,\cX})
=&
\log \Tr \rho^{1+s} \overline{\rho}^{-s}, \\
I(X;E|W_{E} \times P_{\mix,\cX})
=&
H(\overline{\rho})-H(\rho) \\
\frac{s}{1+s}
I^{{\rm G}}_s(X;E|W_{E} \times P_{\mix,\cX})
=&
\log \Tr (\sum_{x\in {\cal X}}\frac{1}{|{\cal X}|} U_x\rho^{1+s}U_x^\dagger)^{1/(1+s)} , 
\end{align*}
where $\overline{\rho}= \sum_{x \in \cX} \frac{1}{|{\cal X}|} U_x \rho U_x^\dagger$.
Especially, when $\rho$ is pure,
\begin{align}
s I_s(X;E|W_{E} \times P_{\mix,\cX})
& =
\log \Tr \overline{\rho}^{1-s} \Label{9-5-c2}\\
I(X;E|W_{E} \times P_{\mix,\cX})
&=
H(\overline{\rho}) \Label{9-5-c3} \\
\frac{s}{1+s}I^{{\rm G}}_s(X;E|W_{E} \times P_{\mix,{\cal A}})
& =
\log \Tr \overline{\rho}^{1/(1+s)}.\Label{9-5-c4}
\end{align}

Now, we consider the case when the code $C_1$ is fixed
and only $C_2$ is randomly chosen.
\begin{lem}\Label{thm10-21}
Assume that the ensemble $\{C_2[\bZ]\}$ satisfies Condition \ref{C2} as a subset of $C_1$.
When the channels $W_{B}$ and $W_{E}$ are additive,
the relations
\begin{align}
\rE_{\bZ} 
I_{\max}(M;E|W_{E} \circ \Gamma_{C_1,C_2[\bZ]}) 
\le &
\eta(
\frac{\mu_s}{{\sL}^{s/2}} e^{\frac{s}{2}I^{{\rm G}}_s(X;E|W_{E} \times P_{\mix,{\cal X}})} ,\log d_E),
\Label{4-27-2-c-d3} \\
\rE_{\bZ} d_{1,\mix}(M;E|W_{E} \circ \Gamma_{C_1,C_2[\bZ]}) 
\le &
\frac{\mu_s}{{\sL}^{s/2}} e^{\frac{s}{2}I^{{\rm G}}_s(X;E|W_{E} \times P_{\mix,{\cal X}})} 
\Label{8-06-d3} 
\end{align}
hold for $0 < s \le 1$.
\end{lem}
Therefore, even if we fixed our error correcting code to $C_1$, 
we can find a subcode $C_2 \subset C_1$ satisfying that
\begin{align*}
I_{\max}(M;E|W_{E} \circ \Gamma_{C_1,C_2}) 
\le &
2
\eta(
\frac{\mu_s}{{\sL}^{s/2}} e^{\frac{s}{2}I^{{\rm G}}_s(X;E|W_{E} \times P_{\mix,{\cal X}})} ,\log d_E),
\\
d_{1,\mix}(M;E|W_{E} \circ \Gamma_{C_1,C_2[\bZ]}) 
\le &
2\frac{\mu_s}{{\sL}^{s/2}} e^{\frac{s}{2}I^{{\rm G}}_s(X;E|W_{E} \times P_{\mix,{\cal X}})} 
\end{align*}
for $0 < s \le 1$.

\begin{proof}
In this case, we have
$d_{1,\mix}(M;E|W_{E} \circ \Gamma_{C_1,C_2[\bZ]}) =d_1'([X] ; E|W_{E} \times P_{\mix,C_1})$.
Hence, (\ref{8-06-b}) implies that
\begin{align}
\rE_{\bZ} 
d_{1,\mix}(M;E|W_{E} \circ \Gamma_{C_1,C_2[\bZ]}) 
\le & 
\frac{\mu_s e^{\frac{s}{2} I_{s}^{{\rm G}}([X];E|W_{E} \times P_{\mix,C_1})} 
}{\sL^{\frac{1}{2}}} .
\Label{8-08-b}
\end{align}
Since $W_{E}$ is additive,
$I_{s}^{{\rm G}}([X];E|W_{E} \times P_{\mix,C_1})=
I_{s}^{{\rm G}}([X];E|W_{E} \times P_{\mix,C_1+x})$.
Hence, Lemma \ref{l1} yields 
\begin{align}
e^{\frac{s}{1+s} I_{s}^{{\rm G}}([X];E|W_{E} \times P_{\mix,C_1})} 
= &
\sum_{x \in \cX}
\frac{1}{|\cX|}
e^{\frac{s}{1+s} I_{s}^{{\rm G}}([X];E|W_{E} \times P_{\mix,C_1+x})} \nonumber\\
\le &
e^{\frac{s}{1+s} I_{s}^{{\rm G}}([X];E|W_{E} \times P_{\mix,\cX})} , \nonumber
\end{align}
which implies that
\begin{align}
I_{s}^{{\rm G}}([X];E|W_{E} \times P_{\mix,C_1})
\le I_{s}^{{\rm G}}([X];E|W_{E} \times P_{\mix,\cX}). \Label{8-08-c}
\end{align}
The combination of 
(\ref{8-08-b}) and (\ref{8-08-c}) yields (\ref{8-06-d3}).
Further,
attaching (\ref{8-26-9-e}) to (\ref{8-06-d3}),
we obtain 
\begin{align}
\rE_{\bZ} 
I_{\mix}(M;E|W_{E} \circ \Gamma_{C_1,C_2[\bZ]}) 
\le 
\eta(
\frac{\mu_s}{{\sL}^{s/2}} e^{\frac{s}{2}I^{{\rm G}}_s(X;E|W_{E} \times 
P_{\mix,{\cal X}})} ,\log d_E),
\Label{9-4-a9}
\end{align}
for $0 < s \le 1$
due to the concavity of 
the map $x \mapsto \eta(x, \log d_E)$.
Finally, (\ref{9-4-a8}) and (\ref{9-4-a8}) imply (\ref{4-27-2-c-d3}). 
\end{proof}

Indeed, when the fixed submodule $C_1$ is isomorphic to a vector space of a finite field $\bF_q$,
we can construct an ensemble of submodules $\{C_2[\bZ]\}$ of $C_1$ satisfying Condition \ref{C2}
by the same method as that given in Section \ref{sV}.



\subsection{The case when uniform distribution is unavailable}\Label{s3b}
Now, we consider the case when the uniform distribution is not available for encoding the message $[x]\in C_1/C_2$.
In this case, we assume that 
a module ${\cal A}$ with the cardinality ${\sL}$ and
a distribution $P_{A}$ on the module ${\cal A}$ is available for this purpose.
We employ a submodule $C_1$ of ${\cal X}$ with the cardinality ${\sM}{\sL}$
and an injective homomorphism $f$ from ${\cal A}$ to $C_1$.
We fix a set of representatives $\{x_1,\ldots, x_{{\sM}}\}$ of all elements of $C_1/f({\cal A})$.
Then, we can define the affine map $f_{|x}(a):=f(a)+x$.

Based on the above structure,
we construct a code for a quantum wiretap channel
$\Phi_{C_1,f,\{x_1,\ldots, x_{{\sM}}\} } = $\par\noindent
$({\sM},
\{ f_{|x_m} \circ P_{A}  \}_{[x_m]\in C_1/f({\cal A})},
\{\cD_{[x]}\}_{[x]\in C_1/f({\cal A})})$
as follows.
The encoding distribution corresponding to the message $[x_m]\in C_1/f({\cal A})$
is given as the distribution $f_{|x_m} \circ P_{A} $ on the coset $[x_m]=x_m+f({\cal A})$, 
and 
the decoding $\cD_{[x]}$ is given as the subset $\sum_{x'\in x+f({\cal A}))} \cD_{x'}$.
As a generalization, for $x \in {\cal X}$,
we consider
a code 
$\Phi_{C_1,f,\{x+x_1,\ldots, x+x_{{\sM}}\} } = 
({\sM},
\{f_{|x+x_m}\circ P_{A} \}_{[x+x_m]\in C_1/f({\cal A})},
\{\cD_{[x']}\}_{[x']\in C_1+x/f({\cal A})})$ as follows.
The encoding distribution corresponding to the message $[x+x_m]\in C_1/f({\cal A})$
is given as the distribution $f_{|x+x_m} \circ P_{A}$ 
on the coset $[x+x_m]=x+x_m+f({\cal A})$, 
and 
the decoding $\cD_{[x']}$ is given as the subset $\sum_{x''\in x'+f({\cal A}) } \cD_{x''}$.
When the hash function $f_{\bZ}$ is randomly chosen as
a homomorphism from ${\cal A}$ to ${\cal X}$
according to a random variable $X$,
the code $C_1[\bZ]$ is also chosen satisfying that 
$f_{\bZ}({\cal A}) \subset C_1[\bZ]$ and $|C_1[\bZ]/f_{\bZ}({\cal A})|=\sM$.

Since the distribution $f_{x+x_M} \circ P_A$ on $\cX$ depends on the random variable $M$ on $\cM$,
it can be regarded as 
the transition matrix $m \mapsto f_{x+x_m} \circ P_A$.
In order to clarify this point, 
this transition matrix is denoted by $P_{X|M}[f_{x+x_M} \circ P_A]$.
Hence,
the transition matrix $(\tilde{x},m) \mapsto f_{\tilde{x}+x_m} \circ P_A$
is denoted by $P_{X|\tilde{X},M}[f_{\tilde{X}+x_M} \circ P_A]$.

\begin{lem}\Label{L09-1}
When the random variable $\tilde{X} \in \tilde{\cal X}={\cal X}$ is subject to the uniform distribution
and $\{f_{\bZ}\}$ satisfies Condition \ref{C2-b},
we obtain the following.
\begin{align}
\rE_{\bZ,\tilde{X}} 
I_{\mix} (M;E|W_{E} \circ P_{X|M}[f_{\bZ|\tilde{X}+x_M} \circ P_A] ) 
\le & \frac{v^s e^{-s H_{1+s}(P_{A})} e^{s I_s(X;E|W_E \times P_{\mix,\cX})}}{s}
\Label{8-09-h}
\\
\rE_{\bZ,\tilde{X}} 
e^{s I_{\mix} (M;E|W_{E} \circ P_{X|M}[f_{\bZ|\tilde{X}+x_M} \circ P_A] ) } 
\le & v^s(1+ e^{-s H_{1+s}(P_{A})}e^{{I}_s(X;E|W_E \times P_{\mix,\cX})})
\Label{8-09-i}\\
\rE_{\bZ,\tilde{X}} 
d_{1,\mix} (M;E|W_{E} \circ P_{X|M}[f_{\bZ|\tilde{X}+x_M} \circ P_A] ) 
\le & 
2 \mu_s e^{\frac{s}{2} (I_s^{{\rm G}}(X;E|W_E \times P_{\mix,\cX})-H_{1+s}(P_{A}))}.
\Label{8-09-l}
\end{align}
\end{lem}

Here, we construct 
an ensemble $\{(C_1[\bZ],f_{\bZ},\{x_{1,\bZ},\ldots, x_{{\sM},\bZ}\})\}$ such that $\{f_{\bZ}\}$ satisfies Condition \ref{C2-b}
when ${\cal X}$ is given as a vector space $\bF_q^k$ of a finite field $\bF_q$.
In the following construction, we choose $C_1[\bZ]$ as an $l_1$-dimensional subspace
and ${\cal A}$ as the $l_2$-dimensional space $\bF_q^{l_2}$.
First, we fix $ (k-l_2) \times (l_1-l_2)$ matrix $D$ with the rank $l_2-l_1$.
Let $\bZ$ be the Toeplitz matrix of the size $(k-l_2) \times l_2$, which contains
$k -1$ random variables taking values in the finite field $\bF_q$, and
$\bZ'$ be the random matrix taking values in the set of invertible matrixes of the size $l_2 \times l_2$
with the uniform distribution.
We choose $f_{\bZ',\bZ}$ to be the multiplication of the random matrix $(\bZ',\bZ)^T$
with two independent random variables $\bZ'$ and $\bZ$,
and
$C_1[\bZ,\bZ']$ to be the image of 
the random matrix 
$\left(
\begin{array}{cc}
\bZ'& 0 \\
\bZ & D
\end{array}
\right)$.
The coset representatives are chosen to be the image of
the matrix 
$\left(
\begin{array}{c}
0 \\
D
\end{array}
\right)$.
Then, $\{f_{\bZ',\bZ}\}$ satisfies Condition \ref{C2-b}.

\begin{proofof}{Lemma \ref{L09-1}}
\begin{align}
& \rE_{\tilde{X}} 
I_{\mix} (M;E|W_{E} \circ P_{X|M}[f_{|\tilde{X}+x_M} \circ P_A] ) \nonumber \\
\le 
& \rE_{\tilde{X}} 
I(M;E|W_{E} \circ P_{X|M}[f_{|\tilde{X}+x_M} \circ P_A] \times P_{\mix,\cM}
\| W_{E} \circ {P_{\mix,{\cal X}}}) \nonumber \\
= &
  \rE_{\tilde{X}} 
\frac{1}{{\sM}}\sum_{m=1}^{\sM} D(W_{E} \circ f_{|\tilde{X}+x_m} \circ P_{A}  \|
W_{E} \circ {P_{\mix,{\cal X}}}) \nonumber \\
= &
\rE_{\tilde{X}} 
D(W_{E} \circ f_{|\tilde{X}} \circ P_{A}  \|W_{E} \circ {P_{\mix,{\cal X}}}) .
\Label{2-25-1-f-1} 
\end{align}
Similarly, we obtain
\begin{align}
& \rE_{\tilde{X}} 
e^{s I_{\mix} (M;E|W_{E} \circ P_{X|M}[f_{|\tilde{X}+x_M} \circ P_A] ) } \nonumber \\
\le &
 \rE_{\tilde{X}} 
e^{s  
I(M;E|W_{E} \circ P_{X|M}[f_{|\tilde{X}+x_M} \circ P_A] \times P_{\mix,\cM}
\| W_E \circ {P_{\mix,{\cal X}}}) } \nonumber \\
= &
 \rE_{\tilde{X}} 
e^{s \frac{1}{{\sM}}\sum_{m=1}^{\sM} {D}(W_E \circ f_{|\tilde{X}+x_m} \circ P_{A}  
 \|W_E \circ {P_{\mix,{\cal X}}})} \nonumber \\
\le &
 \rE_{\tilde{X}} 
\frac{1}{{\sM}}\sum_{m=1}^{\sM} e^{s {D}(W_E \circ f_{|\tilde{X}+x_m}\circ P_{A}   
\|W_E \circ {P_{\mix,{\cal X}}})} \Label{8-09-j} \\
= &
\frac{1}{{\sM}}\sum_{m=1}^{\sM} 
\rE_{\tilde{X}} e^{s {D}(W_E \circ f_{|\tilde{X}+x_m}\circ P_{A}   
\|W_E \circ {P_{\mix,{\cal X}}})} \nonumber \\
= &
 \rE_{\tilde{X}} 
e^{s {D}(W_E \circ f_{|\tilde{X}} \circ P_{A} 
 \|W_E \circ {P_{\mix,{\cal X}}})},
\Label{2-25-1-f-3}
\end{align}
where (\ref{8-09-j}) follows from the convexity of $x \mapsto e^x$.
These upper bounds do not depend on the choice of the set of representatives $\{x_1,\ldots, x_{{\sM}}\}$.

Then, (\ref{9-19-1-b}) and (\ref{9-20-1-b}) imply that
\begin{align}
& 
\rE_{\bZ} \rE_{\tilde{X}} 
D(W_E\circ f_{\bZ|\tilde{X}} \circ P_{A} 
\|W \circ {P_{\mix,\tilde{\cal X}}}) \nonumber \\
\le &
\frac{v^s e^{-s H_{1+s}(P_{A})} e^{s I_s(X;E|W_E \times P_{\mix,\tilde{\cX}})}}{s},
\Label{2-25-1-g-1} \\
& 
\rE_{\bZ} 
 \rE_{\tilde{X}} 
e^{s {D}(W_E \circ f_{\bZ|\tilde{X}} \circ P_{A} 
 \|W_E \circ {P_{\mix,\tilde{\cal X}}})} \nonumber\\
\le & 
v^s(1+ e^{-s H_{1+s}(P_{A})}e^{{I}_s(\tilde{X};E|W_E \times P_{\mix,\tilde{\cX}})}) .
\Label{2-25-1-g-3}
\end{align}
Hence, combination of (\ref{2-25-1-f-1}) and (\ref{2-25-1-g-1}) yields (\ref{8-09-h}),
and
combination of (\ref{2-25-1-f-3}) and (\ref{2-25-1-g-3}) yields (\ref{8-09-i}).

Relation (\ref{8-09-o}) yields that
\begin{align}
& \rE_{\tilde{X}} 
d_{1,\mix} (M;E|W_{E} \circ P_{X|M}[f_{\bZ|\tilde{X}+x_M} \circ P_A] )\nonumber \\
= & 
d_{1}' (M;E,\tilde{X}| W_{E} \circ P_{X|\tilde{X},M}[f_{\bZ|\tilde{X}+x_M} \circ P_A] \times P_{\mix,{\cal M}\times \tilde{\cal X}}) \nonumber\\
\le & 
2 d_{1}' (M,\tilde{X};E| W_{E} \circ P_{X|\tilde{X},M}[f_{\bZ|\tilde{X}+x_M} \circ P_A] \times P_{\mix,{\cal M}\times \tilde{\cal X}}) \nonumber\\
= & 
2 \sum_{\tilde{x} \in \tilde{\cX}, m \in {\cal M}}
\frac{1}{|\tilde{\cX}|\cdot |{\cal M}|}
\| W_{E} \circ f_{\bZ|\tilde{x}+x_m} \circ P_A - 
W_{E} \circ P_{\mix,\tilde{\cX}} \|_1 \nonumber\\
=&
2 \sum_{\tilde{x} \in \tilde{\cX}}
\frac{1}{|\tilde{\cX}|}
\| W_{E} \circ f_{\bZ|\tilde{x}} \circ P_A 
- W_{E} \circ P_{\mix,\tilde{\cX}} \|_1\nonumber \\
=&
2 d_{1}' (\tilde{X};E| W_{E} \circ P_{X|\tilde{X}}[f_{\bZ|\tilde{X}} \circ P_A] \times P_{\mix,\tilde{\cal X}}) .
\Label{8-09-n}
\end{align}
Since the ensemble of the hash functions $(x,a) \mapsto f_{\bZ}(a)+ x \in \cX$
satisfies Condition \ref{C1}, Lemma \ref{L-15} yields that
\begin{align}
& \rE_{\bZ} 
d_{1}' (\tilde{X};E| W_{E} \circ P_{X|\tilde{X}}[f_{\bZ|\tilde{X}} \circ P_A] \times P_{\mix,\tilde{\cal X}}) \nonumber \\
\le &
\mu_s |\cX|^{\frac{s}{2}}
e^{-\frac{s}{2}H_{1+s}(A,X|E| W_E \times P_{\mix,{\cal X}} \times P_A \|\sigma_E) }
\nonumber\\
= &
\mu_s |\cX|^{\frac{s}{2}}
e^{-\frac{s}{2}(H_{1+s}(P_A)+H_{1+s}(X|E| W_E \times P_{\mix,{\cal X}} \|\sigma_E) )}\nonumber \\
=&
\mu_s e^{\frac{s}{2} (I_s(X;E|W_E \times P_{\mix,\cX}\|\sigma_E)-H_{1+s}(P_{A}))}\nonumber\\
=&
\mu_s e^{\frac{s}{2} (I_s^{{\rm G}}(X;E|W_E \times P_{\mix,\cX})-H_{1+s}(P_{A}))}
\Label{8-09-m} ,
\end{align}
where $\sigma_E$ is 
$(\sum_x P_{\mix,\cX}(x) W_{E|x}^{1+s})^{\frac{1}{1+s}}/\Tr (\sum_x P_{\mix,\cX}(x) W_{E|x}^{1+s})^{\frac{1}{1+s}}$.
Hence, combination of (\ref{8-09-n}) and (\ref{8-09-m}) implies (\ref{8-09-l}).
\end{proofof}

When the channel $W_{E}$ is additive,
the relations
\begin{align}
I_{\mix} 
(M;E|W_{E} \circ P_{X|M} [f_{x_M} \circ P_A] ) 
&=I_{\mix} 
(M;E|W_{E} \circ P_{X|M} [f_{x_M+x} \circ P_A] ) 
\nonumber\\
d_{1,\mix}(M;E|W_{E} \circ P_{X|M} [f_{x_M} \circ P_A] ) 
&=d_{1,\mix} 
(M;E|W_{E} \circ P_{X|M} [f_{x_M+x} \circ P_A] ) 
\nonumber
\end{align}
hold for any $x \in \cX$.
Hence, Lemma \ref{L09-1} can be simplified to
\begin{align}
\rE_{\bZ} 
I_{\mix} (M;E|W_{E} \circ P_{X|M}[f_{\bZ|x_M} \circ P_A] ) 
\le & \frac{v^s e^{-s H_{1+s}(P_{A})} e^{s I_s(X;E|W_E \times P_{\mix,\cX})}}{s}
\Label{8-09-h1}
\\
\rE_{\bZ} 
e^{s I_{\mix} (M;E|W_{E} \circ P_{X|M}[f_{\bZ|x_M} \circ P_A] ) } 
\le & v^s(1+ e^{-s H_{1+s}(P_{A})}e^{s{I}_s(X;E|W_E \times P_{\mix,\cX})})
\Label{8-09-i1}\\
\rE_{\bZ} 
d_{1,\mix} (M;E|W_{E} \circ P_{X|M}[f_{\bZ|x_M} \circ P_A] ) 
\le & 
2 \mu_s e^{\frac{s}{2} (I_s^{{\rm G}}(X;E|W_E \times P_{\mix,\cX})-H_{1+s}(P_{A}))}.
\Label{8-09-l1}
\end{align}
Now, similar to Lemma \ref{thm10-21}, we consider the case when the code $C_1$ is fixed
and only $C_2$ is randomly chosen.
\begin{lem}\Label{thm10-212}
Assume that 
$\{f_{\bZ}\}$ is an ensemble of functions from ${\cal A}$ to $C_1$ and
satisfies Condition \ref{C2}.
When the channels $W_{B}$ and $W_{E}$ are additive,
the relations
\begin{align}
\rE_{\bZ} 
I_{\mix} (M;E|W_{E} \circ P_{X|M}[f_{\bZ|x_M} \circ P_A] ) 
\le &
\eta (2\mu_s e^{\frac{s}{2} (I_s^{{\rm G}}(X;E|W_E \times P_{\mix,\cX})-H_{1+s}(P_{A}))}, \log d_E)
\Label{4-27-2-c-d4} \\
\rE_{\bZ} 
d_{1,\mix} (M;E|W_{E} \circ P_{X|M}[f_{\bZ|x_M} \circ P_A] ) 
\le &
2\mu_s e^{\frac{s}{2} (I_s^{{\rm G}}(X;E|W_E \times P_{\mix,\cX})-H_{1+s}(P_{A}))}.
\Label{8-06-d4} 
\end{align}
hold for $0 < s \le 1$.
\end{lem}
\begin{proof}
Relation (\ref{8-06-d4}) can be shown as follows.
Relation (\ref{8-06-d2}) with $\cX=C_1$ implies that
\begin{align*}
& 
\rE_{\bZ} 
d_{1,\mix} (M;E|W_{E} \circ P_{X|M}[f_{\bZ|x_M} \circ P_A] ) \\
\le & 
2 \mu_s e^{\frac{s}{2} (I_s^{{\rm G}}(X;E|W_E \times P_{\mix,C_1})-H_{1+s}(P_{A}))}.
\end{align*}
Hence, using (\ref{8-08-c}), we obtain (\ref{8-06-d4}).
Relation (\ref{4-27-2-c-d4}) 
can be obtained in the same way as (\ref{4-27-2-c-d3}) in Lemma \ref{thm10-21}.
\end{proof}

Therefore, even though the uniform distribution is not available, 
if a distribution close to the uniform distribution is available,
there exists a code with a performance similar to Lemma \ref{thm10-21}.

\section{Application to Pauli channel}\Label{s10}
As a simple example, we treat 
a Pauli channel in the $d$-dimensional system ${\cal H}$.
First, we define 
the discrete Weyl-Heisenberg representation $W$ for $\bbZ_d^2$:
\begin{align}
\sX &:= \sum_{j=1}^d |j+1 \rangle \langle j| , \quad
\sZ := \sum_{j=1}^d \omega^j |j \rangle \langle j| \\
\sW(x,z) &:= \sX^x \sZ^z,
\end{align}
where $\omega$ is the root of the unity with the order $d$.
Using this representation and a probability distribution $P_{X,Z}$ on $\bbZ_d^2$, 
we can define the Pauli channel:
\begin{align}
\Lambda_{P_{X,Z}}(\rho):= \sum_{(x,z)\in \bbZ_d^2}
P_{X,Z}(x,z) \sW(x,z) \rho \sW(x,z)^\dagger .
\end{align}
In the following, we assume that the eavesdropper can access all of the environment of the channel $\Lambda_{P_{X,Z}}$.
When the state $|j\rangle$ is input to the channel $\Lambda_{P_{X,Z}}$, 
the environment system is spanned by the basis $\{|x,z \rangle \}$ and 
the state $W_{E|j}$ of the environment 
is given as
\begin{align}
W_{E|j} &=\sum_{z}
P_Z(z) |j,z:P_{X,Z}\rangle \langle j,z:P_{X,Z}|\\
|j,z:P_{X,Z}\rangle &:= \sum_{x=1}^d \omega^{j x}
\sqrt{P_{X|Z}(x|z)} |x,z \rangle.
\end{align}
Then, the average state is
\begin{align}
\overline{\rho}_{E}=
\sum_{x,z} P_{X,Z}(x,z) |x,z \rangle\langle x,z |
=
\sum_{z} P_{Z}(z) \overline{\rho}_{E|z},
\end{align}
where
\begin{align}
\overline{\rho}_{E|z}:=
\sum_{x} P_{X|Z}(x|z) |x,z \rangle\langle x,z |.
\end{align}
Then, $I_s(X;E| W_{E} \times P_{\mix,\cX})$ and $I(X;E|W_{E} \times P_{\mix,\cX})$ 
are calculated by using the conditional entropy $H (X|Z|P_{X,Z})$ and the conditional R\'{e}nyi entropy $H_{1-s} (X|Z|P_{X,Z})$ as
\begin{align*}
& I(X;E|W_{E} \times P_{\mix,\cX}) =
\frac{1}{d}\sum_{j} D(W_{E|j}\| \overline{\rho}_{E})
=
\frac{1}{d}\sum_{j} 
\sum_{z} P_Z(z)
D( |j,z:P_{X,Z}\rangle \langle j,z:P_{X,Z}| \| \overline{\rho}_{E|z}) \\
=&
\frac{1}{d}\sum_{j} 
\sum_{z} P_Z(z) H(X|P_{X|Z=z})
=
H (X|Z|P_{X,Z}), \\
& e^{s I_s(X;E| W_{E} \times P_{\mix,\cX})} = 
\frac{1}{d}\sum_{j} \Tr W_{E|j}^{1+s} \overline{\rho}_{E}^{-s} 
= \frac{1}{d}\sum_{j} 
\sum_{z} P_Z(z)
\Tr |j,z:P_{X,Z}\rangle \langle j,z:P_{X,Z}|^{1+s} \overline{\rho}_{E|z}^{-s} \\
=& \frac{1}{d}\sum_{j} 
\sum_{z} P_Z(z)
\langle j,z:P_{X,Z}| \overline{\rho}_{E|z}^{-s} |j,z:P_{X,Z}\rangle 
= \frac{1}{d}\sum_{j} 
\sum_{z} P_Z(z)
\sum_{x}P_{X|Z}(x|z)^{1-s} 
=e^{s H_{1-s}(X|Z|P_{X,Z})}, 
\end{align*}
which implies $I_s(X;E| W_{E} \times P_{\mix,\cX})= H_{1-s}(X|Z|P_{X,Z})$.
Similarly, 
we obtain the simplification:
\begin{align*}
& e^{sI^{{\rm G}}_s(X;E| W_{E} \times P_{\mix,{\cal X}})} \\
= &
(\Tr
(
\frac{1}{d}
\sum_j \sum_z P_Z(z)^{1+s}|j,z:P_{X,Z}\rangle \langle j,z:P_{X,Z}| 
)^{\frac{1}{1+s}}
)^{1+s} \\
= &
(\Tr
(
\sum_{z,x}P_{Z}(z)^{1+s} P_{X|Z}(x|z) |x,z\rangle \langle x,z|
)^{\frac{1}{1+s}}
)^{1+s} \\
=&
(\sum_{x,z}P_Z(z)P_{X|Z}(x|z)^{\frac{1}{1+s}} )^{1+s} \\
=&
e^{s H_{\frac{1}{1+s}}(X|Z|P_{X,Z})},
\end{align*}
which implies
\begin{align*}
I^{{\rm G}}_s(X;E| W_{E} \times P_{\mix,{\cal X}})
=
H_{\frac{1}{1+s}}(X|Z|P_{X,Z}).
\end{align*}

When we employ the same ensemble of codes as in Subsection \ref{s3a},
Inequalities (\ref{4-27-2-c-d}), (\ref{2-25-1-c-1-d}), and 
(\ref{8-06-d2}) imply that
\begin{align}
\rE_{\bZ} I_{\max}(M;E|W_{E} \circ \Gamma_{C_1[\bZ],C_2[\bZ]}) 
\le &
v^s \frac{e^{s H_{1-s}(X|Z|P_{X,Z})}}{{\sL}^s s},
\Label{4-27-2-e} \\
\rE_{\bZ} 
e^{s I_{\max}(M;E|W_{E}\circ \Gamma_{C_1[\bZ],C_2[\bZ]}) } 
\le &
v^s (1+ \frac{e^{s H_{1-s}(X|Z|P_{X,Z})}}{{\sL}^s} )
\Label{2-25-1-e-2} \\
\rE_{\bZ} d_{1,\mix}(M;E|W_{E}\circ \Gamma_{C_1[\bZ],C_2[\bZ]}) 
\le &
2 \mu_s \frac{ e^{\frac{s}{2} H_{\frac{1}{1+s}}(X|Z|P_{X,Z})}}{{\sL}^{s/2}}
\Label{4-27-2-e-3} 
\end{align}
for $0 < s \le 1$, where $v$ 
and $v_s$
are the number of 
eigenvalues of $\overline{\rho}_{E}$ and $ \sum_{j} W_{E|j}^{1+s}$.

So, the asymptotic required sacrifice rate is $H (X|Z|P_{X,Z})$.
When the the sacrifice rate is $R$,
the exponential decreasing rate of the upper bounds
defined in Subsection \ref{s2b} are calculated as
\begin{align}
e_{{\rm R}}(R|W_{E},P_{\mix,\cX})
&=
\max_{0< s \le 1}
sR -s H_{1-s}(X|Z|P_{X,Z})  \Label{10-21-103}\\
e_{{\rm G}}(R|W_{E},P_{\mix,\cX})
&=
\max_{0< s \le 1}
\frac{s}{2}R - \frac{s}{2} H_{1/(1+s)}(X|Z|P_{X,Z}) .
\Label{10-21-104}
\end{align}

When the random variables $X$ and $Z$ are independent under the distribution $P_{X,Z}$, i.e.,
$P_{X,Z}(x,z)=P_X(x)P_Z(z)$,
the state of the environment is written in the form
\begin{align}
W_{E|j} &=
|j:P_X\rangle \langle j:P_X| \otimes 
\sum_{z} P_Z(z) |z\rangle\langle z|
\\
|j:P_X\rangle &:= \sum_{x=1}^d \omega^{j x}
\sqrt{P_X(x)} |x \rangle.
\end{align}
Since $\sum_{x} P_X(x) |x \rangle\langle x|$
does not depend on $j$,
the state $\rho_{E|j}$ can be essentially regarded as $|j:P_X\rangle \langle j:P_X|$.
Then, the average state is
\begin{align}
\overline{\rho}_{E}=
\sum_{x} P_X(x) |x \rangle\langle x |.
\end{align}
In this case, 
\begin{align}
I_s(X;E| W_{E}\times P_{\mix,\cX}) &= H_{1-s}(P_X) \\
I(X;E|W_{E} \times P_{\mix,\cX}) & =H (P_X).
\end{align}

Then, in the $n$-fold memoryless extension of the channel 
$\Lambda_{P_{X,Z}}$, 
when the asymptotic sacrifice rate $R$ is greater than
$H (P_X)$, our communication becomes secure.
As is mentioned in Subsection \ref{s2b},
we obtain the lower bound of the exponential decreasing rate of the eavesdropper's Holevo information:
\begin{align}
e_{{\rm R}}(R|W_{E},P_{\mix,\cX})=\max_{0< s \le 1} sR -s H_{1-s}(P_X) ,
\Label{8-19-1}
\end{align}
and
\begin{align}
e_{{\rm G}}(R|W_{E},P_{\mix,\cX})=
\max_{0 \le s \le 1} \frac{s}{2}R  - \frac{s}{2} H_{1/(1+s)}(P_X).
\end{align}
As is shown in Lemma \ref{l-13}, $e_{{\rm R}}(R|W_{E},P_{\mix,\cX})$ is better than $e_{{\rm G}}(R|W_{E},P_{\mix,\cX})$.
A similar analysis has been done by 
the stabilizer formalism by Tsurumaru et al\cite{Tsuru}.
Their approach evaluates the virtual phase error probability, which is less than
\begin{align}
\epsilon(n,R):=
\min_{0 \le s \le 1} e^{n (-sR  +s H_{1/(1+s)}(P_X))}.\Label{8-19-1-b}
\end{align}
Using the relation between the eavesdropper's Holevo information and the virtual phase error probability\cite{H-qkd},
we obtain
\begin{align}
\rE_{\bZ} I_{\mix}(M;E|W_{E} \circ \Gamma_{C_1[\bZ],C_2[\bZ]}) 
\le
h(\epsilon(n,R))+ n \epsilon(n,R) \log d,
\end{align}
where $h(x):= -x \log x-(1-x) \log (1-x)$.
Using their approach, 
we obtain
the lower bound of
the exponential decreasing rate of the eavesdropper's Holevo information:
\begin{align}
\max_{0 \le s \le 1} sR  -s H_{1/(1+s)}(P_X)
=2 e_{{\rm G}}(R|W_{E},P_{\mix,\cX}).
\Label{8-19-2}
\end{align}
As is mentioned in Subsection \ref{s22-2}, the function $s \mapsto H_{1+s}(P_X)$ is monotone decreasing, i.e.,
$s H_{1-s}(P_X) \ge s H_{1/(1+s)}(P_X)$.
Hence, the exponent $2 e_{{\rm G}}(R|W_{E},P_{\mix,\cX})$ by Tsurumaru et al\cite{Tsuru}
is greater than our exponents 
$e_{{\rm R}}(R|W_{E},P_{\mix,\cX})$ and $e_{{\rm G}}(R|W_{E},P_{\mix,\cX})$.
Its numerical verification is given in Fig \ref{f1}.

\begin{figure}[htbp]
\begin{center}
\scalebox{0.7}{\includegraphics[scale=1]{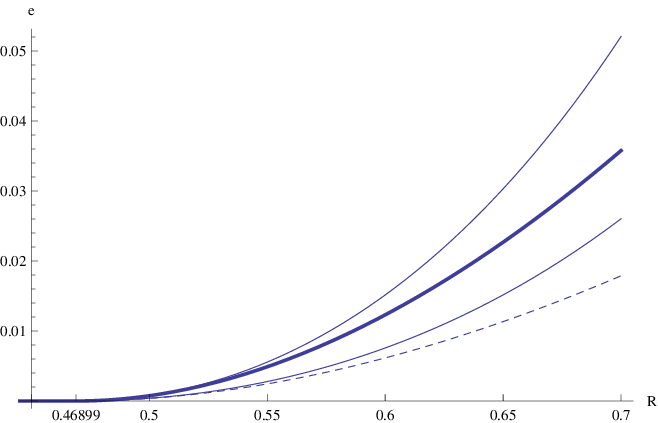}}
\end{center}
\caption{
Lower bounds of exponent.
Thick line: $e_{{\rm R}}(R|W_{E},P_{\mix,\cX})$,
Upper Normal line: $2 e_{{\rm G}}(R|W_{E},P_{\mix,\cX})$\cite{Tsuru},
Lower Normal line: $e_{{\rm G}}(R|W_{E},P_{\mix,\cX})$,
Dashed line: $\frac{1}{2}e_{{\rm R}}(R|W_{E},P_{\mix,\cX})$
with $p=0.1$, 
$H(P_X)$=0.46899.
}
\Label{f1}
\end{figure}%

Next, we focus on $L_1$ distinguishability.
As is mentioned in Subsection \ref{s2b}, we have two lower bounds of exponents 
$\frac{1}{2}e_{{\rm R}}(R|W_{E},P_{\mix,\cX})$ and $e_{{\rm G}}(R|W_{E},P_{\mix,\cX})$ under this criterion,
and the inequality
$\frac{1}{2}e_{{\rm R}}(R|W_{E},P_{\mix,\cX}) \le e_{{\rm G}}(R|W_{E},P_{\mix,\cX})$
holds.
Using the relation between the universal composability and the virtual phase error probability\cite{H-qkd2},
we obtain
\begin{align}
\rE_{\bZ} d_{1,\mix}(M;E|W_{E} \circ \Gamma_{C_1[\bZ],C_2[\bZ]}) 
\le
2\sqrt{2} \sqrt{\epsilon(n,R)}.
\end{align}
Using their approach, 
we obtain
the lower bound of
the exponential decreasing rate of $L_1$ distinguishability:
\begin{align}
\frac{1}{2}\max_{0 \le s \le 1} sR  - s H_{1/(1+s)}(P_X)
=e_{{\rm G}}(R|W_{E},P_{\mix,\cX}).
\Label{8-19-2-a}
\end{align}
That is, the exponent by Tsurumaru et al\cite{Tsuru}
is the same as our better exponent $e_{{\rm G}}(R|W_{E},P_{\mix,\cX})$.
Hence, we can conclude that our method is not better than 
the stabilizer formalism 
when the phase error occurs independently of the bit error in the Pauli channel.

However, when the phase error depends on the bit error in the Pauli channel,
the stabilizer formalism cannot provide the error exponent so clearly.
In this case, in the evaluation of phase error probability, 
we have to take account into error probability in the error correction concerning the bit error.
So, it is not easy to derive a simple bound for the virtual phase error probability
as (\ref{8-19-1-b}).
Our bounds (\ref{4-27-2-e}) and (\ref{4-27-2-e-3})
derive lower bounds (\ref{10-21-103}) and (\ref{10-21-104}) of the exponential decreasing rates 
of the eavesdropper's information and $L_1$ distinguishability.
Therefore, our method has an advantage over the stabilizer formalism when the phase error depends on the bit error. 
Further, by using (\ref{2-25-1-e-2}), our method provides the equivocation rate
while the stabilizer formalism cannot derive the equivocation rate even in the independent case.

\section{Conclusion}
We have given a protocol for quantum wiretap channel with 
an auxiliary random variable subject to a non-uniform distribution.
Then, 
when the distribution of the auxiliary random variable 
is not the uniform distribution but is close to the uniform distribution,
we have derived an upper bound for exponential decreasing rate 
of leaked information in the quantum mutual information criterion 
and $L_1$ distinguishability.
Further, we have derived the equivocation rate for 
a given quantum wiretap channel.
For a practical construction, 
we have proposed a code for quantum wiretap channel that requires only a pair of a linear code and an auxiliary random variable subject to a non-uniform distribution.
We have also derived an upper bound for the leaked information for this protocol.
The organization of the case with a linear code is different from the case with a non-linear code with respect to the construction of the code ensemble and the evaluation of average performance.
We can apply the same discussion from the evaluation of average performance
to the existence of codes and the asymptotic analysis.
Hence, we have omitted the latter part in Section \ref{s3}.
The average performances of the code ensemble based on linear codes
have been evaluated only with the single-shot form.
These results are summarized in Table \ref{table1}.
Finally, we have treated Pauli channel as a typical example.

Further, 
in the evaluations (\ref{8-26-10-a02}) and (\ref{8-24-12}),
the factor $ \lambda_s$ increases only linearly under a general setting.
Hence, 
a part of our results derived from them
are expected to be applied to the non-stationary and non-memoryless case,
e.g., the Markovian case.
Unfortunately, we could not prove Lemma \ref{l2} for the non-commutative case.
The extension of this lemma to the non-commutative case
is a future study.

\begin{table}[htb]
  \caption{Summary of obtained bounds}
\begin{center}
  \begin{tabular}{|l|l|c|c|c|c|c|c|c|} 
\hline
\multicolumn{2}{|l|}{Error correcting code} 
& general random code 
& \multicolumn{4}{|c|}{randomized linear code} 
& \multicolumn{2}{|c|}{fixed linear code} 
\\ \hline
\multicolumn{2}{|l|}{Auxiliary random number} 
& non-uniform
& \multicolumn{2}{|c|}{uniform} 
& \multicolumn{2}{|c|}{non-uniform} 
& non-uniform
& uniform
\\ \hline
\multicolumn{2}{|l|}{Channel}
& general
& general
& additive 
& general
& additive 
& additive 
& additive 
\\ \hline
\multirow{3}{*}{Bound of average}
& $ I_{\mix}(W_E \circ \Gamma) $
&(\ref{2-25-1}) 
&(\ref{4-27-2-c}) 
&(\ref{4-27-2-c-d}) 
&(\ref{8-09-h}) 
&(\ref{8-09-h1}) 
&(\ref{4-27-2-c-d3}) 
&(\ref{4-27-2-c-d4}) 
\\ \cline{2-9}
& $ d_{1,\mix}(W_E \circ \Gamma) $
&(\ref{8-26-10-a}), (\ref{8-26-10-a02}) 
&(\ref{8-06-d}) 
&(\ref{8-06-d2}) 
&(\ref{8-09-l}) 
&(\ref{8-09-l1}) 
&(\ref{8-06-d3}) 
&(\ref{8-06-d4})
\\ \cline{2-9}
& $ e^{s I_{\mix}(W_E \circ \Gamma)} $
&(\ref{2-25-1-d})
&(\ref{2-25-1-c-2}) 
&(\ref{8-06-d2}) 
&(\ref{8-09-i}) 
&(\ref{8-09-i1}) 
& --
& --
\\ \hline
\multirow{3}{*}{Key inequality}
& $ I_{\mix}(W_E \circ \Gamma) $ 
&(\ref{9-19-1})
&\multicolumn{2}{|c|}{(\ref{9-19-1-a})}
&\multicolumn{2}{|c|}{(\ref{9-19-1-b})}
&(\ref{9-19-1-a})
&(\ref{9-19-1-b})
\\ \cline{2-9}
& $ d_{1,\mix}(W_E \circ \Gamma) $
& \multicolumn{7}{|c|}{(\ref{8-26-13-a}), (\ref{8-26-13-f2})}
\\ \cline{2-9}
& $ e^{s I_{\mix}(W_E \circ \Gamma)} $
&(\ref{9-20-1})
&\multicolumn{2}{|c|}{(\ref{9-20-1-a})}
&\multicolumn{2}{|c|}{(\ref{9-20-1-b})}
& --
& --
\\ \hline
\end{tabular}
\end{center}

\vspace{2ex}

Key inequality expresses the key inequality for deriving the evaluation.
-- means that there is no specialized discussion.
\Label{table1}
\end{table}

\section*{Acknowledgments}
The author is grateful to the reviewer for giving helpful comments for the first version.
He is partially supported by a MEXT Grant-in-Aid for Young Scientists (A) No. 20686026 and Grant-in-Aid for Scientific Research (A) No. 23246071.
He is partially supported by the National Institute of Information and Communication Technolgy (NICT), Japan.
The Centre for Quantum Technologies is funded by the
Singapore Ministry of Education and the National Research Foundation
as part of the Research Centres of Excellence programme. 

\bibliographystyle{IEEE}

\begin{thebibliography}{99}
\bibitem{Hayashi-book}
M. Hayashi,
Quantum Information: An Introduction, Springer (2006).

\bibitem{CK79} I. Csisz\'{a}r and J. K\"{o}rner,
``Broadcast channels with confidential messages,'' 
{\em IEEE Trans. Inform. Theory}, 
vol. {\bf 24}, No. 3, 339--348, 1978.

\bibitem{Carter}
L. Carter and M. Wegman,
``Universal classes of hash functions,''
{\em J. Comput. Sys. Sci.,}
vol. {\bf 18}, No. 2, 143--154, 1979.

\bibitem{Krawczyk}
H. Krawczyk. LFSR-based hashing and authentication. Advances in Cryptology --- CRYPTO '94. Lecture Notes in Computer Science, vol. 839, Springer-Verlag, pp 129--139, 1994.


\bibitem{Hayashi}
M. Hayashi, ``General non-asymptotic and asymptotic formulas in channel resolvability and identification capacity and its application to wire-tap channel,'' 
{\em IEEE Trans. Inform. Theory}, 
vol. {\bf 52}, No. 4, 1562--1575, 2006. 

\bibitem{Wyner}
A. D. Wyner, ``The wire-tap channel,'' {\em Bell. Sys. Tech. Jour.}, vol. 54, 1355--1387, 1975.

\bibitem{Deve}
I. Devetak,
``The private classical information capacity and quantum information 
capacity of a quantum channel,''
{\em IEEE Trans. Inform. Theory}, vol. {\bf 51}, No. 1, 44--55, 2005. 

\bibitem{WNI}
A. Winter, A. C. A. Nascimento, and
H. Imai,
``Commitment Capacity of Discrete Memoryless Channels,''
{\em Proc. 9th Cirencester Crypto and Coding Conf.}, LNCS 2989, pp 35-51, Springer, Berlin 2003;
cs.CR/0304014 (2003)

\bibitem{Gal}
R.\ G.\ Gallager, 
{\em Information Theory and Reliable Communication}, 
John Wiley \& Sons, 1968.

\bibitem{Ren05} 
R. Renner, ``Security of Quantum Key Distribution,'' 
PhD thesis, Dipl. Phys. ETH, Switzerland, 2005. 
arXiv:quantph/0512258.

\bibitem{Csiszar} 
I. Csisz\'{a}r, 
``Almost Independence and Secrecy Capacity,''
{\em Problems of Information Transmission}, vol. {\bf 32}, no.1, pp.40-47, 1996.

\bibitem{Tsuru}
T. Tsurumaru, M. Hayashi, ``Dual universality of hash functions and its applications to quantum cryptography,''
{\em IEEE Trans. Inform. Theory}, 
vol. {\bf 59}, No. 7, 4700-4717, (2013).

\bibitem{expo-chan}
M. Hayashi, ``Error exponent in asymmetric quantum hypothesis testing and its application to classical-quantum channel coding,'' 
{\em Physical Review A}, vol. {\bf 76}, 062301 (2007). 

\bibitem{DevP}
I. Devetak, P. W. Shor, 
``The capacity of a quantum channel for simultaneous transmission of classical and quantum information,''
{\em Comm. Math. Phys.}, {\bf 256}, 287 (2005).

\bibitem{H2001}
M. Hayashi,
``Optimal sequence of POVMs in the sense of Stein's lemma in quantum
hypothesis,''
{\em J. Phys. A: Math. and Gen.}, {\bf 35}, 10759-10773 (2002).

\bibitem{H-qkd}
M. Hayashi, ``Practical Evaluation of Security for Quantum Key Distribution," 
{\em Physical Review A}, vol. {\bf 74}, 022307 (2006).


\bibitem{H-qkd2}
M. Hayashi, ``Upper bounds of eavesdropper's performances in finite-length code with the decoy method,'' 
{\em Physical Review A}, vol. {\bf 76}, 012329 (2007); 
{\em Physical Review A}, vol. {\bf 79}, 019901(E) (2009).


\bibitem{H-cq}
M. Hayashi, 
``Large deviation analysis for quantum security 
via smoothing of R\'{e}nyi entropy of order 2,''
arXiv:1202.0322 (2012). 

\bibitem{H-tight}
M. Hayashi, 
``Tight exponential analysis of universally composable privacy amplification and its applications,''
arXiv:1010.1358 (2010). 
Accepted in {\em IEEE Trans. Inform. Theory}.

\bibitem{H-leaked}
M. Hayashi, 
``Exponential decreasing rate of leaked information in universal random privacy amplification,''
{\em IEEE Trans. Inform. Theory}, 
vol. {\bf 57}, No. 6, 3989-4001, (2011).

\bibitem{han-verdu}
T. S. Han and S. Verd\'{u}, 
``Approximation theory of output statistics,''
{\em IEEE Trans. Inform. Theory}, 
vol. {\bf 39}, no. 3, pp. 752-772, (1993).

\bibitem{SMC-non}
M. Hayashi and R. Matsumoto,
``Secure Multiplex Coding with Dependent and
Non-Uniform Multiple Messages''
arXiv:1202.1332 (2012). 

\bibitem{H-precise}
M. Hayashi, 
``Precise evaluation of leaked information with
universal$_2$ privacy amplification in the presence of
quantum attacker,''
{\em Proceedings of the IEEE International Symposium on Information Theory (ISIT 2012)}, Cambridge, MA, July 2012.
arXiv:1202.0601

\bibitem{Hay1}
M. Hayashi, ``Information Spectrum Approach to Second-Order Coding Rate in Channel Coding,'' 
{\em IEEE Trans. Inform. Theory}, 
vol. {\bf 55}, No.11, 4947 - 4966 (2009); 

\bibitem{Pol}
Y. Polyanskiy, H.V. Poor, and S. Verd\'{u}, 
``Channel coding rate in the finite blocklength regime,'' 
{\em IEEE Trans. Inform. Theory}, 
vol. {\bf 56}, 2307 - 2359 (2010).

\bibitem{Wat}
S. Watanabe
``Private and Quantum Capacities of More Capable and Less Noisy Quantum Channels,''
{\em Phys. Rev. A} {\bf 85}, 012326 (2012) 

\bibitem{korner80}
J.~K\"orner and A.~Sgarro, ``Universally attainable error exponents for
  broadcast channels with degraded message sets,'' \emph{IEEE Trans.\ Inform.\
  Theory}, vol.~26, no.~6, pp. 670--679, Nov. 1980.

\bibitem{Ogawa-Nagaoka}
T.~Ogawa and H.~Nagaoka,
``Strong converse and {Stein's} lemma in quantum hypothesis testing,''
{\em IEEE Trans. Inform. Theory}, {\bf 46}, 2428--2433 (2000).




\end{thebibliography}

\appendices

\section{Proof of Lemma \ref{l3}}\Label{sl3}
When $s \ge 0$,
in order to show (\ref{8-12-a}), it is enough to show that
\begin{align*}
\min_{\sigma_E} 
\sum_x P_X(x) \Tr W_x^{1+s} \sigma_E^{-s}
= ( \Tr  ( \sum_x P_X(x) W_x^{1+s} )^{\frac{1}{1+s}})^{1+s}.
\end{align*}
The reverse operator Holder inequality
\begin{align*}
\Tr XY \ge (\Tr X ^{1/(1+s)})^{1+s} (\Tr Y^{-1/s})^{-s}
\end{align*}
holds for two non-negative matrixes $X$ and $Y$.
Then,
\begin{align*}
& \Tr (\sum_x P_X(x) W_x^{1+s}) \sigma_E^{-s} 
\ge 
( \Tr  ( \sum_x P_X(x) W_x^{1+s} )^{\frac{1}{1+s}})^{1+s}
(\Tr \sigma_E^{-s\cdot -1/s})^{-s} 
= 
( \Tr  ( \sum_x P_X(x) W_x^{1+s} )^{\frac{1}{1+s}})^{1+s}.
\end{align*}
The equality holds when $\sigma_E= (\sum_x P_X(x) W_x^{1+s})^{\frac{1}{1+s}}/\Tr (\sum_x P_X(x) W_x^{1+s})^{\frac{1}{1+s}}$.

When $s < 0$,
in order to show (\ref{8-12-a}), it is enough to show that
\begin{align*}
\max_{\sigma_E} 
\sum_x P_X(x) \Tr W_x^{1+s} \sigma_E^{-s}
= ( \Tr  ( \sum_x P_X(x) W_x^{1+s} )^{\frac{1}{1+s}})^{1+s}.
\end{align*}
The operator Holder inequality
\begin{align*}
\Tr XY \ge (\Tr X ^{1/(1+s)})^{1+s} (\Tr Y^{-1/s})^{-s}
\end{align*}
holds for two non-negative matrixes $X$ and $Y$.
Similarly, we have
\begin{align*}
& \Tr (\sum_x P_X(x) W_x^{1+s}) \sigma_E^{-s} 
\le 
( \Tr  ( \sum_x P_X(x) W_x^{1+s} )^{\frac{1}{1+s}})^{1+s}.
\end{align*}
The equality holds when $\sigma_E= (\sum_x P_X(x) W_x^{1+s})^{\frac{1}{1+s}}/\Tr (\sum_x P_X(x) W_x^{1+s})^{\frac{1}{1+s}}$.

\section{Proof of (\ref{8-09-o})}\Label{a8-12-1}
Define 
$W_{y}:= \sum_{x \in \cX}\frac{1}{|\cX|} W_{x,y}$.
Then, 
\begin{align*}
& \| W_{y} -W \circ P_{\mix, \cX\times \cY }\|_1 \\
\le &
\| \sum_{x \in \cX}\frac{1}{|\cX|} W_{x,y} -W \circ P_{\mix, \cX\times \cY }\|_1 \\
\le &
\sum_{x} \frac{1}{|\cX|} \| W_{x,y} -W \circ P_{\mix, \cX\times \cY }\|_1 .
\end{align*}
Thus,
\begin{align*}
& d_1'(X;Y,E|\rho_{X,Y,E})
=
\sum_{x} P_{\mix,\cX}(x)
\sum_{y} P_{\mix,\cY}(y) \| W_{x,y} -W_{y}\|_1 \\
\le &
\sum_{x} P_{\mix,\cX}(x)
\sum_{y} P_{\mix,\cY}(y) \| W_{x,y} -W \circ P_{\mix, \cX\times \cY }\|_1 \\
& +
\sum_{x} P_{\mix,\cX}(x)
\sum_{y} P_{\mix,\cY}(y) \| W_{y} -W \circ P_{\mix, \cX\times \cY }\|_1 \\
=&
\sum_{x} P_{\mix,\cX}(x)
\sum_{y} P_{\mix,\cY}(y) \| W_{x,y} -W \circ P_{\mix, \cX\times \cY }\|_1 \\
&+
\sum_{y} P_{\mix,\cY}(y) \| W_{y} -W \circ P_{\mix, \cX\times \cY }\|_1 \\
\le &
2 \sum_{x} P_{\mix,\cX}(x)
\sum_{y} P_{\mix,\cY}(y) \| W_{x,y} -W \circ P_{\mix, \cX\times \cY }\|_1 \\
=& 2 d_1'(X,Y;E|\rho_{X,Y,E}).
\end{align*}

\section{Proof of (\ref{9-4-c1})}\Label{a9-4-c1}
For an element $x' \in \cX$, 
we denote the addition map $x \mapsto x+x'$ by $\Add_{x'}$.
For a distribution $P_X$ on $\cX$,
we consider the distribution $\Add_{x} \circ P_X$ on $\cX$.
The symmetry of the channel $W$, 
we have
\begin{align}
I(X;E|W\times P_X)
=
I(X;E|W\times (\Add_{x'} \circ P_X))\Label{9-4-c2}
\end{align}
for any $x' \in \cX$.
Using (\ref{8-04-a}), we also have
\begin{align}
& \sum_{x'\in \cX} \frac{1}{|\cX|}
I(X;E|W\times (\Add_{x'} \circ P_X))
\le
\sum_{x'\in \cX} \frac{1}{|\cX|}
I(X;E|W\times (\Add_{x'} \circ P_X)\|
W\circ P_{\mix,\cX}) \nonumber \\
=&
\sum_{x'\in \cX} \frac{1}{|\cX|}
\sum_{x\in \cX} P_X(x-x')
D(W_{x}\| W\circ P_{\mix,\cX}) 
=
\frac{1}{|\cX|}
\sum_{x\in \cX} 
D(W_{x}\| W\circ P_{\mix,\cX}) 
=I_{\mix}(W).\Label{9-4-c3}
\end{align}
Combining (\ref{9-4-c2}) and (\ref{9-4-c3}),
we obtain 
$I(X;E|W\times P_X) \le I_{\mix}(X;E|W)$,
which implies (\ref{9-4-c1}).

\section{Secret key generation with universal composability}\Label{s4-1}
We assume that Alice and Bob share a common classical random number $a \in \cA$,
and Eve has a quantum state $\rho_a \in {\cal H}_{E}$, which is correlated to the random number $a$. 
The task is to extract a common random number 
$f(a)$ from the random number $a \in \cA$, which is almost independent of 
Eve's quantum state.
Here, Alice and Bob are only allowed to apply the same function $f$ to the common random number $a \in \cA$.
Now, we focus on an ensemble of the functions $f_{\bZ}$ from 
$\cA$ to $\{1, \ldots, M\}$, where $\bZ$ denotes a random variable describing 
the stochastic behavior of the function $f$.
An ensemble of the functions $f_{\bZ}$ is called universal$_2$ 
when it satisfies the following condition\cite{Carter}:
\begin{condition}\Label{C1}
$\forall a_1 \neq \forall a_2\in \cA$,
the probability that $f_{\bZ}(a_1)=f_{\bZ}(a_2)$ is
at most $\frac{1}{M}$.
\end{condition}

For example, when $M$ is an arbitrary integer and 
the cardinality $|\cA|$ is an arbitrary multiple of $M$,
an ensemble $\{f_{\bZ}\}$ 
satisfying the above condition
is given in the following way.
First, we fix a function $f$ from $\cA$ to $\{1, \ldots, M\}$
such that the cardinality $|f^{-1}\{i\}|$ is $\frac{|\cA|}{M}$.
We randomly choose a permutation $g \in S_{\cA}$ on $\cA$ with the uniform distribution,
where $S_{\cA}$ denotes the set of permutation on $\cA$.
So, we can make a random function $\{f\circ g \}_{g \in S_{\cA}}$.
This ensemble satisfies Condition \ref{C1}.

\begin{lem}[{\cite[Lemma 33]{H-cq}}]\Label{L-15}
When 
an ensemble of the functions $f_{\bZ}$ from 
$\cA$ to $\{1, \ldots, M\}$
satisfies Condition \ref{C1},
any density matrix $\sigma_E$ on the system ${\cal H}_{E}$ satisfies 
\begin{align}
& \rE_{\bZ} d_1'(f_{\bZ}(A);E|\rho_{E,A}) \nonumber \\
\le &
(4+ \sqrt{v}) M^{s/2} e^{-\frac{s}{2}H_{1+s}(A|E| \rho_{E,A} \|\sigma_E )}
\Label{8-26-13-a} \\
& \rE_{\bZ} d_1'(f_{\bZ}(A);E|\rho_{E,A}) \nonumber \\
\le &
(4+ \sqrt{\lceil \lambda \rceil}) M^{s/2} 
e^{-\frac{s}{2}H_{1+s}(A|E| \rho_{E,A} \|\sigma_{E} )+\frac{s}{2}}\Label{8-26-13-f2},
\end{align}
where $v$ is the number of eigenvalues of $\sigma_E$
and $\lambda$ is defined as the real number $\log a_1-\log a_0$
by using the maximum eigenvalue $a_1$ and the minimum eigenvalue $a_0$ of $\sigma$.
\end{lem}


\section{Proof of Lemmas \ref{Lee1} and \ref{Lee3}}
\Label{a53}
First, we show Lemma \ref{Lee1}.
\begin{align}
& \rE_{\Phi}D( W \circ \Phi \circ P_A
\| W \circ {P_X})
=
\rE_{\Phi} 
\Tr
(\sum_{a} P_A(a) W_{\Phi(a)})
(\log (\sum_{a} P_A(a) W_{\Phi(a)}) -\log (W \circ {P_X}))\nonumber
 \\
=&
\rE_{\Phi} 
\Tr
(\sum_{a} P_A(a) W_{\Phi(a)})
(\log (\sum_{a'} P_A(a') W_{\Phi(a')})
-\log (W \circ {P_X}) )
 \nonumber\\
=&
\Tr
\sum_{a} \rE_{\Phi(a)} 
P_A(a) W_{\Phi(a)}
\rE_{\Phi|\Phi(a)} 
(\log (P_A(a) W_{\Phi(a)}+
\sum_{a'\neq a } P_A(a') W_{\Phi(a')})
-\log (W \circ {P_X}) ) \nonumber\\
\le &
\Tr
\sum_{a} 
\rE_{\Phi(a)} 
P_A(a) W_{\Phi(a)}
( \log (P_A(a) W_{\Phi(a)}+
\rE_{\Phi|\Phi(a)} 
\sum_{a'\neq a } P_A(a') W_{\Phi(a')})
-\log (W \circ {P_X}) ) \Label{1}\\
= &
\Tr
\sum_{a} 
\rE_{\Phi(a)} 
P_A(a) W_{\Phi(a)}
( \log (P_A(a) W_{\Phi(a)}+
\sum_{a'\neq a } P_A(a') W \circ {P_X})
-\log (W \circ {P_X}) ) \nonumber\\
\le &
\Tr
\sum_{a} 
\rE_{\Phi(a)} 
P_A(a) W_{\Phi(a)}
( \log (P_A(a) W_{\Phi(a)}+W \circ {P_X})
- \log (W \circ {P_X}) ) \Label{2} \\
\le &
\Tr
\sum_{a} 
\rE_{\Phi(a)} 
P_A(a) W_{\Phi(a)}
( \log (P_A(a) v {\cal E}_{W \circ {P_X}} (W_{\Phi(a)})+W \circ {P_X} )
-\log (W \circ {P_X}) ) \Label{3} \\
= &
\frac{1}{s}\Tr
\sum_{a} 
\rE_{\Phi(a)} 
P_A(a) W_{\Phi(a)}(y)
\log (I+ v P_A(a) {\cal E}_{W \circ {P_X}} (W_{\Phi(a)})  (W \circ {P_X})^{-1} )^s \nonumber\\
\le &
\frac{v^s}{s} 
\Tr
\sum_{a} 
\rE_{\Phi(a)} 
P_A(a)^{1+s} W_{\Phi(a)} 
({\cal E}_{W \circ {P_X}} (W_{\Phi(a)}))^{s} (W \circ {P_X})^{-s} \Label{ee4} \\
= &
\frac{v^s}{s} 
\sum_{a} 
\rE_{\Phi(a)} 
P_A(a)^{1+s}  \Tr ({\cal E}_{W \circ {P_X}} (W_{\Phi(a)}))^{1+s} (W \circ {P_X})^{-s} \nonumber \\
= &
\frac{v^s}{s} 
\sum_{a} 
P_A(a)^{1+s} 
\Tr
\sum_{x}
P_X(x) ({\cal E}_{W \circ {P_X}} (W_{x}))^{1+s} (W \circ {P_X})^{-s} 
=
\frac{v^s}{s} 
e^{-s H_{1+s}(A)} e^{s I_s (X;E|{\cal E}_{W \circ {P_X}}[W],P_X)} \nonumber \\
\le &
\frac{v^s}{s} 
e^{-s H_{1+s}(A)}e^{s I_s(X;E|W \times P_X)},\Label{1-30-1}
\end{align}
which implies (\ref{9-19-1}).
In the above derivation, 
(\ref{1}) follows from the concavity of $x \mapsto \log x$,
(\ref{2}) follows from $\sum_{a'\neq a } P_A(a') \le 1$,
(\ref{3}) follows from (\ref{8-15-23}),
(\ref{ee4}) follows from the following inequality,
and
(\ref{1-30-1}) follows from (\ref{8-15-7-b}).
The inequality $(x+y)^s \le x^s+y^s$ yields that
\begin{align}
\log (1+x)= \frac{1}{s}\log(1+x)^s \le
\frac{1}{s}\log(1^s+x^s) = \frac{x^s}{s}.\Label{ee5}
\end{align}
Next, we show (\ref{9-19-2}).
Since $s \mapsto \underline{D}_s^*(W \circ \Phi \circ P_A \| W \circ P_X)$ is monotone increasing, 
we obtain
$s \underline{D}( W \circ \Phi \circ P_A \| W \circ P_X)
\le
s \underline{D}_s^*(A;E| W \circ \Phi \circ P_A \| W \circ P_X)$.
Thus,
\begin{align}
& \rE_{\Phi} e^{s \underline{D}_s^*(W \circ \Phi \circ P_A \| W\circ P_X)}
=
\rE_{\Phi} 
\Tr
(\sum_{a} P_A(a) W_{\Phi(a)})
((W \circ {P_X})^{-1/2} (\sum_{a} P_A(a) W_{\Phi(a)}) (W \circ {P_X})^{-1/2})^s  \nonumber\\
=&
\Tr
\sum_{a} 
\rE_{\Phi(a)} 
P_A(a) W_{\Phi(a)}
\rE_{\Phi|\Phi(a)} 
((W \circ {P_X})^{-1/2}(P_A(a) W_{\Phi(a)}+
\sum_{a'\neq a } P_A(a') W_{\Phi(a')}) (W \circ {P_X})^{-1/2})^s \nonumber\\
\le &
\Tr
\sum_{a} 
\rE_{\Phi(a)} 
P_A(a) W_{\Phi(a)}
((W \circ {P_X})^{-1/2} (P_A(a) W_{\Phi(a)}+
\rE_{\Phi|\Phi(a)} 
\sum_{a'\neq a } P_A(a') W_{\Phi(a')})
(W \circ {P_X})^{-1/2})^s
 \Label{1-a}\\
= &
\Tr
\sum_{a} 
\rE_{\Phi(a)} 
P_A(a) W_{\Phi(a)}(y)
((W \circ {P_X})^{-1/2} (P_A(a) W_{\Phi(a)}+
\sum_{a'\neq a } P_A(a') W \circ {P_X})(W \circ {P_X})^{-1/2})^s\nonumber
\\
\le &
\Tr
\sum_{a} 
\rE_{\Phi(a)} 
P_A(a) W_{\Phi(a)}
((W \circ {P_X})^{-1/2}(P_A(a) W_{\Phi(a)}+W \circ {P_X})(W \circ {P_X})^{-1/2})^s
 \Label{2-a} \\
=&
\Tr
\sum_{a} 
\rE_{\Phi(a)} 
P_A(a) W_{\Phi(a)}
((W \circ {P_X})^{-1/2}P_A(a) W_{\Phi(a)}(W \circ {P_X})^{-1/2}+I )^s\nonumber
 \\
\le &
\Tr
\sum_{a} 
\rE_{\Phi(a)} 
P_A(a) W_{\Phi(a)}
(((W \circ {P_X})^{-1/2}P_A(a) W_{\Phi(a)}(W \circ {P_X})^{-1/2})^s +I )  \Label{3-a} \\
= &
1+ 
\sum_{a} 
\rE_{\Phi(a)} 
P_A(a)^{1+s} \Tr W_{\Phi(a)} ((W \circ {P_X})^{-1/2} W_{\Phi(a)}(W \circ {P_X})^{-1/2})^s
\nonumber  \\
= &
1+ 
\sum_{a} 
P_A(a)^{1+s} 
\sum_{x}
P_X(x) \Tr W_{x} ((W \circ {P_X})^{-1/2} W_{x} (W \circ {P_X})^{-1/2})^s
=
1+ 
(\sum_{a} P_A(a)^{1+s} )e^{\underline{I}_s^*(X;E|W \times P_X)}.\nonumber
\end{align}
In the above derivation, 
(\ref{1-a}) follows from the concavity of $x \mapsto x^s$,
(\ref{2-a}) follows from $\sum_{a'\neq a } P_A(a') \le 1$,
(\ref{3-a}) follows from the inequality $(x+y) \le x^s+y^s$.
Then, we obtain (\ref{9-19-2}).

Next, we show (\ref{9-19-1-b}) of Lemma \ref{Lee3} by modifying the proof of (\ref{9-19-1}).
We introduce the random variable $Z:=f_{\bZ|\tilde{X}}(a) =f_{\bZ}(a)+\tilde{X}$.
The random variable $Z$ is independent of the choice of $f_{\bZ}$.
Since $\rE_{\bZ|Z} W_{f_{\bZ|\tilde{X}}(a)}=W \circ {P_{\mix,{\cal X}}}$ for $a\in {\cal A}$,
the proof of (\ref{9-19-1}) can be applied to the proof of (\ref{9-19-1-b})
by replacing $\Phi(a)$, $\Phi|\Phi(a)$, and $P_X$ by $Z$, $\bZ|Z$ and $P_{\mix,\cX}$.
The proof of (\ref{9-19-2}) can be applied to the proof of (\ref{9-19-2-b}) with the same replacement.


\section{Proof of Lemma \ref{Lee2}}
Next, we show (\ref{9-19-1-a}) of Lemma \ref{Lee2} by modifying the proof of (\ref{9-19-1-b}).
In this case, for any element $x \in {\cal X}$,
\begin{align*}
\frac{P_{\bZ,\tilde{X}}( x \in C[\bZ]+\tilde{X}, C[\bZ]=C )}{P_{\bZ,\tilde{X}}(C[\bZ]=C)}
=\frac{\sL}{|{\cal X}|}
=P_{\bZ,\tilde{X}}( x \in C[\bZ]+\tilde{X} ).
\end{align*}
Thus,
\begin{align*}
P_{\bZ,\tilde{X}}( C[\bZ]=C | x \in C[\bZ]+\tilde{X})
=\frac{P_{\bZ,\tilde{X}}( x \in C[\bZ]+\tilde{X}, C[\bZ]=C )}{P_{\bZ,\tilde{X}}( x \in C[\bZ]+\tilde{X} )}
=P_{\bZ,\tilde{X}}(C[\bZ]=C).
\end{align*}
Further, 
for $x\neq x' \in {\cal X}$,
when $ x \in C[\bZ]+\tilde{X}$,
$x'\in C[\bZ]+\tilde{X}$ if and only if $x'-x \in C[\bZ]$.
Thus,
\begin{align*}
P_{\bZ,\tilde{X}}( x' \in C[\bZ]+\tilde{X} | x \in C[\bZ]+\tilde{X})
=P_{\bZ,\tilde{X}}( x'-x \in C[\bZ] | x \in C[\bZ]+\tilde{X})
=P_{\bZ,\tilde{X}}(x'-x \in C[\bZ]).
\end{align*}
Therefore, 
\begin{align}
& \rE_{\bZ,\tilde{X}}D( W_{P_{\mix,C[\bZ]+\tilde{X}}} \| W \circ {P_{\mix,{\cal X}}}) \nonumber \\
=&
\Tr
\sum_{x} 
\frac{1}{\sL}
P_{\bZ,\tilde{X}}(x \in C[\bZ]+\tilde{X})
 W_{x} 
\rE_{\bZ|x \in C[\bZ]+\tilde{X} } 
(\log (\frac{1}{\sL} W_{x}+\frac{1}{\sL}
\sum_{x'\neq x \in C[\bZ]+\tilde{X}}  W_{x'} )
-\log (W \circ {P_{\mix,{\cal X}}}) ) \nonumber \\
\le &
\Tr
\sum_{x} 
\frac{1}{\sL}
P_{\bZ,\tilde{X}}(x \in C[\bZ]+\tilde{X})
W_{x}
(\log (\frac{1}{\sL} W_{x}+\frac{1}{\sL}
\rE_{\bZ|x \in C[\bZ]+\tilde{X} } 
\sum_{x'\neq x \in C[\bZ]+\tilde{X}}  W_{x'} )
-\log (W \circ {P_{\mix,{\cal X}}}) ) \Label{1-b}\\
= &
\Tr
\sum_{x} 
\frac{1}{\sL}
P_{\bZ,\tilde{X}}(x \in C[\bZ]+\tilde{X})
W_{x}
(\log (\frac{1}{\sL} W_{x}+\frac{1}{\sL}
\rE_{\bZ} 
\sum_{x' \neq 0 \in C[\bZ]}  W_{x'+x} ) -\log (W \circ {P_{\mix,{\cal X}}}) ) \nonumber \\
= &
\Tr
\sum_{x} 
\frac{1}{\sL}
\frac{\sL}{|\cX|}
W_{x}
(\log (\frac{1}{\sL} W_{x}+\frac{1}{\sL}\frac{\sL}{|{\cal X}|}
\sum_{x' \neq 0\in {\cal X}}  W_{x'+x} ) -\log (W \circ {P_{\mix,{\cal X}}}) ) \nonumber \\
= &
\Tr
\sum_{x} 
\frac{1}{|\cX|}
W_{x}
(\log (\frac{1}{\sL} W_{x}+ W \circ {P_{\mix,{\cal X}}}) -\log (W \circ {P_{\mix,{\cal X}}}) ) \nonumber \\
\le &
\Tr
\sum_{x} 
\frac{1}{|\cX|}
W_{x}
(\log (\frac{v}{\sL} {\cal E}_{W \circ {P_{\mix,{\cal X}}}}(W_{x})+ W \circ {P_{\mix,{\cal X}}}) -\log (W \circ {P_{\mix,{\cal X}}}) ) \Label{3-b}\\
= &
\Tr
\sum_{x} 
\frac{1}{|\cX|}
W_{x}
\log (I+ \frac{v}{\sL} {\cal E}_{W \circ {P_{\mix,{\cal X}}}}(W_{x}) (W \circ {P_{\mix,{\cal X}}})^{-1}) \nonumber \\
= &
\Tr
\sum_{x} 
\frac{1}{s|\cX|}
W_{x}
\log (I+ \frac{v}{\sL} {\cal E}_{W \circ {P_{\mix,{\cal X}}}}(W_{x}) (W \circ {P_{\mix,{\cal X}}})^{-1})^s \nonumber \\
\le &
\sum_{x} 
\frac{1}{|\cX|}
\frac{v^s}{s \sL^s}
\Tr
W_{x} {\cal E}_{W \circ {P_{\mix,{\cal X}}}}(W_{x})^s
(W \circ {P_{\mix,{\cal X}}})^{-s}
\Label{4-b}
=\sum_{x} 
\frac{1}{|\cX|}
\frac{v^s}{s \sL^s}
\Tr
{\cal E}_{W \circ {P_{\mix,{\cal X}}}}(W_{x})^{1+s}
(W \circ {P_{\mix,{\cal X}}})^{-s}
\\
= &
\frac{v^s}{s \sL^s}
e^{I_s(X;E|{\cal E}_{W \circ {P_{\mix,{\cal X}}}}[W] \times P_{\mix,\cX})} 
\le 
\frac{v^s}{s \sL^s}
e^{I_s(X;E|W \times P_{\mix,\cX})} .\Label{1-30-2}
\end{align}
In the above derivation, 
(\ref{1-b}) follows from the concavity of $x \mapsto \log x$,
(\ref{3-b}) follows from (\ref{8-15-23}),
(\ref{4-b}) follows from (\ref{ee5}),
and
(\ref{1-30-2}) follows from (\ref{8-15-7-b}).
Then, we obtain (\ref{9-19-1-a}).

Next, we show (\ref{9-19-2-a}).
\begin{align}
& \rE_{\bZ,\tilde{X}} e^{\underline{D}_s^*(W_{P_{\mix,C[\bZ]+\tilde{X}}} \| W \circ {P_{\mix,{\cal X}}} )} \nonumber \\
=&
\Tr
\sum_{x} 
\frac{1}{\sL}
P_{\bZ,\tilde{X}}(x \in C[\bZ]+\tilde{X})
 W_{x} 
\rE_{\bZ|x \in C[\bZ]+\tilde{X} } 
((W \circ {P_{\mix,{\cal X}}})^{-1/2} (\frac{1}{\sL} W_{x}+\frac{1}{\sL} \sum_{x'\neq x \in C[\bZ]+\tilde{X}}  W_{x'}) 
(W \circ {P_{\mix,{\cal X}}})^{-1/2})^s  \nonumber\\
\le &
\Tr
\sum_{x} 
\frac{1}{\sL}
P_{\bZ,\tilde{X}}(x \in C[\bZ]+\tilde{X})
 W_{x} 
((W \circ {P_{\mix,{\cal X}}})^{-1/2} (\frac{1}{\sL} W_{x}+\frac{1}{\sL} \rE_{\bZ|x \in C[\bZ]+\tilde{X} } \sum_{x'\neq x \in C[\bZ]+\tilde{X}}  W_{x'}) 
(W \circ {P_{\mix,{\cal X}}})^{-1/2})^s  \Label{1-c}\\
= &
\Tr
\sum_{x} 
\frac{1}{|{\cal X}|}
 W_{x} 
((W \circ {P_{\mix,{\cal X}}})^{-1/2} (\frac{1}{\sL} W_{x}+\frac{1}{\sL} \rE_{\bZ} \sum_{x'- x \in C[\bZ],x'\neq x}  W_{x'}) 
(W \circ {P_{\mix,{\cal X}}})^{-1/2})^s \nonumber \\
\le &
\Tr
\sum_{x} 
\frac{1}{|{\cal X}|}
 W_{x} 
((W \circ {P_{\mix,{\cal X}}})^{-1/2} (\frac{1}{\sL} W_{x}+\frac{1}{|{\cal X}|} \sum_{x' \neq x \in {\cal X}}  W_{x'}) 
(W \circ {P_{\mix,{\cal X}}})^{-1/2})^s \nonumber\\
\le &
\Tr
\sum_{x} 
\frac{1}{|{\cal X}|}
 W_{x} 
((W \circ {P_{\mix,{\cal X}}})^{-1/2} (\frac{1}{\sL} W_{x}+W \circ {P_{\mix,{\cal X}}}) 
(W \circ {P_{\mix,{\cal X}}})^{-1/2})^s  \nonumber\\
= &
\Tr
\sum_{x} 
\frac{1}{|{\cal X}|}
 W_{x} 
(I +\frac{1}{\sL} (W \circ {P_{\mix,{\cal X}}})^{-1/2} W_{x}(W \circ {P_{\mix,{\cal X}}})^{-1/2} )^s  \nonumber \\
\le &
\Tr
\sum_{x} 
\frac{1}{|{\cal X}|}
 W_{x} 
(I +\frac{1}{\sL^s} ((W \circ {P_{\mix,{\cal X}}})^{-1/2} W_{x}(W \circ {P_{\mix,{\cal X}}})^{-1/2} )^s ) \Label{3-c}\\
= &
1+ \frac{1}{\sL^s}
\sum_{x} 
\frac{1}{|{\cal X}|}
\Tr  W_{x}  ((W \circ {P_{\mix,{\cal X}}})^{-1/2} W_{x}(W \circ {P_{\mix,{\cal X}}})^{-1/2} )^s  
= 
1+ \frac{1}{\sL^s} e^{\underline{I}_s^*(X;E|W \times P_{\mix,\cX})}.\nonumber
\end{align}
In the above derivation, 
(\ref{1-c}) follows from the concavity of $x \mapsto x^s$,
(\ref{3-c}) follows from $(x+y) \le x^s+y^s$.
Thus, similar to (\ref{9-19-2}), we obtain (\ref{9-19-2-a}).

\end{document}